\documentclass[11pt]{article}
\usepackage[sectionbib]{natbib}
\usepackage{array,epsfig,fancyheadings,rotating}
\usepackage{hyperref}
\usepackage{sectsty,secdot}
\sectionfont{\fontsize{12}{14pt plus.8pt minus .6pt}\selectfont}

\subsectionfont{\fontsize{12}{14pt plus.8pt minus .6pt}\selectfont}
\usepackage{amsmath,amssymb,amsfonts,amsthm}
\usepackage{multirow}
\usepackage{caption}
\usepackage{graphicx}
\usepackage{float}
\usepackage{geometry}
\usepackage{listings}
\usepackage{xcolor}
\usepackage{algorithmic}
\usepackage{algorithm}
\usepackage{comment}
\usepackage{enumitem}
\usepackage{subcaption}
\usepackage{authblk}
\newtheorem{theorem}{Theorem}
\newtheorem{lemma}{Lemma}

\newtheorem{proposition}{Proposition}
\theoremstyle{definition}

\newtheorem{assumption}{Assumption}
\def\n{\noindent}
\newcommand{\red}[1]{#1}
\newcommand{\blue}[1]{#1}
\newcommand\fro[1]{\| #1 \|_{\rm{F}}}
\newcommand\oneinf[1]{\| #1 \|_{1,\infty}}
\newcommand\twoinf[1]{\| #1 \|_{2,\infty}}
\newcommand\ltwo[1]{\| #1 \|_{\ell_2}}
\newcommand\lone[1]{\| #1 \|_{\ell_1}}
\newcommand\lzero[1]{\| #1 \|_{\ell_0}}
\newcommand\linf[1]{\| #1 \|_{\ell_{\infty}}}
\newcommand\op[1]{\| #1 \|}
\newcommand\nuc[1]{\| #1 \|_{*}}
\newcommand{\mat}[1]{\begin{bmatrix}#1 \\ \end{bmatrix}}
\newcommand{\inp}[2]{\langle #1,#2\rangle}
\renewcommand{\arraystretch}{1.05}

\def\calA{{\mathcal A}}
\def\calB{{\mathcal B}}
\def\calC{{\mathcal C}}
\def\calD{{\mathcal D}}
\def\calE{{\mathcal E}}
\def\calF{{\mathcal F}}
\def\calG{{\mathcal G}}
\def\calH{{\mathcal H}}
\def\calI{{\mathcal I}}
\def\calJ{{\mathcal J}}
\def\calK{{\mathcal K}}
\def\calL{{\mathcal L}}
\def\calM{{\mathcal M}}
\def\calN{{\mathcal N}}
\def\calO{{\mathcal O}}
\def\calP{{\mathcal P}}
\def\calQ{{\mathcal Q}}
\def\calR{{\mathcal R}}
\def\calS{{\mathcal S}}
\def\calT{{\mathcal T}}
\def\calU{{\mathcal U}}
\def\calV{{\mathcal V}}
\def\calW{{\mathcal W}}
\def\calX{{\mathcal X}}
\def\calY{{\mathcal Y}}
\def\calZ{{\mathcal Z}}

\def\bcalA{{\boldsymbol{\mathcal A}}}
\def\bcalB{{\boldsymbol{\mathcal B}}}
\def\bcalC{{\boldsymbol{\mathcal C}}}
\def\bcalD{{\boldsymbol{\mathcal D}}}
\def\bcalE{{\boldsymbol{\mathcal E}}}
\def\bcalF{{\boldsymbol{\mathcal F}}}
\def\bcalG{{\boldsymbol{\mathcal G}}}
\def\bcalH{{\boldsymbol{\mathcal H}}}
\def\bcalI{{\boldsymbol{\mathcal I}}}
\def\bcalJ{{\boldsymbol{\mathcal J}}}
\def\bcalK{{\boldsymbol{\mathcal K}}}
\def\bcalL{{\boldsymbol{\mathcal L}}}
\def\bcalM{{\boldsymbol{\mathcal M}}}
\def\bcalN{{\boldsymbol{\mathcal N}}}
\def\bcalO{{\boldsymbol{\mathcal O}}}
\def\bcalP{{\boldsymbol{\mathcal P}}}
\def\bcalQ{{\boldsymbol{\mathcal Q}}}
\def\bcalR{{\boldsymbol{\mathcal R}}}
\def\bcalS{{\boldsymbol{\mathcal S}}}
\def\bcalT{{\boldsymbol{\mathcal T}}}
\def\bcalU{{\boldsymbol{\mathcal U}}}
\def\bcalV{{\boldsymbol{\mathcal V}}}
\def\bcalW{{\boldsymbol{\mathcal W}}}
\def\bcalX{{\boldsymbol{\mathcal X}}}
\def\bcalY{{\boldsymbol{\mathcal Y}}}
\def\bcalZ{{\boldsymbol{\mathcal Z}}}

\def\AA{{\mathbb A}}
\def\BB{{\mathbb B}}
\def\CC{{\mathbb C}}
\def\DD{{\mathbb D}}
\def\EE{{\mathbb E}}
\def\FF{{\mathbb F}}
\def\GG{{\mathbb G}}
\def\HH{{\mathbb H}}
\def\II{{\mathbb I}}
\def\JJ{{\mathbb J}}
\def\KK{{\mathbb K}}
\def\LL{{\mathbb L}}
\def\MM{{\mathbb M}}
\def\NN{{\mathbb N}}
\def\OO{{\mathbb O}}
\def\PP{{\mathbb P}}
\def\QQ{{\mathbb Q}}
\def\RR{{\mathbb R}}
\def\SS{{\mathbb S}}
\def\TT{{\mathbb T}}
\def\UU{{\mathbb U}}
\def\VV{{\mathbb V}}
\def\WW{{\mathbb W}}
\def\XX{{\mathbb X}}
\def\YY{{\mathbb Y}}
\def\ZZ{{\mathbb Z}}

\def\a{{\boldsymbol a}}
\def\b{{\boldsymbol b}}
\def\c{{\boldsymbol c}}
\def\d{{\boldsymbol d}}
\def\e{{\boldsymbol e}}
\def\f{{\boldsymbol f}}
\def\g{{\boldsymbol g}}
\def\h{{\boldsymbol h}}
\def\i{{\boldsymbol i}}
\def\j{{\boldsymbol j}}
\def\k{{\boldsymbol k}}
\def\l{{\boldsymbol l}}
\def\m{{\boldsymbol m}}
\def\n{{\boldsymbol n}}
\def\o{{\boldsymbol o}}
\def\p{{\boldsymbol p}}
\def\q{{\boldsymbol q}}
\def\r{{\boldsymbol r}}
\def\s{{\boldsymbol s}}
\def\t{{\boldsymbol t}}
\def\u{{\boldsymbol u}}
\def\v{{\boldsymbol v}}
\def\w{{\boldsymbol w}}
\def\x{{\boldsymbol x}}
\def\y{{\boldsymbol y}}
\def\z{{\boldsymbol z}}

\def\A{{\boldsymbol A}}
\def\B{{\boldsymbol B}}
\def\C{{\boldsymbol C}}
\def\D{{\boldsymbol D}}
\def\E{{\boldsymbol E}}
\def\F{{\boldsymbol F}}
\def\G{{\boldsymbol G}}
\def\H{{\boldsymbol H}}
\def\I{{\boldsymbol I}}
\def\J{{\boldsymbol J}}
\def\K{{\boldsymbol K}}
\def\L{{\boldsymbol L}}
\def\M{{\boldsymbol M}}
\def\N{{\boldsymbol N}}
\def\O{{\boldsymbol O}}
\def\P{{\boldsymbol P}}
\def\Q{{\boldsymbol Q}}
\def\R{{\boldsymbol R}}
\def\S{{\boldsymbol S}}
\def\T{{\boldsymbol T}}
\def\U{{\boldsymbol U}}
\def\V{{\boldsymbol V}}
\def\W{{\boldsymbol W}}
\def\X{{\boldsymbol X}}
\def\Y{{\boldsymbol Y}}
\def\Z{{\boldsymbol Z}}

\def\rss{\textsf{RSS}}
\def\bic{\textsf{BIC}}
\def\svd{\textsf{SVD}}
\def\ma{\textsf{mat}}
\def\vec{\textsf{vec}}
\def\cov{\textsf{Cov}}
\def\rank{\textsf{rank}}
\def\bdiag{\textsf{bdiag}}
\def\blk{\textsf{blk}}

\def\bbeta{{\boldsymbol{\beta}}}
\def\bnu{{\boldsymbol{\nu}}}
\def\bgamma{{\boldsymbol{\gamma}}}
\def\bepsilon{{\boldsymbol{\epsilon}}}
\def\bSigma{{\boldsymbol{\Sigma}}}
\def\bGamma{\boldsymbol{\Gamma}}
\def\hat{\widehat}
\def\tilde{\widetilde}
\def\bdiag{\textsf{blkdiag}}
\def\ker{\textsf{Ker}}
\def\supp{\textsf{supp}}
\def\model{\textsf{LIAR}}
\def\blkstack{\textsf{blkstack}}
\def\vecstack{\textsf{vecstack}}
\def\var{\textsf{Var}}
\def\modelone{\textsf{SP-LIAR}}
\def\knb{J}
\newcommand{\smallsq}{\scalebox{0.5}{$\square$}}
\def\ks{J_{\smallsq}}
\def\ke{K_{\textsf{e}}}

\title{Local Interaction Autoregressive Model for High-Dimensional Time Series Data}
\author[1]{Jingyang Li}
\author[2]{Yang Chen}
\affil[1]{Fudan University}
\affil[2]{University of Michigan, Ann Arbor}
\date{}

\begin{document}
\hypersetup{
  pdftitle={Local Interaction Autoregressive Model for High-Dimensional Time Series Data},
  pdfauthor={Jingyang Li and Yang Chen}
}
\maketitle

\begin{abstract}
High-dimensional matrix and tensor time series often exhibit local dependency, where each entry interacts mainly with a small neighborhood. Accounting for local interactions in a prediction model can greatly reduce the dimensionality of the parameter space, leading to more efficient inference and more accurate predictions. We propose a Local Interaction Autoregressive (\model{}) framework 
% and study Separable \model{}, a variant with shared row and column components, for high-dimensional matrix/tensor time series forecasting problems. 
for high-dimensional
matrix and tensor time series forecasting problems, and study Separable \model{}, a matrix
variant with shared row and column components.
We derive a scalable parameter estimation algorithm via parallel least squares with a BIC-type neighborhood selector. Theoretically, we show consistency of neighborhood selection and derive error bounds for kernel and auto-covariance estimation. Numerical simulations show that the BIC selector recovers the true neighborhood with high success rates, the \model{} achieves small estimation errors, and the forecasts outperform matrix time-series baselines. In a Total Electron Content (TEC) application, the proposed method identifies localized spatio-temporal propagation patterns and achieves improved prediction compared with non-local time series prediction models. 
% In summary, exploiting locality provides a principled framework for modeling and a computationally efficient algorithm for prediction in high-dimensional time series.
\end{abstract}

\section{Introduction}
%\red{[Applications of multivariate time series; Literature]}
Multivariate time series analysis is a well-established field with a rich literature of research (see e.g., \cite{lutkepohl2005new, tsay2013multivariate, tsay2023matrix}).
Recently, there has been a growing interest in high-dimensional time series modeling. 
{Existing research in this area can be broadly categorized into several complementary directions. The vector autoregression (VAR) framework has been extensively studied with sparsity or banded assumptions to accommodate high dimensionality while retaining interpretability \citep{davis2016sparse, guo2016high, gao2019banded}.
Building on this, matrix autoregression (MAR) models leverage the row–column structure of matrix-valued series for more efficient parameterization and interpretation \citep{chen2021autoregressive, hsu2021matrix, sun2023matrix}.
Further extensions to tensor autoregressive settings capture higher-order interactions and preserve multi-way dependencies in array-valued data \citep{li2021multi}.
Meanwhile, dynamic factor models offer an alternative route to dimension reduction by representing high-dimensional processes through a few latent dynamic factors~\citep{lam2012factor, wang2019factor, chen2022factor, chen2023statistical, han2024cp}.}

In particular, matrix- and tensor-valued time series have become increasingly common in economics, geophysics, and environmental science nowadays~\citep{hsu2021matrix,chen2021autoregressive,sun2023matrix} with the availability of high-resolution data collected at grid points over an extended amount of time.
For instance, to analyze the dynamics of multiple economic factors across various countries simultaneously, data can be organized into a matrix form \citep{chen2021autoregressive}.
On the other hand, geophysical data are typically higher-dimensional.
For instance, the global total electron content (TEC) distribution quantifies the electron density within the Earth’s ionosphere along the vertical path between a radio transmitter and a ground-based receiver. Accurately predicting global TEC is essential for anticipating the effects of space weather on positioning, navigation, and timing (PNT) services. Each TEC map takes the form of a $181 \times 361$ matrix, organized on a spatial-temporal grid with a resolution of 1 degree in both latitude and longitude~\citep[see e.g.,][]{sun2022matrix,sun2023complete}. As shown in Figure~\ref{fig:TEC}, consecutive TEC maps evolve coherently over time, forming a high-dimensional matrix-valued time series.
%Similarly complex scientific data is found in AIA-HMI Solar Flare Imaging~\citep[see e.g.,][]{sun2023tensor,chen2024solar}. A solar flare is an intense and localized eruption of electromagnetic radiation in the Sun’s atmosphere. High-energy radiation emissions from solar flares can significantly impact Earth’s space weather and potentially interfere with radio communication. Each flare originates from a solar active region, a transient and localized area characterized by complex magnetic fields. \red{[Remove the HMI description]}The data for each event is structured as a tensor with dimensions $313 \times 671 \times 10$, which consists of 8 Atmospheric Imaging Assembly (AIA) channels and 2 Helioseismic and Magnetic Imager (HMI) channels (see Figure \ref{fig:AIA}) of size $313 \times 671$.
%\begin{figure}[h]
%	\centering
%	% First subfigure
%	\begin{subfigure}[b]{0.45\textwidth}
%		\centering
%		\includegraphics[width=\linewidth]{figure/TEC-eg1030.png}
%		\caption{Global total electronic content (TEC) distribution at one fixed time point.}
%		\label{fig:TEC}
%	\end{subfigure}
%	\hfill
%	% Second subfigure
%	\begin{subfigure}[b]{0.45\textwidth}
%		\centering
%		\includegraphics[width=\linewidth]{figure/stackedAIA.png}
%		\caption{Multi-channel images at one fixed time point for solar flare predictions. }
%		\label{fig:AIA}
%	\end{subfigure}
%	
%	\caption{Examples of complex scientific data structures. }
%    % Left: A 2D matrix representing global TEC. Right: A 3D tensor of multi-channel solar flare images.}
%	\label{fig:sidebyside}
%\end{figure}
\begin{figure}[t]
	\centering
	\includegraphics[width=0.95\linewidth]{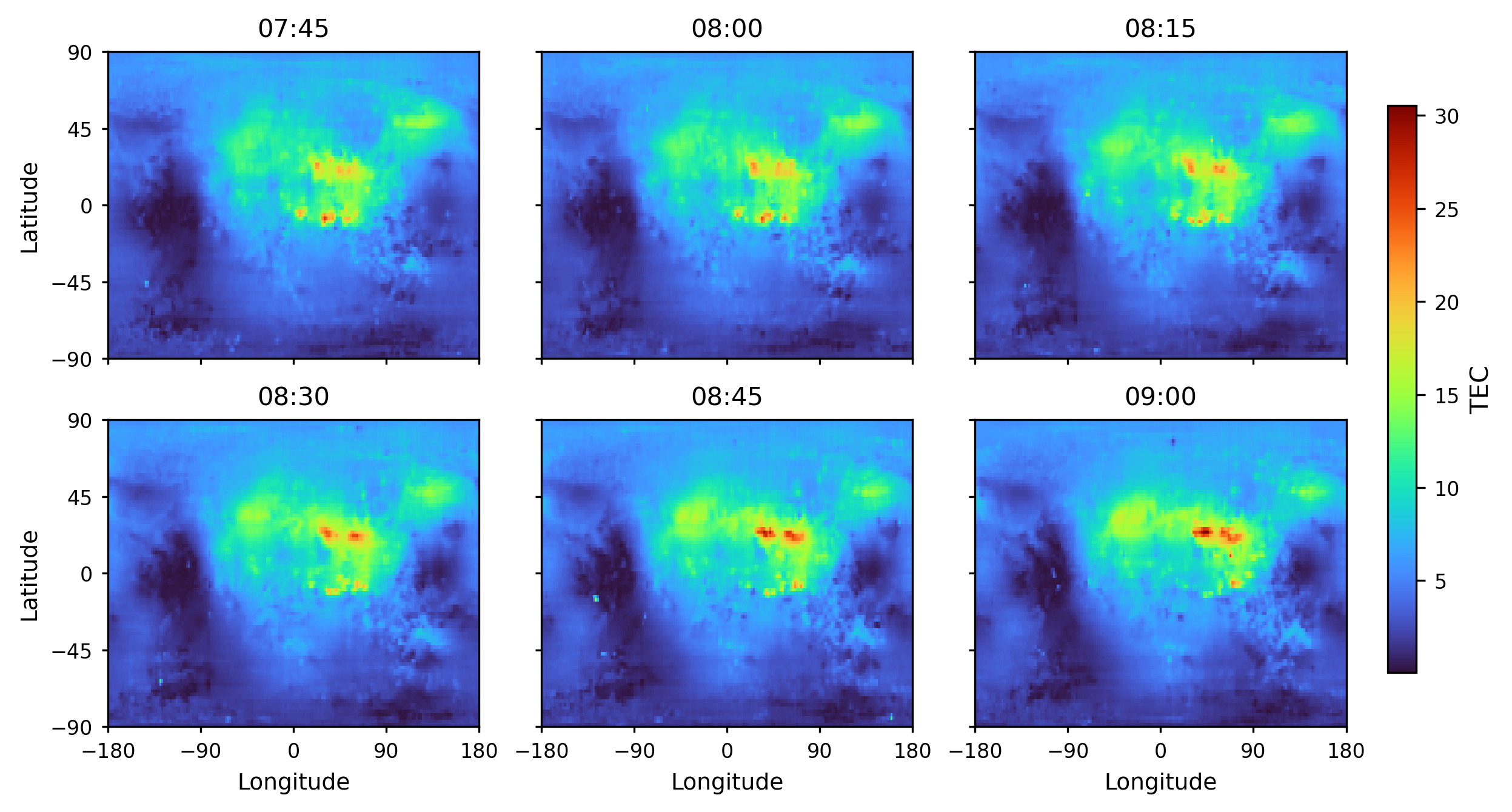}
	\caption{Consecutive global total electron content (TEC) maps from June 14, 2019, shown at 15-minute intervals.}
	\label{fig:TEC}
\end{figure}

While we can always resort to the well-studied vector time series analysis through vectorization, this would lead to a problem of a significantly larger  size and introduce many redundant parameters.
Recent studies, such as the factor autoregressive model \citep{chen2022factor, chen2023statistical}, have addressed this issue, but the resulting parameters often lack clear interpretability.
Additionally, the vectorization destroys the spatial information that is crucial in geophysics and environmental science. 
To address these challenges,  several methods based on autoregressive (AR) structures specifically designed for matrix/tensor time series have been developed.
For instance, \cite{chen2021autoregressive} proposed the following matrix autoregressive (MAR) model with a bilinear form: $\X_t = \A\X_{t-1}\B^\top + \E_t$,
where $\X_t$ is the $M\times N$ matrix. 
Bayesian method has also been proposed for variable selection in high-dimensional matrix autoregressive models \citep{celani2024bayesian}.
Subsequently, \cite{sun2023matrix} incorporated vector-valued time series data into the matrix autoregression model.  
However, tensor time series remains a less explored area. \cite{li2021multi} proposed a multi-linear form for tensor time series as a generalization of the matrix case. 
Despite these advancements, none of these approaches fully leverage the spatial dependency in the AR coefficient matrices.
In contrast, \cite{hsu2021matrix,jiang2024regularized} took spatial dependency into account by assuming $\A$ and $\B$ to be banded matrices, which significantly decreases the number of parameters.
Related rank-$R$ matrix autoregressive structures for spatio-temporal data were
further considered by \cite{hsu2024rank}. 

%\red{[Drawback of existing method; Challenges]}
Although MAR \citep{chen2021autoregressive} reduces the number of parameters compared with vectorization, and the MAR-ST (spatio-temporal) model \citep{hsu2021matrix} and banded MAR \citep{jiang2024regularized} further exploits spatial structures for matrix time series, computational cost remains a significant challenge for high-dimensional data. 
This high cost is largely because their algorithms are inherently iterative, a requirement that stems directly from the model assumptions being bilinear in the parameters.
% The structured MAR model employs an iterative least squares method to estimate the coefficients, necessitating the inversion of two large matrices in each iteration. Despite the model's use of spatial structures, the algorithm does not efficiently exploit these structures. 
Additionally, their model's underlying constraints can be overly stringent. Specifically, the MAR model can be understood as performing $MN$ low-rank matrix regressions, necessitating that the coefficient matrices have a rank of 1 and share identical column subspaces across rows and identical row subspaces across columns. While this structure is manageable for the MAR model due to its adequate number of parameters, it becomes excessively restrictive when banded structures are simultaneously imposed.

%\red{[Our proposed model]}
In this paper, we propose a general framework for matrix- and tensor-variate autoregressive modeling that incorporates spatial structure. Our framework is flexible and covers MAR \citep{chen2021autoregressive}, MAR-ST \citep{hsu2021matrix}, banded MAR \citep{jiang2024regularized}, and tensor AR \citep{li2021multi} as special cases. 
To incorporate the spatial structure, we specifically focus on the Local Interaction Autoregressive (\model) model, which applies to both matrix and tensor data. 
Our model enjoys a locally adaptive dependency property in that it allows different entries to depend on different sizes of neighborhoods. 
In practice, the size of the neighborhood is unknown. We propose a Bayesian information criterion for the selection of the neighborhood. 
We show that this criterion leads to consistent neighborhood selection under various asymptotic regimes.
When the neighborhoods are given, we propose parallel least squares (LS) estimators for the coefficients. We also provide asymptotic properties of the parallel LS estimator. 
%To further reduce the model's parameters, while also imposing less stringent constraints than the MAR-ST and banded MAR models mentioned previously, we consider a separable structure on the coefficients, which we call Separable-\model{} (\modelone{}).
To further reduce the number of parameters, we also consider a separable structure
on the local coefficients, following related low-rank and rank-(R) autoregressive
coefficient structures in the MAR and tensor autoregressive literature
\citep{hsu2024rank,li2021multi}. We refer to this variant as
Separable-\model{} (\modelone{}).

\subsection{Our Contributions}
Our contributions are summarized as follows. First, we introduce a unified framework for matrix and tensor time series that accommodates high dimensionality and multi-way structure in a coherent setup. It subsumes much of the existing literature as special cases. Specifically, we propose a new model that relaxes the structural constraints of MAR-ST and banded MAR while preserving the same order of parameters. Second, we design algorithms that are both simple to implement and suitable for parallel execution. Our algorithms provide closed-form estimators, eliminating the need for iterative processes. This greatly reduces the computational complexity (see Table \ref{table:computation} for the comparison with existing literature). Third, we establish asymptotic guarantees for our estimators without imposing structural assumptions on the noise, and identifiability poses no obstacle for \model{} or \modelone{} because inference is conducted at the equivalence-class level rather than on individual components. Moreover, our neighborhood-selection criterion is consistent under mild conditions across multiple regimes, including jointly growing sample size, ambient dimension, and neighborhood size. Finally, on synthetic data, our criterion selects the true neighborhood with high success rates, kernel and auto-covariance/auto-correlation estimates are accurate, and compared with matrix time-series baselines our method achieves the lowest prediction RMSE with much shorter runtime. On large-scale total electron content (TEC) data, our approach attains competitive predictive accuracy with markedly faster runtime than existing methods, and the learned neighborhoods exhibit interpretable spatial patterns. Additional tensor-based TEC analysis and experimental results are provided in Appendix~\ref{sec:supp-tensor}.

		\begin{table}[t]
		\centering
		\footnotesize	
		\begin{tabular}{|c|c|c|c|c|}
			\hline
			Model & Algorithm & Parallel? &Parameters &  Computation \\ \hline\hline
			% \multirow{2}{*}{ \centering MAR \citep{chen2021autoregressive}} 
            MAR& Projection Method  & No  & $O(M^2)$  & $O(M^4T + M^6)$ \\ 
            \cline{2-5}
			\citep{chen2021autoregressive}&Iterative LS & No & $O(M^2)$& $O(M^3T\cdot N_{\textsf{iter}})$\\ 
            \hline
			{\parbox{3cm}{ \centering MAR-ST \\ \citep{hsu2021matrix}}}& Iterative LS& No & $O(M)$ & $O\big((M^3+M^2T)\cdot N_{\textsf{iter}}\big)$ \\ 
            \hline
			\model ~(This paper) & Parallel LS & Yes & $O(M^2)$& $O(M^2 T)$ \\ 
            \hline
			\modelone ~(This paper) & Parallel LS & Yes & $O(M)$ & $O(M^2 T  +M^3)$ \\ \hline
		\end{tabular}
		\caption{Comparison between our proposed methods and existing methods for matrix autoregressive model. The matrix data is of size $M\times M$, and there are in total $T$ samples. We assume the size of neighborhood considered in  MAR-ST and our model to be constant. The computational cost displayed for MAR-ST is the minimized one when the innovation is also structured. }
		\label{table:computation}
		\end{table}
	
	%	We also study the asymptotic property of the proposed algorithms. Our theory also applies to the structured MAR. Moreover, our analysis requires no assumption on the structure of the error. 
	%	For \model, there is no identifiability issue, and for the \modelone, we avoid the identifiability issue by analyzing the equivalent class. 
	%	We derive the asymptotic normality for the coefficients and we also consider the estimation of the auto-correlation. 
	%	We also show our proposed criterion can consistently select the correct neighborhood under mild conditions.
	%	Moreover, our theory is valid under three different regimes: we allow both the sample size, the ambient dimension of the data, and the size of the neighborhood to go to infinity. 

\subsection{Organization of the Paper}
The rest of the paper is organized as follows. Section 2 presents the local interaction autoregressive framework, the separable SP-LIAR variant, the parallel least-squares estimator, the neighborhood-selection criterion, and a brief extension to tensor time series. Section 3 studies the asymptotic properties of the proposed estimators and establishes neighborhood-selection consistency. Sections 4 and 5 present numerical experiments and the TEC real-data application, respectively. Further tensor implementation and tensor-based TEC experiments are provided in Appendix~\ref{sec:supp-tensor}.

\section{Methodology}\label{sec:mar}
% \subsection{Notations}
Bold lower-case (e.g., $\z$) denote vectors, bold capitals (e.g., $\A,\B$) denote matrices.
All vectorization operations, denoted by $\text{vec}(\cdot)$, and any implicit index orderings strictly adhere to the column-major convention.
We index matrix location $\calS = [M]\times[N]$, and write an index $\i = (i_1,i_2)\in\calS$.  
Matrix entries are explicitly denoted with square brackets, as $[\A]_{\i}$.
%norms
Denote $\fro{\cdot}$ the Frobenius norm of matrices, and denote $\|\cdot\|_{\ell_p}$ the $\ell_p$-norm of vectors or vectorized tensors for $0\leq p\leq \infty$. Here $\linf{\v}$ denotes the largest magnitude of the entries of $\v$.
The spectral radius of a matrix $\A\in\RR^{M\times M}$ (denoted by $\rho(\A)$) is defined to be the maximum modulus of eigenvalues. The operator norm of a matrix $\A\in\RR^{M\times N}$ (denoted by $\op{\A}$) is defined to be the largest singular value.  
We denote $a\vee b = \max\{a,b\}, a\wedge b = \min\{a,b\}$.

\subsection{Local Interaction Autoregressive (\model{}) Model}
Let $\{\X_t\}_{t\ge 0}\subset\RR^{\calS}$ be a matrix-valued time series.
For each site $\i=(i_1,i_2)\in\calS$ and lag $p\in[P]$, let $\calJ_{\i}\subset\calS$ be a (location-specific) spatial neighborhood, and let $\M_{p,\i}\in\RR^{\calS}$ satisfy $\supp(\M_{p,\i})=\calJ_{\i}$.
The \model{} specification assumes that each entry of $\X_t$ depends only on past values within its neighborhood:
\begin{align}\label{LIAR:lagP} 
	\X_t
	=
	\sum_{\i\in\calS}\sum_{p\in[P]}
	\inp{\X_{t-p}}{\M_{p,\i}}\ \e_{i_1}\e_{i_2}^\top
	+\E_t,
\end{align}
where $\E_t$ is a mean-zero innovation independent of the past series, and
$\e_{i_1}\in\RR^{M}$ and $\e_{i_2}\in\RR^{N}$ are standard basis vectors.
Both the \emph{shape} and \emph{size} of $\calJ_{\i}$ may vary with $\i$; i.e., neighborhoods need not be rectangular.
Our goal is to recover the neighborhoods $\{\calJ_{\i}\}_{\i\in\calS}$ and estimate the coefficients $\{\M_{p,\i}\}_{\i\in\calS, p\in[P]}$.

In particular, this formulation admits a VAR representation and includes several familiar special cases:
\begin{enumerate}[leftmargin=1.1em, itemsep = -1pt, topsep=2pt]
	
	\item \textit{VAR representation.}
	Denote $\x_t=\vec(\X_t)$ and $\bepsilon_t=\vec(\E_t)$. Flattening both sides of \eqref{LIAR:lagP} yields
	$$
	\x_t \;=\; \sum_{p\in[P]} \M_p^{\vec}\,\x_{t-p} \;+\; \bepsilon_t,
	$$
	where $\M_p^{\vec}$ has the structured form
	$
	\M_p^{\vec} \;=\; \sum_{\i\in\calS} \big(\e_{i_2}\otimes \e_{i_1}\big)\vec(\M_{p,\i})^\top.
	$
	
	\item \textit{Banded VAR \citep{guo2016high}.}
	If $N=1$ and $\calJ_{i}=\{u:\,|u-i|\le K\}$, the model reduces to a banded VAR on a length-$M$ vector with bandwidth $K$.
	
	\item \textit{MAR \citep{chen2021autoregressive}.}
	Assume each local coefficient is rank one with shared components: $\M_{p,\i}=\a_{p,i_1}\,\b_{p,i_2}^\top$.
	Let $\e_{i_1}^\top\A_p=\a_{p,i_1}^\top$ and $\e_{i_2}^\top\B_p=\b_{p,i_2}^\top$ for $\i = (i_1,i_2)\in\calS$.
	And $\calJ_{\i}=\calS$ for all $\i$, then
	$$
	\X_t \;=\; \sum_{p\in[P]} \A_p\,\X_{t-p}\,\B_p^\top \;+\; \E_t,
	$$
	which is the matrix autoregressive (MAR) model.
	
	\item \textit{MAR-ST \citep{hsu2021matrix} or banded MAR \citep{jiang2024regularized}.}
	Under the same rank-one form $\M_{p,\i}=\a_{p,i_1}\,\b_{p,i_2}^\top$ with location-specific rectangular neighborhoods
	$
	\calJ_{\i} \;=\; \{(u_1,u_2):\ |u_1-i_1|\le K_{1,i_1},\ |u_2-i_2|\le K_{2,i_2}\},
	$
	the representation above holds with $\A_p$ and $\B_p$ banded, yielding the structured MAR.

	\item \textit{Separable LIAR (\modelone{}).}
%	To balance number of parameter and expressive power, we impose a low-rank separable structure on each local coefficient. 
	To further reduce the number of parameters, we consider a separable structure on the local coefficients. Related low-rank and rank-$R$ autoregressive coefficient structures have been studied in the MAR literature
	\citep{hsu2024rank}.
	Specifically, let
	$\M_{p,\i} = \sum_{r\in[R]}\a_{p,r,i_1}\b_{p,r,i_2}^\top$, with $[\a_{p,r,i_1}]_{u_1} = 0$ if $|i_1-u_1|>K_{1,i_1}$, and $[\b_{p,r,i_2}]_{u_2} = 0$ if $|i_2-u_2|>K_{2,i_2}$, and $R \leq \min_{ij}\{K_{1,i_1},K_{2,i_2}\}$.
	Then the model can be written as 
	\begin{align}\label{banded:rank-r} 
		\X_t = \sum_{p\in[P]}\sum_{r\in[R]}\A_{p,r}\X_{t-p}\B_{p,r}^\top+\E_t, 
	\end{align} 
	where $\e_{i_1}^\top\A_{p,r} = \a_{p,r,i_1}^\top$, $\e_{i_2}^\top\B_{p,r} = \b_{p,r,i_2}^\top$ are banded matrices. Here $R$ governs the rank of local interactions: $R=1$ recovers MAR-ST or banded MAR, and \modelone{} approaches the unrestricted \model{} as $R$ increases.
\end{enumerate}
In small-scale settings (e.g., macroeconomic panels), global MAR can be adequate \citep{chen2021autoregressive}, but in large-scale spatio–temporal systems it becomes computationally and statistically burdensome due to dense, high-dimensional kernels. MAR-ST \citep{hsu2021matrix} and banded MAR \citep{jiang2024regularized} mitigate this cost by imposing banded structure for matrices. But this rank-1 separable form restriction limits expressiveness and cannot capture location-specific variation.

By contrast, \model{} introduces local dependence through neighborhoods $\{\calJ_{\i}\}_{\i\in\calS}$ that are allowed to vary in shape and size across locations, which captures spatially varying (location-specific) dependence while dramatically reducing the number of parameters compared with MAR. As an intermediate option, \modelone{} imposes a low-rank separable structure within each neighborhood, offering a tunable bridge between structured MAR and the fully flexible \model{}: smaller rank yields more parsimonious models, whereas increasing the rank brings \modelone{} closer to unrestricted \model{} while preserving locality and interpretability.

\subsection{Estimating Coefficients}\label{sec:entrywise-estimation}
We now present estimation methods for \model{}, and we then extend the procedure to the separable \modelone{} variant.
For clarity of exposition, we will focus on $P=1$ in the main text, deferring the treatment of the general lag $P$ case to Section~\ref{sec:lagP} of the Appendix. We drop the subscript $p$ in \eqref{LIAR:lagP} and consider the following lag-1 model:
\begin{align}\label{LIAR:lag1}
\X_t = \sum_{\i\in\calS}
	\inp{\X_{t-1}}{\M_{\i}}\ \e_{i_1}\e_{i_2}^\top
	+\E_t.
\end{align}

For each $\i\in\calS$, we denote $\x_t^{(\i)} := \X_t(\calJ_{\i}), \m_{\i} = \M_{\i}(\calJ_{\i})\in\RR^{|\calJ_{\i}|}$.
Here $\M(\calJ)\in\RR^{|\calJ|}$ is the vector that collects the entries $\{[\M]_{\i}:\i\in\calJ\}$ in column-major order for $\calJ\subset\calS$. 
% Define the stacked coefficient and per-time design vectors
% \begin{align*}
% 	\m_{\i} := \mat{\m_1^{(\i)}\\\vdots\\ \m_P^{(\i)}}, \quad \r_t^{(\i)} := \mat{\x_{t-1}^{(\i)}\\ \vdots\\ \x_{t-P}^{(\i)}}\in\RR^{P|\calJ_{\i}|}
% \end{align*}
Then the scalar regression for entry $\i$ reads, for $t=2,\dots,T$,
	$[\X_t]_{\i} = \inp{\x_{t-1}^{(\i)}}{\m_{\i}} + [\E_t]_{\i}. $
Stacking over $t$ gives the standard linear model
$
\z_{\i} = \Y_{\i}\m_{\i}+\bnu_{\i},
$
where 
\begin{align}\label{defz}
	\z_{\i}
	&= \mat{[\X_2]_{\i}& \cdots & [\X_T]_{\i}}^\top, \notag\\
	\bnu_{\i}
	&= \mat{[\E_{2}]_{\i} & \cdots & [\E_T]_{\i}}^\top
	\in\RR^{T-1}, \notag\\
	\Y_{\i}
	&= \mat{\x_{1}^{(\i)}& \cdots& \x_{T-1}^{(\i)}}^\top
	\in\RR^{(T-1) \times |\calJ_{\i}|}.
\end{align}
% \begin{align}\label{defz}
% 	&\z_{\i} = \mat{[\X_2]_{\i}& \cdots & [\X_T]_{\i}}^\top,\quad 
%     \bnu_{\i}= \mat{[\E_{2}]_{\i} & \cdots & [\E_T]_{\i}}\in\RR^{T-1},\\	
% 	&\Y_{\i} = \mat{\x_{1}^{(\i)}& \cdots& \x_{T-1}^{(\i)}}^\top\in\RR^{(T-1) \times |\calJ_{\i}|}.
% \end{align}
The least squares estimator is
$
\hat\m_{\i}= \arg\min_{\m} \frac{1}{2}\ltwo{\z_{\i} - \Y_{\i}\m}^2= (\Y_{\i}^\top\Y_{\i})^{-1}\Y_{\i}^\top\z_{\i}.
$
This procedure can be carried out in parallel for different $\i\in\calS$.
We summarize the procedure in Algorithm \ref{alg:pls}. 
\begin{algorithm}[H]
	\caption{Parallel Least Squares}
	\begin{algorithmic}\label{alg:pls}
		\STATE{\textbf{Input: }Data $\{\X_t\}_{t=1}^T$, local neighborhood $\{\calJ_{\i}\}_{\i\in\calS}$}
		\STATE{\textbf{ParFor} $\i\in\calS$}
		\STATE{\quad Collect the covariate $\Y_{\i} $ and response $\z_{\i}$ defined in \eqref{defz}}
		\STATE{\quad Solve the least-squares problem $\hat\m_{\i}= (\Y_{\i}^\top\Y_{\i})^{-1}\Y_{\i}^\top\z_{\i}$}
		\STATE{\textbf{End ParFor}}
		\STATE{\textbf{Output: }$\{\hat\m_{\i}\}_{\i\in\calS}$}
	\end{algorithmic}
\end{algorithm}

\paragraph*{Efficiency and Computational Cost Analysis of Parallel LS. }
Due to the spatial locality, our parallel LS is highly efficient. Additionally, the estimator is non-iterative with a closed-form solution.
Let $\knb = \max_{\i}|\calJ_{\i}|$. Then for each sub-problem, we only need to solve a least squares problem, resulting in a computational cost of order $O(\knb^2T + \knb^3)$. We note that the computation of coefficients for a single entry is independent of the matrix data's ambient dimension. The total computational cost is $O\big(MN(\knb^2T + \knb^3)\big)$ .
In most applications, $\knb$ is small constant, so the computational cost reduces to $O(MNT)$. Moreover, the parallel LS can be implemented concurrently, further decreasing the algorithm’s runtime.
In contrast, while the spatial structure was utilized in \cite{hsu2021matrix}, their iterative least squares algorithm does not fully exploit it, leading to a significantly higher computational cost (see Table \ref{table:computation}).

\paragraph*{Coefficient Estimation for \modelone{}. }
The estimation for the coefficients for \modelone~ is more challenging as it admits no closed-form solution due to the separable structure. 
We proceed in two stages: first obtain per-location LS estimates as in \model{}, then enforce separability via a low-rank projection.
In \modelone{}, $\calJ_{\i}$ are rectangles by construction, and we denote $\M^{[\i]}:=\M_{\i}[\calJ_{\i}]$ as a sub-matrix of $\M_{\i}$. 
Arrange these matrices into the block matrix
\begin{align}\label{Mblk}
	\M^{\blk} = \mat{\M^{[1,1]}&\cdots&\M^{[1,N]}\\ \vdots&&\vdots \\ \M^{[M,1]}&\cdots &\M^{[M,N]}}. 
\end{align}
Then under \modelone{}, $\M^{\blk}$ has rank $R$. 
This motivates us to consider the best rank $R$ approximation of the following estimator of $\M^{\blk}$, where $\hat\m_{\i}$ are obtained from Algorithm \ref{alg:pls}:
\begin{align*}
	\hat\M^{\blk} = \mat{\textsf{mat}(\hat\m_{(1,1)})&\cdots&\textsf{mat}(\hat\m_{(1,N)})\\ \vdots&&\vdots \\ \textsf{mat}(\hat\m_{(M,1)})&\cdots &\textsf{mat}(\hat\m_{(M,N)})}. 
\end{align*}
Our algorithm is summarized in Algorithm \ref{alg:spliar}.
\begin{algorithm}[H]
	\caption{Parallel Least Square for \modelone}
	\begin{algorithmic}\label{alg:spliar}
		\STATE{\textbf{Input: }Local neighborhood $\calJ_{\i}$, data $\{\X_t\}_{t=1}^T$, rank $R$}
		\STATE{Run Algorithm \ref{alg:pls} for $\{\hat\m_{\i}\}_{\i\in\calS}$}
		\STATE{Perform rank $R$ SVD approximation on $\hat\M^{\blk}$ and get:
			$\hat\M_{R}^{\blk}=\svd_R(\hat\M^{\blk})$}
		%			 $\hat\M_{p,R}^{\blk}=\sum_{r=1}^R\hat\u_{p,r}\hat\v_{p,r}^\top$}
	\STATE{\textbf{Output: }$\{\hat\M_{R}^{[\i]}\}_{\i\in\calS}$, where $\hat\M_{R}^{[\i]}$ is the $(i_1,i_2)$-th block of $\hat\M_{R}^{\blk}$}
\end{algorithmic}
\end{algorithm}

\paragraph*{Identifiability Issue. }
In the case where $R=1$ \citep{chen2021autoregressive}, it is assumed that one of the coefficient matrices has a unit Frobenius norm. They also mentioned the identifiability issue becomes subtler when $R>1$. In our formulation of \modelone{}, this concern disappears as 
the primary quantity of interest is the equivalence class rather than two individual components.

\subsection{Neighborhood Selection}
We now discuss how to select the neighborhood $\calJ_{\i}$ for \model{}. 
Ideally, one would hope to select from all subsets of $\calS$.
However, due to the \textit{non-nested} model structure, this is difficult without refining the candidate class.
We illustrate this with a simple example.

\noindent\textbf{Example (Non-nested ambiguity).}
Consider the linear regression problem
$y_t = \inp{\X_t}{\M}$, $t = 1,\cdots, T,$ with $\M$ supported on $\calJ= \{(u,v):|u-i_0|\leq K_1, |v-j_0|\leq K_2\}$. 
Consider an alternative candidate $\tilde\calJ =  \{(u,v):|u-i_0|\leq k_1, |v-j_0|\leq k_2\}$ for some $k_1<K_{1}, k_2>K_{2}$ (so neither rectangle contains the other). 
There are natural designs $\X_t$ under which one can construct $\hat\M$ supported on $\tilde\calJ$ that reproduces exactly the same responses $y_t$ as $\M$. A formal statement and proof are provided in Section~\ref{sec:non-nested} of the Appendix.

Without nesting, empirical criteria cannot reliably separate such candidates. We therefore restrict the search to a nested sequence indexed by $k$ for each location $\i = (i_1,i_2)$:
$$\calJ_{\i}[1]\subsetneq\cdots\subsetneq\calJ_{\i}[k_0]\subsetneq\cdots\subsetneq\calJ_{\i}[K_0]\subset\calS,$$
where $\calJ_{\i}[k_0] = \calJ_{\i}$, and $K_0$ is a prescribed cap. 
Write $s_k = |\calJ_{\i}[k]|$ (depends on $k$, and implicitly on $\i$).
For any level $k$, let $\x_t[k]:= \X_t\big(\calJ_{\i}[k]\big)\in\RR^{s_k}$. Then we form the design matrix 
% $\Y[k]\in\RR^{(T-1)\times s_k}$:
$
\Y[k]= \mat{\x_{1}[k] &\cdots &\x_{T-1}[k]}^\top.
$
Also recall $\z = \mat{[\X_{2}]_{\i}& \cdots & [\X_T]_{\i}}^\top$. 
The residual sum of squares is
$
\rss_{\i}(k) := \ltwo{\z - \Y[k](\Y[k]^\top\Y[k])^{-1}\Y[k]^\top\z}^2.
$
And we use the Bayesian information criterion for the bandwidth selection: 
\begin{align}\label{BIC-entrywise}
\bic_{\i}(k)
=
\log \rss_{\i}(k)
+
D_0\,\frac{|\calJ_{\i}[k]|}{T}\,\log(M\vee N\vee T).
\end{align}
For each $\i\in\calS$, we select $\hat k_{\i}$ by
$\hat k_{\i}=\arg\min_{k\le K_0}\ \bic_{\i}(k)$.

\paragraph*{Choice of $D_0$ and $K_0$.}
We set $D_0=\log\log T$ when considering the regime where $T,M,N\to\infty$ while
$\knb=\max_{\i\in\calS}|\calJ_{\i}|$ remains fixed. The cutoff $K_0$ can be chosen by
examining the curvature of $\bic_{\i}(k)$. When $\knb$ also diverges, choosing $D_0$
is more delicate; this high-dimensional regime is addressed in Section~\ref{sec:mar-theory}.

\subsection{Extension to Tensor Time Series}\label{sec:tensor-extension}

The \model{} framework extends naturally to tensor-valued time series. Let
\(\{\bcalX_t\}_{t\ge 0}\subset \RR^{\calS_d}\), where
\(\calS_d=[N_1]\times\cdots\times[N_d]\). For each tensor entry
\(\i=(i_1,\ldots,i_d)\in\calS_d\), let \(\calJ_{\i}\subset\calS_d\) denote its local
neighborhood, and let \(\bcalM_{\i}\) be supported on \(\calJ_{\i}\). The tensor
\model{} model is
\[
\bcalX_t
=
\sum_{\i\in\calS_d}
\inp{\bcalX_{t-1}}{\bcalM_{\i}}
\, \e_{i_1}\circ\cdots\circ\e_{i_d}
+
\bcalE_t .
\]
After vectorization, this model can be written as a structured VAR,
\[
\x_t = \M\x_{t-1}+\bepsilon_t,
\qquad
\M=\sum_{\i\in\calS_d}
(\e_{i_d}\otimes\cdots\otimes\e_{i_1})
\vec(\bcalM_{\i})^\top,
\]
where \(\x_t=\vec(\bcalX_t)\) and \(\bepsilon_t=\vec(\bcalE_t)\). Writing
\(\x_t^{(\i)}=\bcalX_t(\calJ_{\i})\) and
\(\m_{\i}=\bcalM_{\i}(\calJ_{\i})\), the scalar regression for entry \(\i\) is
$[\bcalX_t]_{\i}=\inp{\x_{t-1}^{(\i)}}{\m_{\i}}+[\bcalE_t]_{\i}.$
Thus, coefficient estimation again reduces to a local least-squares problem. The
BIC neighborhood selector and the matrix-case theoretical arguments carry over analogously after
replacing matrix neighborhoods by tensor neighborhoods. The detailed tensor
least-squares implementation is given in Appendix~\ref{sec:supp-tensor}.

%\begin{align}\label{BIC-entrywise}
%	\bic_{ij}(k) = \log\rss_{ij}(k) +D_0\frac{|\calJ_{ij}(k)|\cdot P}{T}\log(M\vee N\vee T).
%\end{align}
%For each $(i,j)$, we select $\hat k_{ij}$ by 
%\begin{align*}
%	\hat k_{ij} = \arg\min_{k\leq K_0} \bic_{ij}(k).
%\end{align*}
%\paragraph*{Choice of $D_0$ and $K_0$. }
%We would simply set $D_0 = \log\log T$ when considering the regime in which $T,M,N$ approach infinity while $\knb = \max_{i,j}|\calJ_{ij}|$ remains fixed. The parameter $K_0$ can be determined by examining the curvature of $\bic_{ij}(k)$ directly.
%However, the choice of $D_0$ is more intricate when $\knb$ also tends to infinity. This more complex scenario will be discussed in the next section.

\section{Theoretical Guarantees}\label{sec:mar-theory}
In this section, we study the asymptotic property of the estimators $\{\hat\m_{\i}\}_{\i\in\calS}$ from Algorithm \ref{alg:pls}.
We define 
\begin{align*}
\ks = \min\bigg\{ &(2k_1+1)(2k_2+1):
\calJ_{\i}\subset\{i_1-k_1,\cdots, i_1+k_1\}\times \\
&\hspace{3cm}\{i_2-k_2,\cdots,i_2+k_2\},\ \forall \i\bigg\}.
\end{align*}
%\begin{align}\label{def:bandwidth}
%	k_1,k_2 &= \arg\min\bigg\{(2k_1+1)(2k_2+1):\\
%	&\hspace{3cm}\calJ_{ij}\subset[1\vee(i-k_1),M\wedge(i+k_1)]\times [1\vee(j-k_2),N\wedge(j+k_2)],\forall i,j\bigg\},\notag
%\end{align}
%and
%$\ks = (2k_1+1)(2k_2+1)$.
Here $\ks$ is the size of the minimal rectangular that uniformly contains all supports, directly capturing geometric locality.
% , and in particular $\ks\geq J = \max_{\i}|\calJ_{\i}|$. 
The asymptotic regimes are: \textsf{Regime 1}, only $T\rightarrow\infty$, while $M,N,\ks$ remain fixed; \textsf{Regime 2}, $T,M,N\rightarrow\infty$, while $\ks$ remains fixed; and \textsf{Regime 3}, $T,M,N,\ks\rightarrow\infty$.
We consider the case when $P=1$, and then \eqref{LIAR:lag1} can be equivalently written as $\x_t = \M\x_{t-1} + \bepsilon_t$, where $\M = \sum_{\i\in\calS} (\e_{i_2}\otimes \e_{i_1})\vec(\M_{\i})^\top$, where $\x_t = \vec(\X_t)$, and $ \bepsilon_t = \vec(\E_t)$.
% Dropping the lag subscript, \eqref{LIAR:lagP} reduces to  
% \begin{align}\label{LIAR:lag1}
% \X_t = \sum_{\i\in\calS}\inp{\X_{t-1}}{\M_{\i}}\e_{i_1}\e_{i_2}^\top+\E_t.
% \end{align}

%And the $\i$-th row of $\M$ is $\m_{\i}^\top = \vec(\M_{\i})$.

\subsection{Consistency and Asymptotic Normality}
When the spectral radius $\rho(\M)<1$, \eqref{LIAR:lag1} is a stationary matrix time series and $\X_t$ admits the following expansion in terms of the innovations:
$
\x_t = \sum_{l=0}^{\infty}\M^l\bepsilon_{t-l}.
$

\begin{assumption}\label{assump:regime}
There exist $q, \tilde q>0, \delta\in(0,1)$, and $\mu_{2q}>0$ independent of $T,M, N,\ks$, such that the following holds under different regimes:
\begin{enumerate}[label = (\arabic*)]
	\item \textsf{Regime 1} (classical): $\rho(\M)<1$; the innovation process $\{\E_{t}\}$ are i.i.d. with zero mean, and each entry has a finite $2+\tilde q$ moments, $\max_{ij}\|[\E_t]_{\i}\|_{2+\tilde q}<+\infty$;
	\item \textsf{Regime 2} (growing dimension): $\op{\M}\leq \delta$; the innovation process $\{\E_{t}\}$ are i.i.d. with zero mean, and
	$\max_{ij}\|[\E_t]_{\i}\|_{2q}\leq \mu_{2q}$ for some $q>2$; and $M N = O(T^{\beta})$, where $\beta \in (0,\frac{q-2}{4})$;
	\item \textsf{Regime 3} (growing dimension and locality): in addition to (2),
	$
	T\geq \tilde C_1\ks^4\log(M\vee N\vee T),
	$
	for some $\tilde C_1>0$ depending only on $\delta$ and $\mu_{2q}$.
\end{enumerate}
\end{assumption}
\blue{These three regimes can be viewed as increasingly high-dimensional extensions of
standard autoregressive theory. Regime~1 corresponds to the classical fixed-dimensional
time-series setting. Regime~2 allows the ambient dimension \(MN\)
to grow while the uniform locality size \(\ks\) remains fixed. The restriction
\(\beta < (q-2)/4\) comes from controlling sample covariance terms uniformly over
the growing dimension under bounded \(2q\)-th moments, and is the same type of
dimension-growth condition as that used in high-dimensional banded VAR models
\cite{guo2016high}. Regime~3 further allows the locality size \(\ks\) to increase.
The lower bound \(T \gtrsim \ks^4\log(M\vee N\vee T)\) ensures that the local
sample Gram matrices remain uniformly well-conditioned when the number of local
regressors also grows.}

The regimes in which $M,N$, and possibly $\ks$, diverge are more challenging.
Given the fact $\rho(\M)\leq \op{\M}$, \textsf{Regime 2} and \textsf{Regime 3} require stronger conditions on both the coefficient matrix and the moments on the innovation process. 
Let the auto-correlation $\bGamma_0:= \cov(\x_t, \x_t)$.
We denote $\bGamma_0(\calJ_{\i}; \calJ_{\i})\in\RR^{|\calJ_{\i}|\times |\calJ_{\i}|}$ the sub-matrix of $\bGamma_0$ by extracting the elements corresponding to the row (column) indices in $\calJ_{\i}$, arranged using column-major ordering.
The following theorem states the asymptotic property for $\hat\m_{\i}$ under either regimes. 
\begin{theorem}\label{thm:entry}
Suppose Assumption \ref{assump:regime} holds. 
Then we have 
\begin{align*}
	&\sqrt{T}\cdot(\hat\m_{\i} - \m_{\i})\implies N\big(0,\sigma_{\i}^2\cdot[\bGamma_0(\calJ_{\i};\calJ_{\i})]^{-1}\big)
\end{align*}
under any of \textsf{Regimes 1--3}, where $\sigma_{\i}^2 = \var([\E_t]_{\i})$.
\end{theorem}
This result does not require extra structure on the full innovation covariance as in \citep{chen2021autoregressive, sun2023matrix}. Our theorem implies that the coefficients $\hat\m_{\i}$ have a convergence rate of $\sqrt{T}$, which matches known results for matrix autoregression \citep{chen2021autoregressive} and matrix autoregression model with auxiliary covariate in \citep{sun2023matrix}, and continues to hold when $M,N$, even $\ks$ grow. 

\blue{\paragraph{Inference based on Theorem~\ref{thm:entry}.}
Theorem~\ref{thm:entry} also provides a direct way to construct confidence intervals for the
local autoregressive coefficients. For a fixed location $\i\in\calS$, define the
plug-in estimators
\[
\hat\bGamma_{\i}
= \frac{1}{T-1}\Y_{\i}^\top\Y_{\i},
\qquad
\hat\sigma_{\i}^2
=
\frac{\ltwo{\z_{\i}-\Y_{\i}\hat\m_{\i}}^2}
{T-1-|\calJ_{\i}|}.
\]
For any fixed vector $\a\in\RR^{|\calJ_{\i}|}$, since
$\hat\bGamma_{\i}\xrightarrow{p}\bGamma_0(\calJ_{\i};\calJ_{\i})$ under the same regimes
as in Theorem~\ref{thm:entry} and $\hat\sigma_{\i}^2$ is a consistent estimator of
$\sigma_{\i}^2$, Slutsky's theorem gives
\[
\frac{\a^\top(\hat\m_{\i}-\m_{\i})}
{\hat\sigma_{\i}\{\a^\top(\Y_{\i}^\top\Y_{\i})^{-1}\a\}^{1/2}}
\implies N(0,1).
\]
Therefore, an asymptotic $1-\alpha$ confidence interval for the linear functional
$\a^\top\m_{\i}$ is
\[
\a^\top \hat\m_{\i}
\pm
z_{1-\alpha/2}
\hat\sigma_{\i}
\{\a^\top(\Y_{\i}^\top\Y_{\i})^{-1}\a\}^{1/2},
\]
where $z_{1-\alpha/2}$ denotes the $1-\alpha/2$ quantile of the standard normal
distribution. In particular, for the $j$-th coordinate of $\m_{\i}$, by taking
$\a=\e_j$, we obtain
$\hat\m_{\i,j}
\pm
z_{1-\alpha/2}
\hat\sigma_{\i}
\left\{[(\Y_{\i}^\top\Y_{\i})^{-1}]_{j,j}\right\}^{1/2}$.
When the neighborhood is selected by the BIC criterion in \eqref{BIC-entrywise},
we use the selected neighborhood $\calJ_{\i}[\hat k_{\i}]$ in the above construction.
Under the selection consistency result in Theorem~\ref{thm:asymptotic}, this plug-in confidence interval has the
same first-order asymptotic validity as the oracle interval based on the true
neighborhood $\calJ_{\i}$.}

Since $\{\hat\m_{\i}\}_{\i\in\calS}$ are estimated based on possibly overlapped data, they are dependent to each other. 
In order to state the joint asymptotic behavior, we adopt column-major stacking operators: $\vecstack\{\v_{\i}: \i\in\calS\}$ denotes the vector formed by concatenating $\v_{\i}$ in column-major order, and $\blkstack\{\M_{\i,\j}:\i,\j\in\calS\}$ denotes the block matrix whose block rows and block columns follow the same column-major order.
%The next theorem states their joint asymptotic behavior. 
\begin{theorem}\label{thm:overall}
Suppose Assumption \ref{assump:regime} holds. 
Then we have 
\begin{align*}
	&\sqrt{T}\cdot\vecstack\{\hat\m_{\i} - \m_{\i}: \i\in\calS\}\implies N(0,\D^{-1}\C\D^{-1})
\end{align*}
under any of \textsf{Regimes 1--3},
where the matrices $\C,\D$ are arranged in a column-major order
\begin{align*}
	\C &= \blkstack\big\{[\bSigma]_{\i,\j}\cdot
	\bGamma_0(\calJ_{\i};\calJ_{\j}): \\
	&\hspace{5.4em}\i,\j\in\calS\big\}, 
	%			\mat{\bSigma_{1,1;1,1}\cdot\bGamma_0(\calJ_{1,1};\calJ_{1,1}) & \cdots &\bSigma_{1,1;M,N}\cdot\bGamma_0(\calJ_{1,1};\calJ_{M,N})\\ \vdots & \ddots&\vdots\\ \bSigma_{M,N;1,1}\cdot\bGamma_0(\calJ_{M,N};\calJ_{1,1}) & \cdots &\bSigma_{M,N;M,N}\cdot\bGamma_0(\calJ_{M,N};\calJ_{M,N})},\\
	\D &= \bdiag\big\{\bGamma_0(\calJ_\i;\calJ_\i): \\
	&\hspace{5.4em}\i\in\calS\big\},
\end{align*}
and $\bSigma = \cov(\bepsilon_t,\bepsilon_t)$. 
\end{theorem}

For \modelone{}, 
we denote the compact rank $R$ SVD of $\M^{\blk}$ (defined in \eqref{Mblk}) by
% we have the rank $R$ decomposition for the following stacked matrix: 
% \begin{align*}
% \M^{\blk} = 
% %	\blkstack\{\M^{(\i)}:\i\in\calS\}
% %	
% %	\M_e = 
% \mat{\M_{11}&\cdots&\M_{1N}\\ \vdots& &\vdots\\ \M_{M1}&\cdots &\M_{MN}} 
% % = \mat{\a_{11}&\cdots&\a_{R1}\\ \vdots& &\vdots\\ \a_{1M}&\cdots &\a_{RM}}\cdot \mat{\b_{11}&\cdots&\b_{R1}\\ \vdots& &\vdots\\ \b_{1M}&\cdots &\b_{RM}}^\top. 
% \end{align*}
% we let the rank $R$ SVD of 
$\M^{\blk}$ be $\M^{\blk} = \U^{\blk}\bSigma^{\blk}\V^{\blk\top}$. 
Note that the vectorization of $\M^{\blk}$ differs from $\vecstack\{\m_{\i}:\i\in\calS\}$ a permutation matrix $\Q$, that is $\vec(\M^{\blk}) = \Q\cdot \vecstack\{\m_{\i}:\i\in\calS\}$. 
We have the following asymptotic result. 
\begin{theorem}\label{thm:separable}
Suppose Assumption \ref{assump:regime} holds. Then we have 
\begin{align*}
	&\sqrt{T}\cdot\vec(\hat\M_{R}^{\blk} - \M^{\blk}) \implies N(0,\P\Q\D^{-1}\C\D^{-1}\Q^\top\P)
\end{align*} 
under any of \textsf{Regimes 1--3},
where $\P = \I - \V_{\perp}^{\blk}\V_{\perp}^{\blk\top}\otimes \U_{\perp}^{\blk}\U_{\perp}^{\blk\top}$.
\end{theorem}
When there is separable structure with the coefficient matrices, the estimator $\hat\M_{R}^{\blk}$ has the same convergence rate but the covariance matrix is smaller in the sense $\P\Q\D^{-1}\C\D^{-1}\Q^\top\P \leq \Q\D^{-1}\C\D^{-1}\Q^\top$ since $\P$ is a projection matrix.

\subsection{Estimation for Auto-Correlation}
In this section, we consider the estimation for the auto-correlation. 
%We define $\bGamma_j:= \cov(\x_t, \x_{t-j})\in\RR^{\calS\times \calS}$. 
We consider the following estimator
$\hat\bGamma_0 = \frac{1}{T}\sum_{t=1}^T \x_t\x_t^\top$. 
And we have the following entry-wise bound for $\linf{\hat\bGamma_0 - \bGamma_0}$, where $\linf{\A} := \linf{\vec(\A)}$ for any matrix $\A$. 
\begin{theorem}\label{thm:linf:auto-cov}
Under Assumption \ref{assump:regime}, we have 
\begin{align*}
	\linf{\hat\bGamma_0 - \bGamma_0} = O_{p}\big((T^{-1}\log(M\vee N\vee T))^{1/2}\ks\big)
\end{align*}
under any of \textsf{Regimes 1--3}.
\end{theorem}
The proof of Theorem \ref{thm:linf:auto-cov} is given in Section~\ref{sec:proof-thm4} of the Appendix.
Under \textsf{Regime 1 or 2}, the entry-wise convergence rate is of order $\tilde O(1/\sqrt{T})$, where $\tilde O$ hides the logarithm factor. 
In either regimes, $\hat\bGamma_0$ converges in probability to $\bGamma_0$.

\subsection{Consistency of the Neighborhood Selection}
We denote $J_{\infty} = \max\big\{\ks, \max_{\i\in\calS} |\calJ_{\i}[K_0]|\big\}$. 
Also recall $\calJ_{\i}[k_0] = \calJ_{\i}$ is the ground-truth.  
In order to establish the consistency of $\hat k_{\i}$, we need the following assumption:
\begin{assumption}\label{assump}
There exists $\kappa_1>0$, independent of the dimensions and sample size, such that
\begin{enumerate}[label = (\arabic*)]
	\item $\lambda_{\min}(\bGamma_0) \geq \kappa_1$; %and $\max_{ij}[\bGamma_0]_{ij,ij}\leq \kappa_2$ \red{[TODO:remove this $\kappa_2$ by $\mu_{2q}$]}.
	%		\item $\op{\M}\leq \delta$. 
	\item For each $\i\in\calS$, there exists $\tilde C_2$ depending only on $\kappa_1, \mu_{2q}$:
	\begin{align*}
		&\fro{\M_{\i}(\calJ_{\i}\backslash\calJ_{\i}[k_0-1])} \\
		&\qquad \geq \left(\frac{\tilde C_2D_0|\calJ_{\i}|}{T}
		\log(M\vee N\vee T)\right)^{1/2};
	\end{align*}
	\item[(3)] Under \textsf{Regime 3}, in addition to (2) in Assumption \ref{assump:regime}, $T\geq \tilde C_1 J_{\infty}^4\log(M\vee N\vee T)$ for some $\tilde C_1>0$ depending only on $\delta, \mu_{2q}$.
\end{enumerate}
\end{assumption}

Here the first condition in Assumption \ref{assump} guarantees that the auto-correlation $\cov(\x_t)$ is strictly positive definite.
The second condition ensures that the index set $\calJ_{\i}[k_0]$ is asymptotically identifiable, and
\[
\left(\tilde C_2D_0\frac{|\calJ_{\i}|}{T}\log(M\vee N\vee T)\right)^{1/2}
\]
is the minimum order of the non-zero coefficients.
And in the third assumption, we need a slightly stronger assumption on the rate at which $J_{\infty}$ grows.

Together with this assumption, we can guarantee the consistency of the selection of neighborhood. 
Recall $\knb = \max_{\i}|\calJ_{\i}|$.
\begin{theorem}\label{thm:asymptotic}
Suppose Assumption \ref{assump:regime}, \ref{assump} holds.
As long as we choose $D_0 \geq \tilde C\knb^2\ks$ for some constant $\tilde C$ depending only on $\kappa_1,\delta, \mu_{2q}$, we have for each $\i\in\calS$, $\PP(\hat k_{\i} = k_0) \rightarrow 1$ under any of \textsf{Regimes 1--3}.
\end{theorem}
Under \textsf{Regime 1 or 2}, when $\ks$ remains fixed, it suffices to set $D_0 = \log(\log T)$. However, Theorem \ref{thm:asymptotic} allows $\ks$ to go to infinity as well. In this case, we need to set $D_0 = C_J\log(\log T)$, where $C_J\geq \knb^2\ks$. Although $\knb$ is not known in practice, we can choose $C_J$ to be large and the condition in Theorem \ref{thm:asymptotic} is satisfied. However, a too large choice of $C_J$ would lead to a stricter marginal condition in Assumption \ref{assump}.

\section{Numerical Experiments}\label{sec:numeric}
In this section, we present a series of numerical experiments to demonstrate the effectiveness of our proposed methodology. Throughout, we focus on both model selection and estimation accuracy, and we compare against several existing benchmarks. 
Unless otherwise noted, we conduct 50 independent trials for each experiment.

Our first experiment evaluates the entrywise BIC in \eqref{BIC-entrywise} for neighborhood-size selection under the lag-one \model{} model \eqref{LIAR:lag1}:
\begin{align*}
\X_t = \sum_{\i}\inp{\X_{t-1}}{\M_{\i}}\e_{i_1}\e_{i_2}^\top+\E_t,
\end{align*}
where $\M_{\i}$ is supported on the region with center at $\i$ and neighborhood size $K=3$ (although our framework allows the neighborhood size to vary across locations, in this simulation we set uniform $K$ for all pixels). 
\blue{For data generation, the nonzero local coefficients were generated within the true
neighborhood and then rescaled to ensure stationarity of the resulting time series.
To maintain a clear signal at the boundary of the true neighborhood, we used slightly
larger coefficients on the outermost ring of the true neighborhood. The innovations
had independent mean-zero Gaussian entries with standard deviation \(0.12\), and the
first \(600\) observations were discarded as burn-in.}
We vary the size of the matrix $(M,N)$ from $\{(10,10), (15,15), (20,20)\}$, and the sample size 
% $T\in\{4000,5000,6000\}$. 
$T\in\{600,800,1000\}$. 
For each configuration, the entrywise BIC is applied over the candidate neighborhood sizes $\{0,1,2,3,4,5\}$ and success is recorded when the selected neighborhood size equals the truth $K=3$. 

Figure~\ref{fig:BIC} displays the success rates for each $(M,N)$ and $T$ combination (together with a per-pixel success-rate heatmap). The results show that the BIC selector
remains informative in these moderate sample settings.
The success rate increases with \(T\) but declines as the size $(M,N)$ grows. The per-pixel heatmap further shows slightly lower rates near the boundaries than in the interior. 
This phenomenon can be explained by the edge effect: boundary pixels have truncated neighborhoods, weakening the signal for identifying $K$ and leading BIC to favor smaller sizes; interior pixels, with full neighborhoods, provide stronger signal.

% \begin{figure}[htbp]
% \centering
% \includegraphics[width=\linewidth]{figure/success_rate_grid.png}
% \caption{Success rates for $M=N\in\{10,15,20\}$
% 	and $T\in\{4000,5000,6000\}$.}
% \label{fig:BIC}
% \end{figure}
\begin{figure}[!t]
\centering
\includegraphics[width=0.85\linewidth]{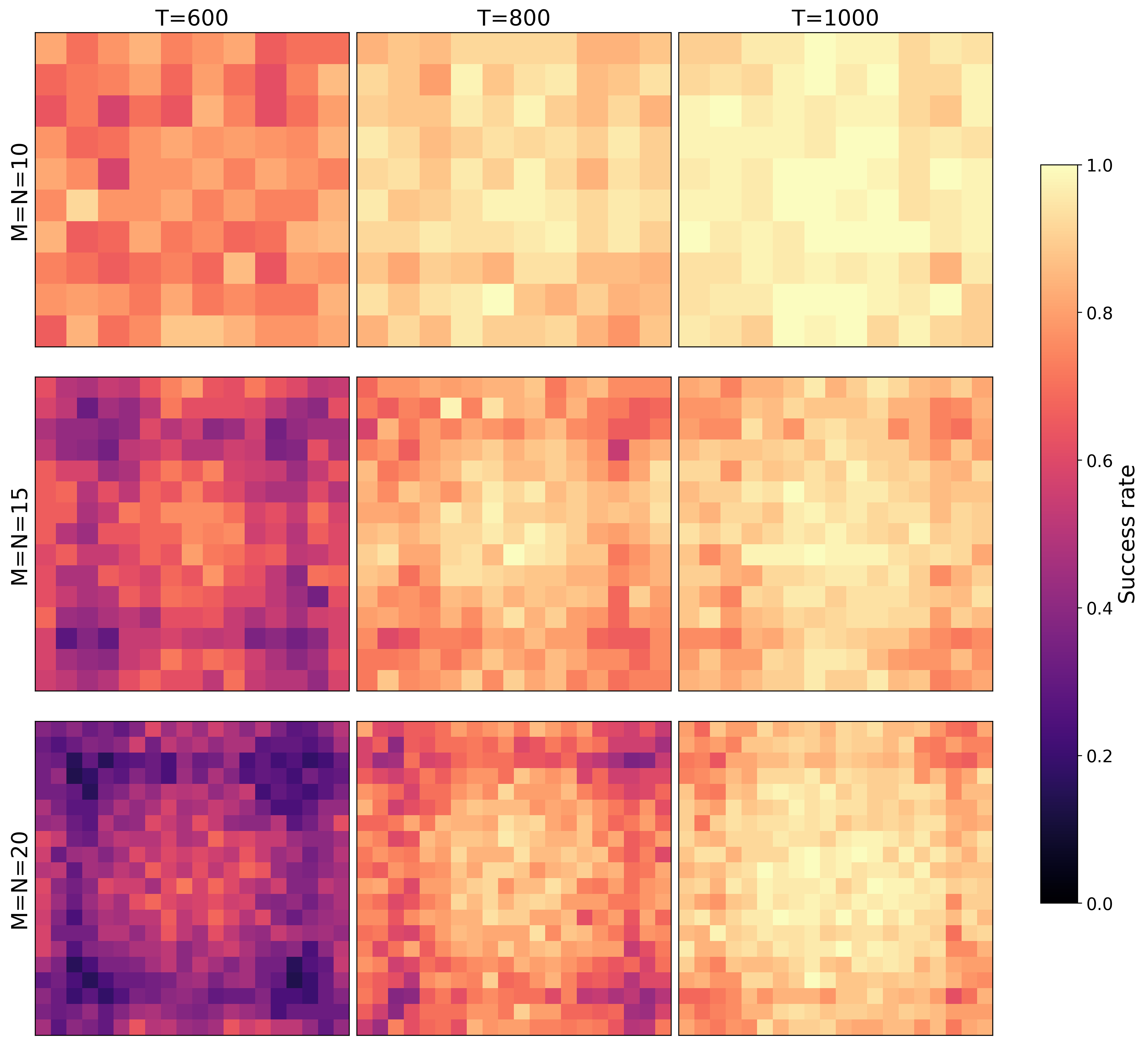}
\caption{Success rates for $M=N\in\{10,15,20\}$
	and $T\in\{600,800,1000\}$.}
\label{fig:BIC}
\end{figure}

Our second experiment evaluates estimation accuracy under the lag-one \model{} model with \(P=1\), \((M,N)=(10,10)\), and true neighborhood size \(K=3\). We vary the sample size \(T\in\{500,1000,1500,\dots,5000\}\). We consider three quantities: (i) kernel estimation error (Frobenius norm), (ii) auto-covariance \(\bGamma_0\) error (Frobenius norm), and (iii) auto-covariance \(\bGamma_0\) error (entrywise \(\ell_\infty\) norm).

Figure~\ref{fig:errors_vs_T} displays the log-scale errors versus \(T\) and the ratio $\fro{\bGamma_0 - \hat\bGamma_0}/\linf{\bGamma_0 - \hat\bGamma_0}$. All three errors decrease as \(T\) increases. The ratio remains approximately constant (about \(23\)) for larger \(T\), indicating that the error is broadly distributed across entries rather than concentrated on a few.
Figure~\ref{fig:errors_vs_invsqrtT} plots the means of the three errors against \(1/\sqrt{T}\). The errors scale approximately linearly with \(1/\sqrt{T}\), especially for larger \(T\), in line with the theoretical \(T^{-1/2}\) rate and supporting consistency under large samples.

\begin{figure}[htbp]
\centering
\includegraphics[width=0.9\linewidth]{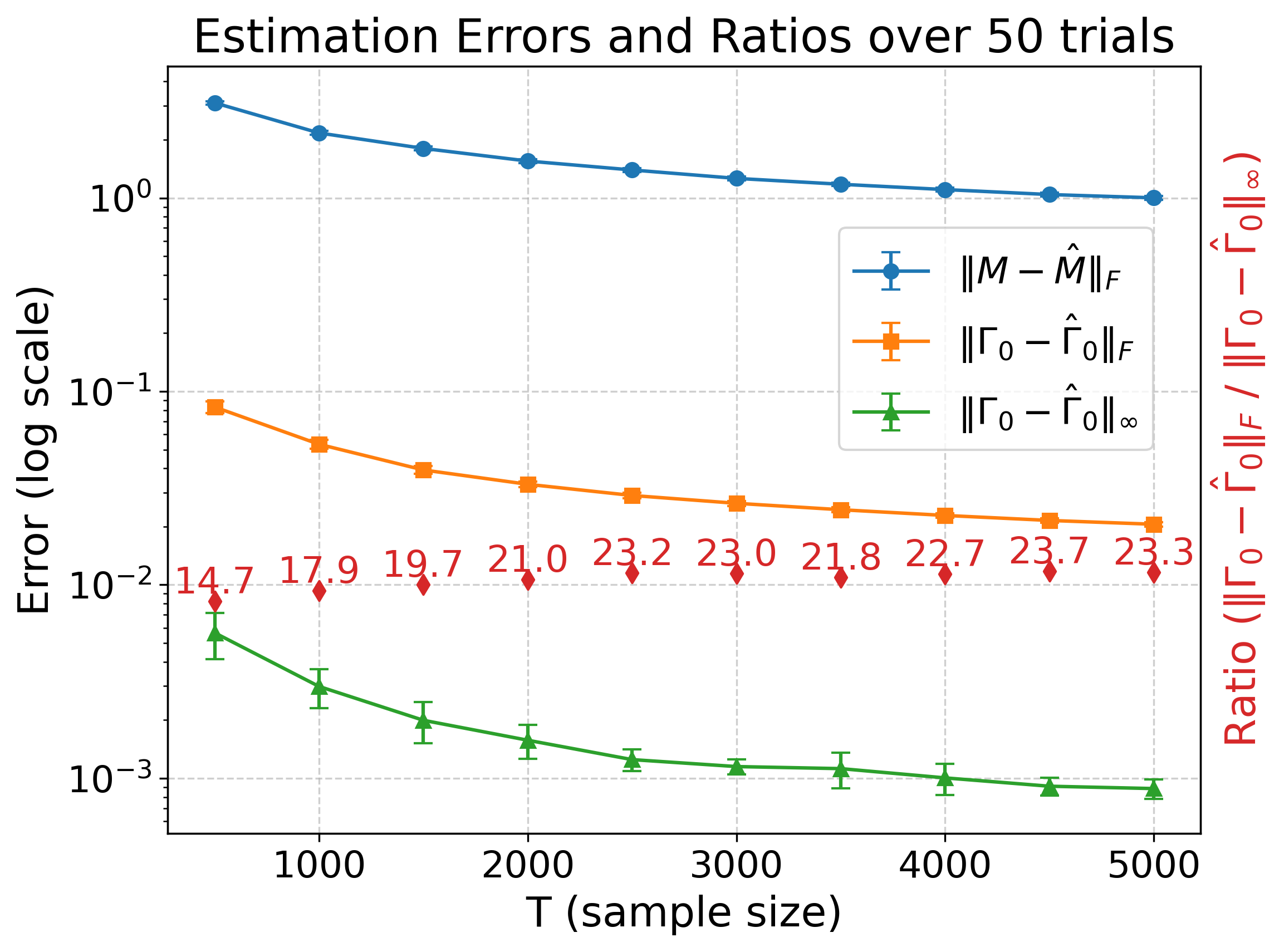}
\caption{Log scale of estimation errors of the kernels and auto-covariance $\bGamma_0$ in Frobenius and infinity norms versus $T$. The ratio of errors for the auto-covariance is also displayed.}
\label{fig:errors_vs_T}
\end{figure}

\begin{figure}[htbp]
\centering
\includegraphics[width=\linewidth]{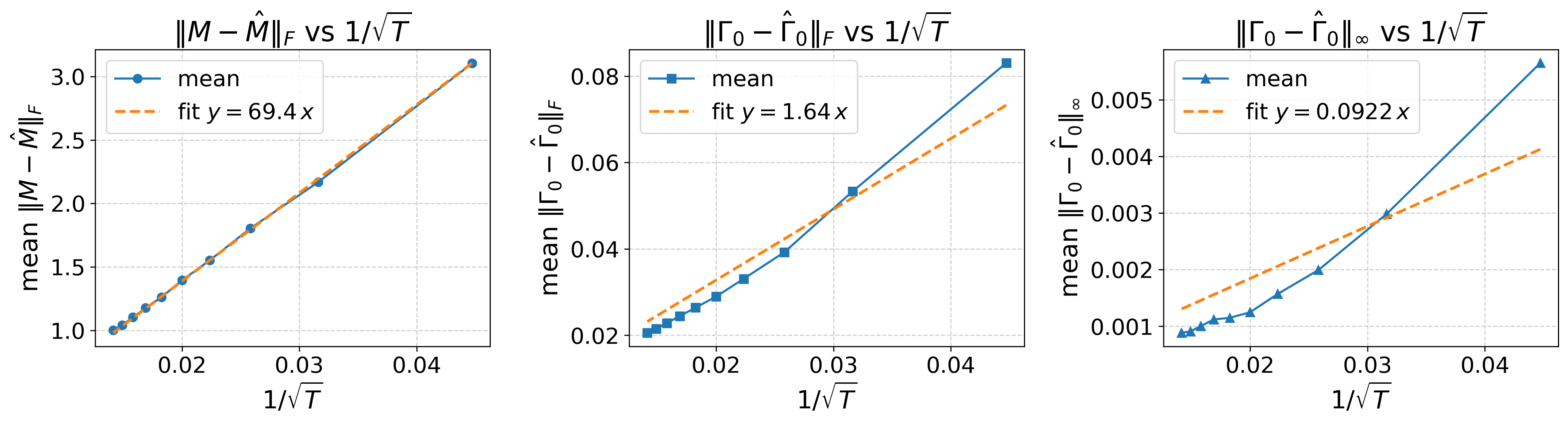}
\caption{Mean of the estimation errors of the kernels and auto-covariance $\bGamma_0$ in Frobenius and infinity norms versus $1/\sqrt{T}$.}
\label{fig:errors_vs_invsqrtT}
\end{figure}

In the third experiment, we compare the performance of our proposed model \model{} and its separable variant \modelone{} against three existing methods: (1) MAR \citep{chen2021autoregressive}, $\X_t = \sum_{p=1}^P\A_p\X_{t-p}\B_p^\top + \E_t$; (2) MAR-ST \citep{hsu2021matrix}, $\X_t = \sum_{p=1}^P\A_p\X_{t-p}\B_p^\top+ \E_t$, where $\A_p,\B_p$ are banded matrices; and (3) pixel-wise autoregression (\textsf{LIAR-P}), $[\X_t]_{\i} = \sum_{p=1}^P\beta_{\i p}[\X_{t-p}]_{\i}+[\E_t]_{\i}$ for each $\i\in\calS$. The last model can be seen as a special case of \model{} with $\calJ_{\i} = \{\i\}$.
\blue{We set \((M,N)=(20,20)\), \(T=200\), and true neighborhood size \(K=2\). The observations are split chronologically, with the first \(90\%\) used for model fitting and the remaining \(10\%\) used for testing. Prediction accuracy is measured by one-step-ahead full-matrix RMSE: at each test time point \(t\), the fitted model uses the observed previous lag(s) to form \(\widehat{\X}_t\), which is compared with the observed \(\X_t\); the RMSE is then computed over all matrix entries and all test time points.}
For MAR-ST we fix the bandwidth at \(K=2\). For MAR and MAR-ST we cap the iterations at \(50\) with tolerance \(10^{-7}\). 

We report prediction RMSE and runtime. As shown in Figure~\ref{fig:rmse_runtime}, \model{} is substantially faster than MAR and MAR-ST while achieving a markedly smaller prediction RMSE. \textsf{LIAR-P} is faster than \model{} but attains an RMSE that is about \(7\%\) higher.

\begin{figure}[htbp]
\centering
\includegraphics[width=\linewidth]{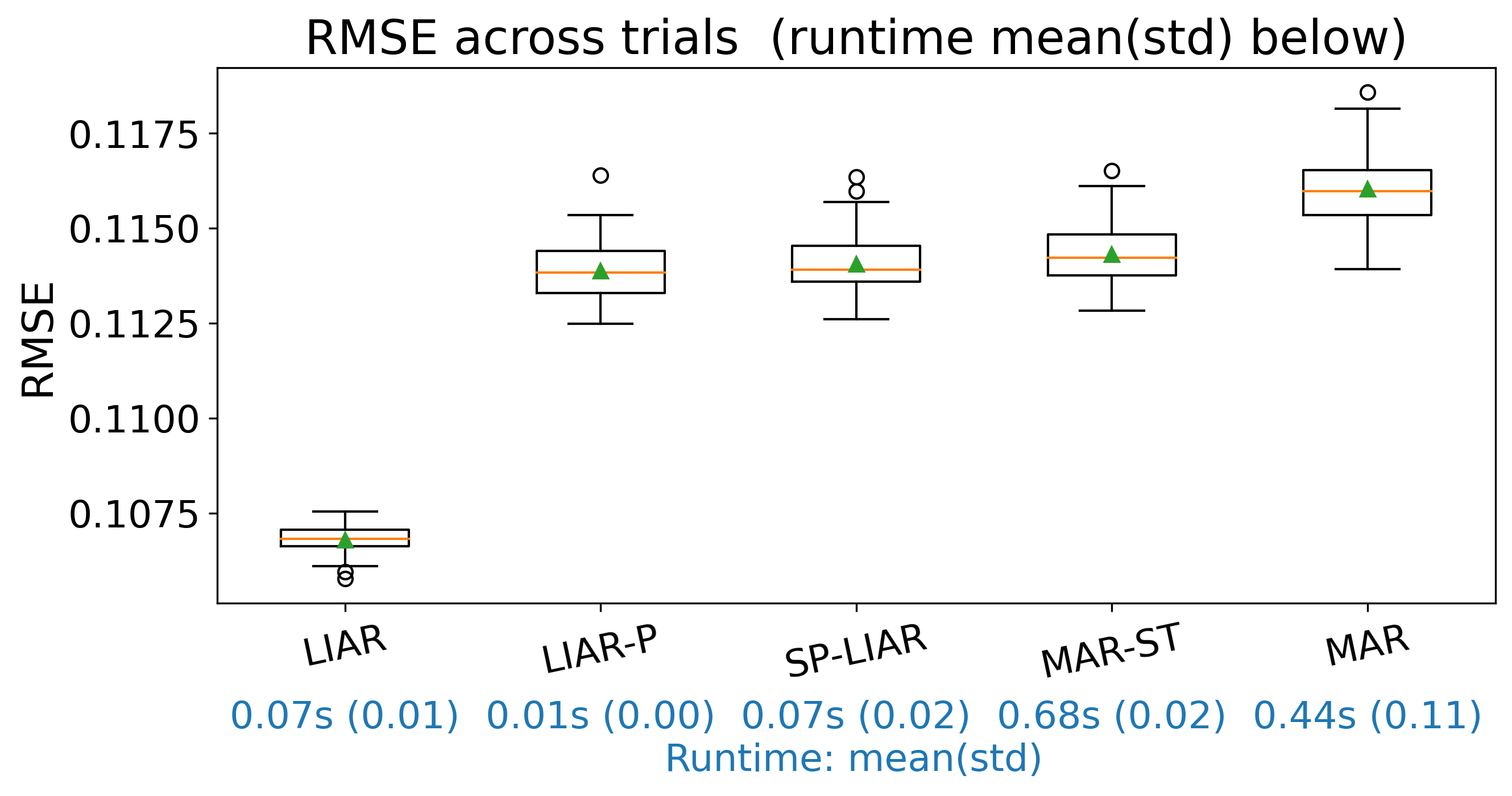}
\caption{RMSE and runtime comparison of \model{}, \modelone{}, MAR, MAR-ST, and \textsf{LIAR-P}.}
\label{fig:rmse_runtime}
\end{figure}

\section{Real data: Total Electron Content}\label{sec:real-data}
\blue{In this section, in addition to prediction accuracy and computational cost, we examine what the fitted local interaction structure reveals about TEC dynamics. We use the BIC-selected neighborhoods as data-driven summaries of local spatio-temporal dependence and compare them across different time periods. Additional tensor-based TEC experiments are reported in Appendix~\ref{sec:supp-tensor}.}

\subsection{\model{} for Matrix Time Series}
We analyze ionospheric total electron content (TEC) fields derived from multi-frequency Global Navigation Satellite System (GNSS) signals, a widely used proxy in space-weather and ionospheric studies. We use the completed TEC product of \citet{sun2023complete}. The original data are sampled on a \(1^\circ\times1^\circ\) (latitude \(\times\) longitude) grid every 5 minutes (\(181\times361\) per time point). For preprocessing, we aggregate the data to \(2^\circ\times2^\circ\) and 15-minute resolution by averaging non-overlapping \(2\times2\) spatial blocks and averaging consecutive three-frame windows in time, yielding a \(91\times181\) matrix time series at each time point.

\paragraph{Neighborhood selection.}
We analyze six representative 10-day periods: 2017-06, 2017-09, 2018-06,
2019-03, 2019-06, and 2019-12.
For each month, we use TEC data from the 11th to the 20th (inclusive), yielding
\(T=960\) time points at 15-minute resolution.
For each period, we apply the entrywise BIC to select the neighborhood size over
square candidates \(K\in\{0,1,2,3,4,5\}\).
Figure~\ref{fig:tec-bic-a} summarizes the selections; across all six periods, the
most frequently chosen sizes are \(K=1\) or \(K=2\).
\blue{The concentration on small neighborhoods suggests that short-term TEC evolution is mainly driven by localized spatial propagation. The variation across locations further indicates that the strength and spatial range of local TEC dependence are spatially heterogeneous.}
Figure~\ref{fig:tec-temporal-neighborhood} illustrates the connection between the raw TEC evolution and the fitted local interaction structure for a representative period on June 14, 2019. The first six panels show consecutive TEC fields at 15-minute intervals, and the last panel shows the corresponding BIC-selected neighborhood size \(K\). The white boundaries separate regions with \(K<2\) from those with \(K\ge 2\). 
\blue{The selected neighborhood sizes are broadly related to local TEC behavior: relatively stable regions, such as parts of the polar areas and the dark low-TEC region in the left-central part of the maps, tend to have smaller neighborhoods, whereas regions with more visible temporal and spatial variation tend to have broader neighborhoods. Thus, the selected neighborhood map gives an interpretable summary of the spatially heterogeneous local dependence learned by \model{}.}

\begin{figure}[t]
\centering
\includegraphics[width=0.95\linewidth]{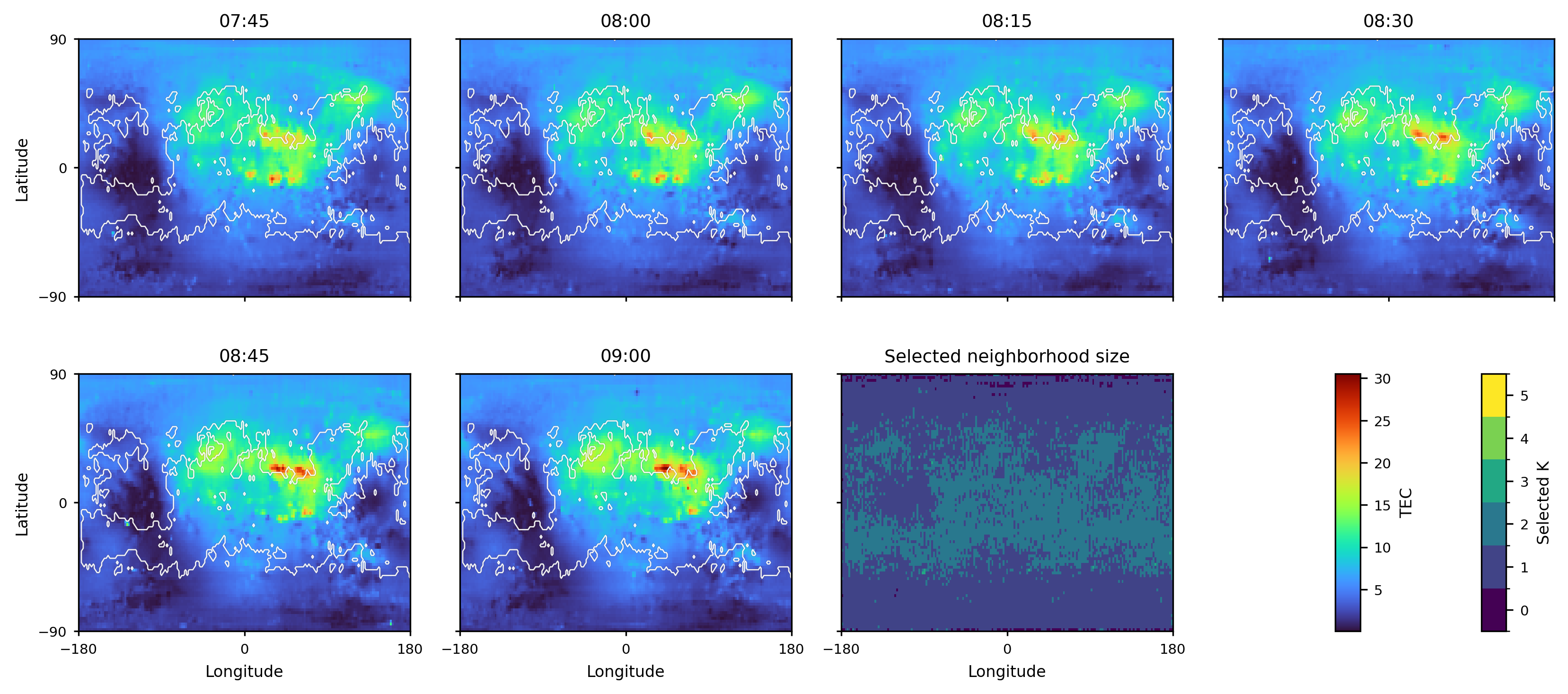}
\caption{Representative TEC fields over consecutive 15-minute time points and the corresponding BIC-selected neighborhood-size map. For visualization, the boundary overlay is slightly smoothed to highlight the large-scale spatial pattern.}
\label{fig:tec-temporal-neighborhood}
\end{figure}

For cross-period comparison, we coarsen the selected neighborhood sizes into two categories,
\(K<2\) and \(K\ge 2\), separating the most localized selections from
relatively broader local neighborhoods.
For each anchor period, we compare it with each of the other five periods by computing
(i) Cohen’s \(\kappa\) between the corresponding two-category maps and
(ii) the Pearson correlation between the corresponding raw TEC fields.
The results are shown in Figure~\ref{fig:tec-kappa-corr}.

In short, Figures~\ref{fig:tec-bic-a} and \ref{fig:tec-kappa-corr} show two points:
(i) The share of pixels with \(K<2\) versus \(K\ge 2\) is relatively stable across periods
(Figure~\ref{fig:tec-bic-a}).
(ii) There is a clear positive association between the Pearson correlation of the raw TEC
fields and Cohen’s \(\kappa\) of the neighborhood maps (Figure~\ref{fig:tec-kappa-corr}).
The positive association suggests that the selected local interaction structure is related to
the underlying ionospheric state: when the TEC fields are similar across periods, the
\emph{local} predictive dependence learned by \model{} is also similar.
\blue{Thus, the fitted \model{} model gives an interpretable summary of how the local dependence structure changes with the TEC field.}

\begin{figure}[t]
\centering
\begin{subfigure}[t]{0.58\textwidth}
	\centering
	\includegraphics[width=\linewidth]{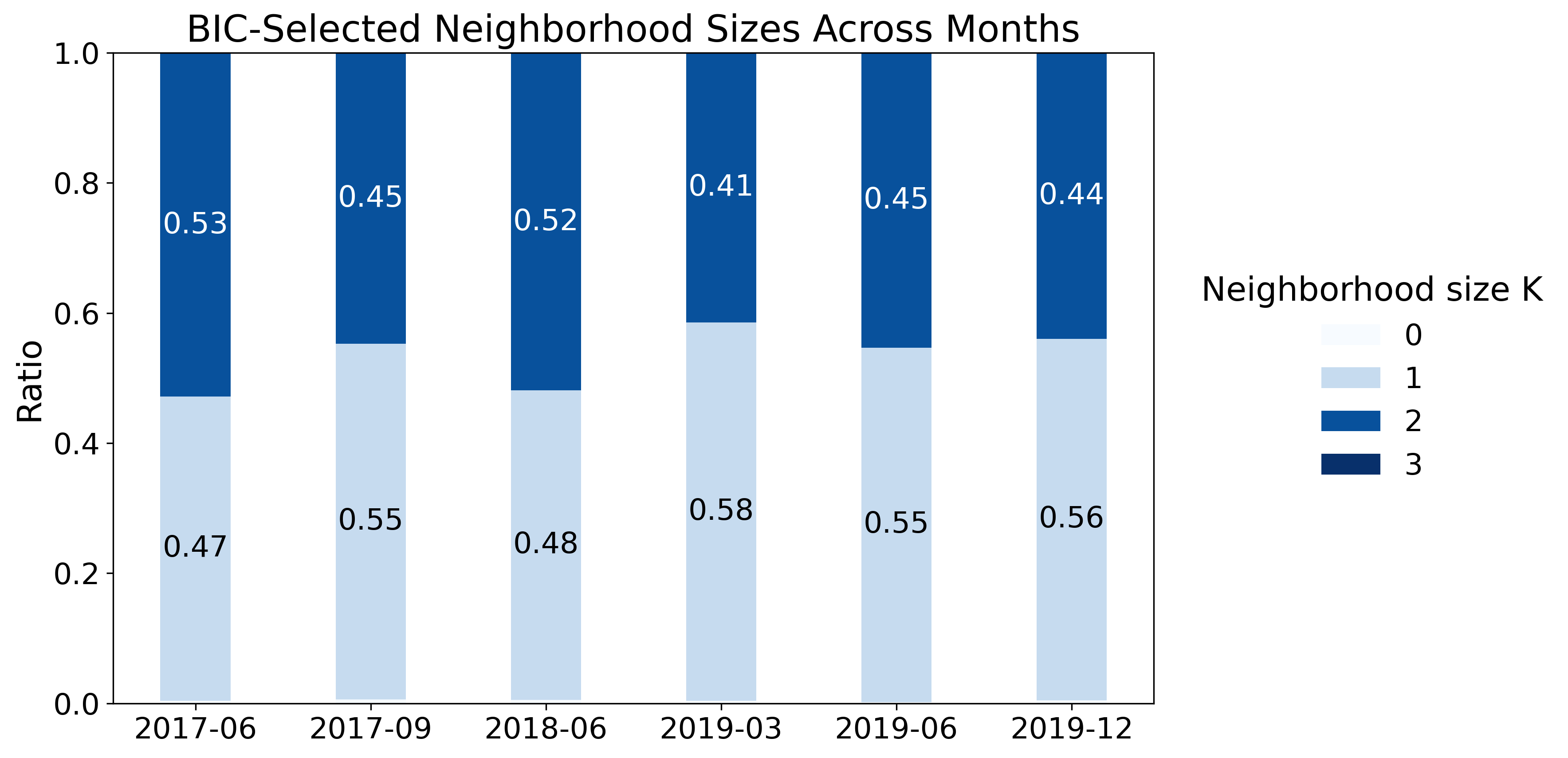}
	\caption{Distribution of BIC-selected neighborhood sizes across six time periods}
	\label{fig:tec-bic-a}
\end{subfigure}\hfill
\begin{subfigure}[t]{0.38\textwidth}
	\centering
	\includegraphics[width=\linewidth]{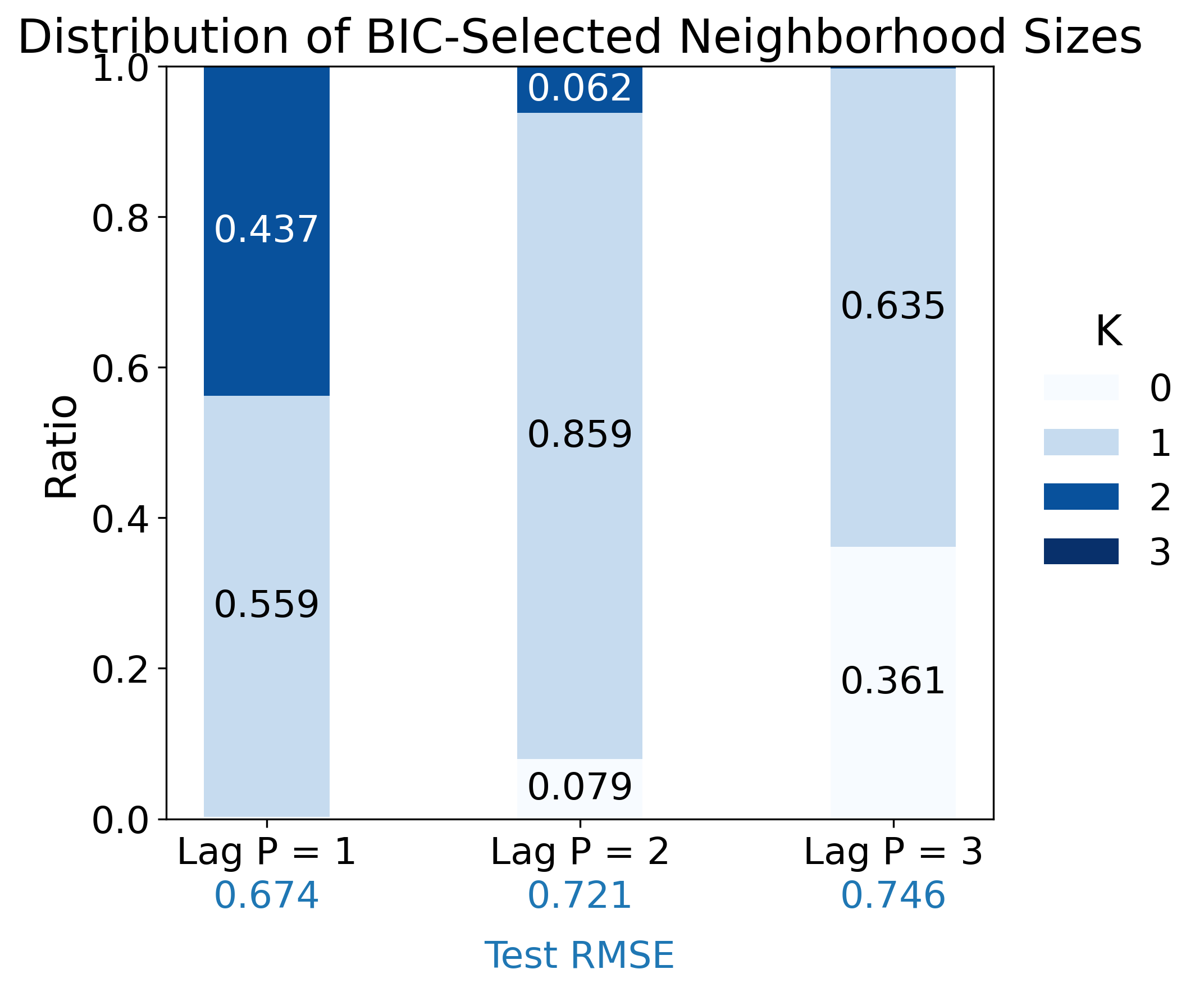}
	\caption{Distribution of BIC-selected neighborhood sizes versus lag $P$}
	\label{fig:tec-bic-b}
\end{subfigure}
\caption{Distribution of BIC-selected neighborhood sizes}
\label{fig:tec-bic-and-k-vs-p}
\end{figure}

% --- Figure (b): Pearson correlation vs. Cohen's kappa ---
\begin{figure}[H]
\centering
\includegraphics[width=0.9\linewidth]{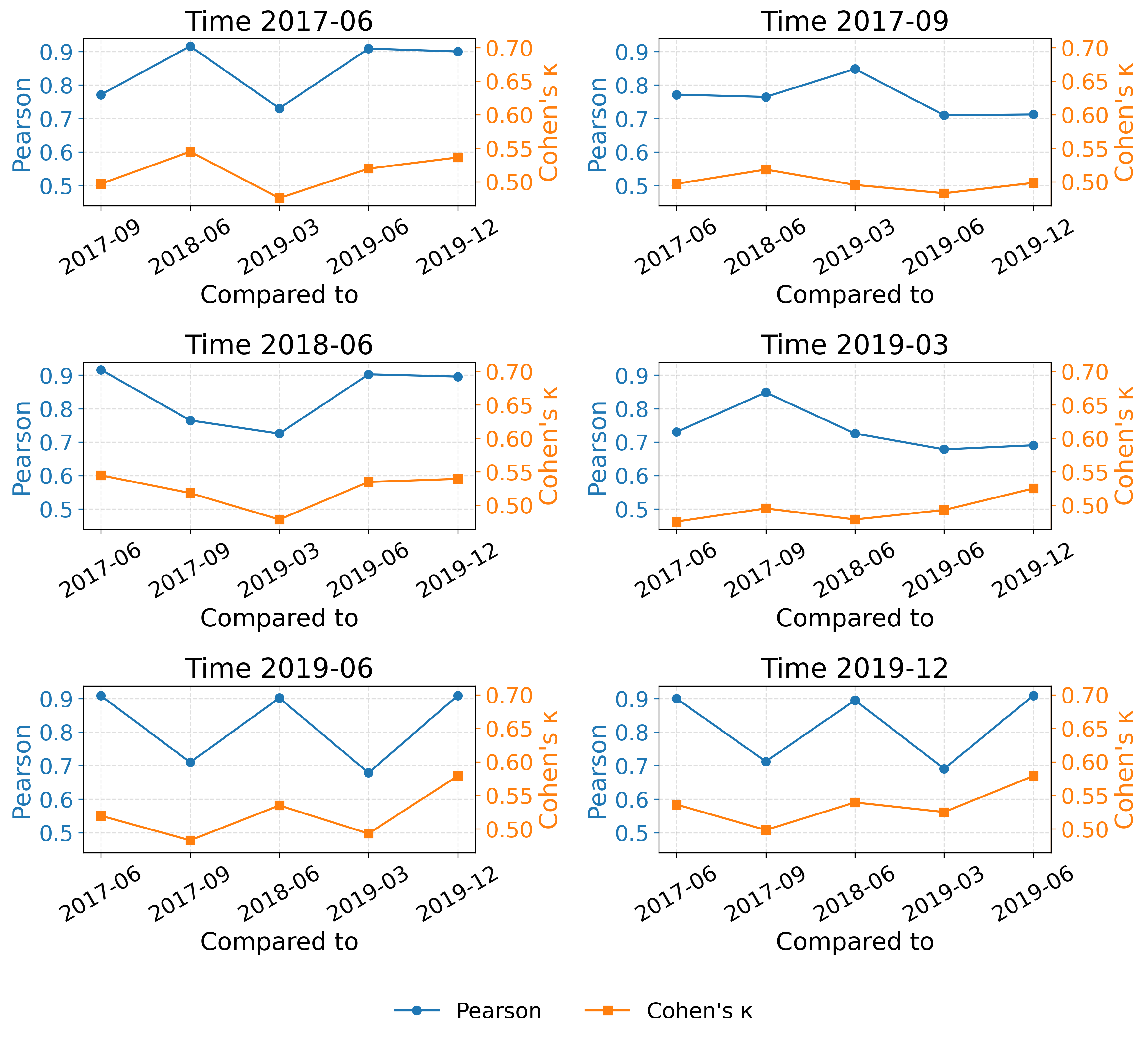}
\caption{Pairwise similarity across six periods: Pearson correlation of TEC fields and Cohen’s \(\kappa\) for the two-category neighborhood maps.}
\label{fig:tec-kappa-corr}
\end{figure}

\paragraph{Lag dependence of selected neighborhoods.}
We focus on June 2019, using \(90\%\) of the series for training and \(10\%\) for testing. We vary the lag \(P\in\{1,2,3\}\) and, for each \(P\), apply the entrywise BIC to select the neighborhood size from \(K\in\{0,1,2,3,4,5\}\). Figure~\ref{fig:tec-bic-b} reports the selection frequencies versus \(P\); the bottom panel shows the test RMSE obtained with the BIC-selected neighborhoods.

Figure~\ref{fig:tec-bic-b} shows that (i) as \(P\) increases, the selected neighborhoods tend to shrink, suggesting that additional temporal lags reduce the need for wider spatial neighborhoods; and (ii) the held-out RMSE does not improve with larger \(P\), which supports the use of \(P=1\) in the reported TEC analysis.

\paragraph{Comparisons with other methods.}
We compare \model{} and its separable variant \modelone{} with MAR \citep{chen2021autoregressive}, MAR-ST \citep{hsu2021matrix}, and \textsf{LIAR-P} across the six 10-day periods described above, using the first \(90\%\) of each series for training and the remaining \(10\%\) for testing. Test RMSE is computed from one-step-ahead full-field predictions over the test period: at each test time point \(t\), the fitted model uses the observed previous lag to form \(\widehat{\X}_t\), which is compared with \(\X_t\); the RMSE is then averaged over all test time points and spatial grid entries.
For each period, we fit all methods and report test RMSE in Figure~\ref{fig:tec-rmse-time}; runtimes for \model{}, MAR, and MAR-ST are reported in Table~\ref{table:runtime}.

The pixel-wise baseline \textsf{LIAR-P} is the fastest method, but it also has the largest RMSE. \modelone{} is closest in structure to MAR-ST; its slightly higher RMSE is expected, since \modelone{} uses a one-shot SVD projection whereas MAR-ST refines the factors iteratively, a pattern also observed by \cite{chen2021autoregressive}. The unrestricted \model{} has RMSE close to MAR and slightly below MAR-ST, while taking much less time than both. Thus, in this TEC example, imposing locality keeps most of the predictive information while reducing the computational cost.
\blue{This agrees with the neighborhood maps above: the selected local neighborhoods reflect meaningful TEC dependence and are useful for prediction, improving one-step-ahead full-field forecasts relative to pixel-wise autoregression.}

\begin{figure}[H]
\centering
\includegraphics[width=0.9\linewidth]{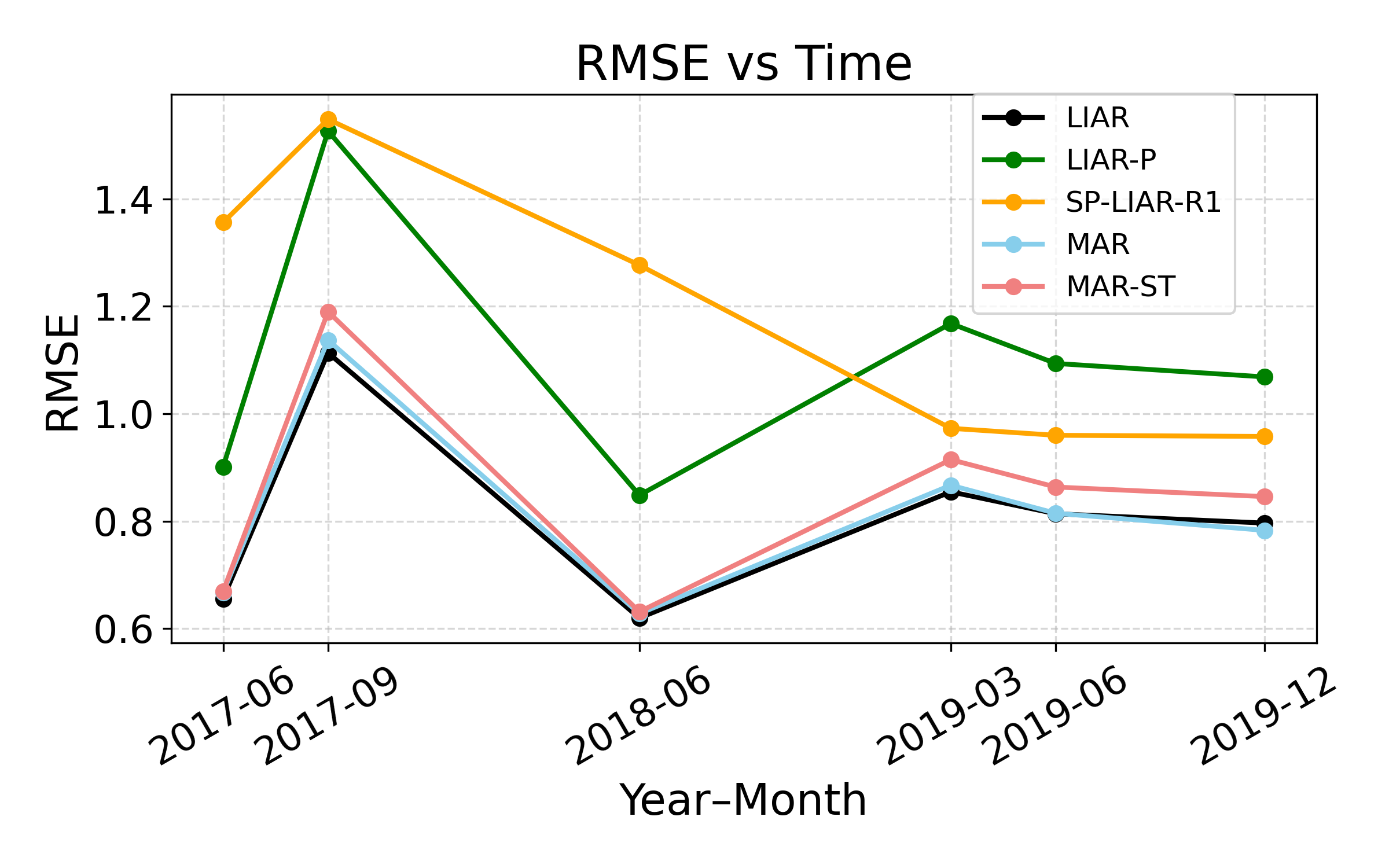}
\caption{Prediction RMSE over time on TEC data, comparing \model{}, \modelone{}, MAR, MAR-ST, and \textsf{LIAR-P}}
\label{fig:tec-rmse-time}
\end{figure}
\begin{table}[t]
\centering
\caption{Runtime (s) by method and period}
\label{table:runtime}
\setlength{\tabcolsep}{6pt}
\renewcommand{\arraystretch}{1.1}
\begin{tabular}{lrrrrrr}
	\hline
	& 2017-06 & 2017-09 & 2018-06 & 2019-03 & 2019-06 & 2019-12 \\
	\hline
	\model{}   & 11.26 &  8.41 &  8.90 &  8.45 &  4.43 &  6.84 \\
	MAR    & 467.03 & 462.89 & 451.69 & 438.63 & 225.81 & 452.26 \\
	MAR-ST & 81.30 & 63.00 & 88.03 & 81.54 & 35.18 & 69.74 \\
	\hline
\end{tabular}
\end{table}

\section{Discussion and Conclusion}
We proposed the \emph{local interaction autoregressive} (\model{}) framework, a general principle for modeling high-dimensional matrix and tensor time series via short-range spatio-temporal dependence. In the matrix setting, the proposed \modelone{} serves as a bridge between traditional MAR/MAR-ST modeling and the \model{} framework. By restricting each entry’s evolution to a data-driven neighborhood, \model{} ensures parsimony and interpretability while maintaining predictive accuracy.

\blue{A natural extension of the proposed neighborhood selector is to allow the selected neighborhoods to vary smoothly over space.
In the current implementation, the BIC-based neighborhood selection is performed
separately for each spatial location. This location-wise selection preserves local adaptivity, but it does not explicitly impose spatial coherence on the resulting neighborhood-size map. A simple practical remedy is to post-process
the selected neighborhood-size map, for example by applying a local median filter or a local majority
rule over a small spatial window. Such a step is straightforward to implement and may reduce isolated selections while retaining spatial heterogeneity.

A more integrated approach is to incorporate spatial smoothness directly into the
model-selection criterion. Let \(\mathcal E\) be the set
of adjacent spatial pairs. One may consider a spatially regularized BIC criterion
\[
\widehat{\k}
=
\arg\min_{\{\k = (k_{\i})_{\i\in\calS}: k_{\i}\leq K_0\}}
\left\{
\sum_{\i} \mathrm{BIC}_{\i}(k_{\i})
+
\lambda \sum_{(\i,\j)\in\mathcal E}|k_\i-k_\j|
\right\},
\]
where \(\lambda\ge 0\) controls the degree of spatial smoothness. The first term
retains the location-wise BIC fit, while the second term discourages abrupt changes of
neighborhood sizes between adjacent locations. This formulation can preserve
spatially varying local interactions while encouraging smoothly varying neighborhood
patterns. Its implementation, however, requires solving a discrete spatial labeling
problem and selecting the additional tuning parameter \(\lambda\). Ordered-label
optimization methods or approximate graph-based algorithms may be useful for
this purpose; we leave a detailed computational and theoretical study of this
spatially regularized selection procedure to future work.}

We developed scalable estimation procedures, including parallel least squares and projection-based methods, a BIC-type neighborhood selector, and statistical theory for kernel and auto-covariance estimation. Extensive simulations show that \model{} achieves lower prediction error and markedly better runtime over competing baselines, including MAR and MAR-ST, across diverse settings. On real-world data, \model{} captures salient local dynamics and provides practical forecasting gains.

\section{Acknowledgements}

We thank Hu Sun (University of Michigan, Ph.D. 2024) and Professor Han Xiao (Rutgers University) for discussions on the matrix autoregressive model with banded structures. YC acknowledges support from NSF AGS Award 2419187, NASA Federal Award No. 80NSSC23M0192, and No. 80NSSC23M0191.

\bibliographystyle{chicago}
\bibliography{reference}

\appendix

\section{Proofs of Main Theorems}\label{sec:main-proofs}
The auxiliary lemmas used in the proofs are collected in Section~\ref{sec:aux-lemmas}. We present
the proofs of the main theorems in the same order as they appear in the main text.

\subsection{Proof of Theorem 1}
The proof is the one-block version of the martingale central limit theorem argument
used in the proof of Theorem~2. Since Theorem~1 also follows immediately from
the joint result in Theorem~2, we present the short marginal derivation here for
completeness.

By Theorem~2,
\[
\sqrt{T}\,\vecstack\{\hat\m_{\i}-\m_{\i}:\i\in\calS\}
\Rightarrow N(0,\D^{-1}\C\D^{-1}).
\]
Taking the block corresponding to a fixed location \(\i\), the limiting covariance is
\[
\bGamma_0(\calJ_{\i};\calJ_{\i})^{-1}
\{[\bSigma]_{\i,\i}\bGamma_0(\calJ_{\i};\calJ_{\i})\}
\bGamma_0(\calJ_{\i};\calJ_{\i})^{-1}
=
\sigma_{\i}^2\bGamma_0(\calJ_{\i};\calJ_{\i})^{-1},
\]
where \([\bSigma]_{\i,\i}=\operatorname{Var}([\E_t]_{\i})=\sigma_{\i}^2\). Hence
\[
\sqrt{T}(\hat\m_{\i}-\m_{\i})
\Rightarrow
N\left(0,\sigma_{\i}^2\bGamma_0(\calJ_{\i};\calJ_{\i})^{-1}\right),
\]
which proves Theorem~1.

\subsection{Proof of Theorem 2}
\begin{proof}
For each $\i$, $\hat \m_{\i}$ admits the following closed-form solution:
\begin{align*}
	\hat \m_{\i} = (\Y_{\i}^\top\Y_{\i})^{-1}\Y_{\i}^\top\z_{\i} = \m_{\i} + (\Y_{\i}^\top\Y_{\i})^{-1}\Y_{\i}^\top\bnu_{\i}.  
\end{align*}
	Therefore, 
	\begin{align*}
		\vecstack\{\hat\m_{\i} - \m_{\i}:\i\in\calS\} = \vecstack\{(\Y_{\i}^\top\Y_{\i})^{-1}\Y_{\i}^\top\bnu_{\i}:\i\in\calS\}.
	\end{align*}
	Let $\bgamma = \vecstack\{\bgamma_{\i}:\i\in\calS\}$, then we have 
	\begin{align*}
%		T^{-1/2}\bgamma^\top \mat{\Y_{11}^\top\bepsilon_{11}  \\\vdots \\\Y_{MN}^\top\bepsilon_{MN}}
		&T^{-1/2}\bgamma^\top\vecstack\{\Y_{\i}^\top\bnu_{\i}: \i\in\calS\}
		 =T^{-1/2}\sum_{\i} \bgamma_{\i}^\top\Y_{\i}^\top\bnu_{\i} \\
		&\qquad = \sum_{t=2}^{T} \underbrace{T^{-1/2}\sum_{\i}\bgamma_{\i}^\top\x_{t-1}^{(\i)}[\E_t]_{\i}}_{:=Z_{t,T}}.
	\end{align*}
	Let $\calA_{t-1}$ denote the $\sigma$-field generated by previous innovations. 
	Then
	\[
	\EE(Z_{t,T}\mid\calA_{t-1}) = 0,
	\]
	and
	\begin{align*}
		\EE(Z_{t,T}^2\mid \calA_{t-1})
		&= T^{-1}\sum_{\i,\j\in\calS}\bSigma_{\i,\j}\bgamma_{\i}^\top\x_{t-1}^{(\i)}
		\x_{t-1}^{(\j)\top}\bgamma_{\j}.
	\end{align*}
Denote
\begin{align*}
	V_{TT}^2
	= T^{-1}\sum_{t=2}^T \sum_{\i,\j\in\calS}
	\bSigma_{\i,\j}\,
	\bgamma_{\i}^\top\,\x_{t-1}^{(\i)} \x_{t-1}^{(\j)\top}\,\bgamma_{\j}.
\end{align*}
On the other hand, we have
\begin{align*}
	s_{TT}
	:= \EE\Big(\sum_{t=2}^T Z_{t,T}\Big)^2
	= T^{-1}(T-1)\sum_{\i,\j}
	\bSigma_{\i,\j}\,
	\bgamma_{\i}^\top\,\bGamma_0(\calJ_{\i};\calJ_{\j})\,\bgamma_{\j}.
\end{align*}
	Then from Lemma \ref{lemma:autocov} (under \textsf{Regime 1}), or from Lemma \ref{lemma:main} (under \textsf{Regime 2 or 3}), we have $V_{TT}s_{TT}^{-1}\overset{p}{\rightarrow} 1$.
	Finally, for all $\epsilon>0$, we can show using the moment condition of $[\E_t]_{\i}$, $$\lim_{T\rightarrow \infty} s_{TT}^{-2}\sum_{t=2}^T\EE \big(Z_{t,T}^2\cdot 1(|Z_{t,T}|\geq \epsilon s_{TT}\big| \calA_{t-1})\big) = 0. $$
	So we conclude from Theorem 5.3.4 in \cite{fuller2009introduction}, 
	\begin{align*}
		T^{-1/2}\vecstack\{\Y_{\i}^\top\bnu_{\i}:\i\in\calS\}\implies N(0,\C).
	\end{align*}
	From Lemma \ref{lemma:autocov} (under \textsf{Regime 1}), or from Lemma \ref{lemma:main} (under \textsf{Regime 2 or 3}), we have $\frac{1}{T}\Y_{\i}^\top\Y_{\i}\overset{p}{\rightarrow}\bGamma_0(\calJ_{\i};\calJ_{\i})$. So we conclude
	\begin{align*}
		\sqrt{T}\vecstack\{\hat\m_{\i} - \m_{\i}:\i\in\calS\}
		\Rightarrow N(0,\D^{-1}\C\D^{-1}).
	\end{align*}
	where $\D = \bdiag\{\bGamma_0(\calJ_{\i};\calJ_{\i}):\i\in\calS\}$.
\end{proof}

\subsection{Proof of Theorem 3}
The proof of this theorem relies on Lemma \ref{lemma:matrix-perturbation}. 
Since $\hat\M_{R}^{\blk}$ is the best rank $R$ approximation of $\hat\M^{\blk}$. And we have $\fro{\hat\M^{\blk} - \M^{\blk}}\rightarrow 0$ as $T\rightarrow\infty$, so the condition in Lemma \ref{lemma:matrix-perturbation} is satisfied for sufficiently large $T$. 
So we conclude
\begin{align*}
	\hat\M_{R}^{\blk} - \M^{\blk}
	&= \hat\M^{\blk} - \M^{\blk}
	- \U_{\perp}^{\blk}\U_{\perp}^{\blk\top}(\hat\M^{\blk} - \M^{\blk}) \\
	&\quad \times \V_{\perp}^{\blk}\V_{\perp}^{\blk\top} + \R,
\end{align*}
where $\fro{\R}\leq \frac{40\fro{\hat\M^{\blk} - \M^{\blk}}^2}{\lambda_{\min}(\M^{\blk})}$, and thus $\sqrt{T}\fro{\R} = o_p(1)$. Next we vectorize both sides and we obtain 
\begin{align*}
	\vec(\hat\M_{R}^{\blk} - \M^{\blk})
	&=\big(\I- \V_{\perp}^{\blk}\V_{\perp}^{\blk\top}
	\otimes\U_{\perp}^{\blk}\U_{\perp}^{\blk\top}\big) \\
	&\quad \times\vec(\hat\M^{\blk} - \M^{\blk}) + \vec(\R).
\end{align*}
Together with Theorem 2, we obtain the result.

\subsection{Proof of Theorem 4}\label{sec:proof-thm4}
\begin{proof}

The result is an immediate consequence of the first part of Lemma~\ref{lemma:main}. Indeed, Lemma~\ref{lemma:main}
implies that, under Assumption~1,
\[
\max_{\i,\j\in\calS}
\left|[\hat\bGamma_0]_{\i,\j}-[\bGamma_0]_{\i,\j}\right|
=
O_p\left\{
\left(\frac{\log(M\vee N\vee T)}{T}\right)^{1/2}
\mu_{2q}^2\ks(1-\delta)^{-4}
\right\}.
\]
Since \(\mu_{2q}\) and \(\delta\) are treated as constants, and since
\[
\linf{\hat\bGamma_0-\bGamma_0}
=
\max_{\i,\j\in\calS}
\left|[\hat\bGamma_0]_{\i,\j}-[\bGamma_0]_{\i,\j}\right|,
\]
we obtain
\[
\linf{\hat\bGamma_0-\bGamma_0}
=
O_p\left\{
\left(\frac{\log(M\vee N\vee T)}{T}\right)^{1/2}\ks
\right\}.
\]
This proves Theorem~4.
\end{proof}

\subsection{Proof of Theorem 5}
\begin{proof}
For each $\i\in\calS$, recall we have $\calJ_{\i}[k_0] = \calJ_{\i}$.
And $\{\hat k_{\i}\neq k_0\} = \{\hat k_{\i} < k_0\}\cup \{\hat k_{\i}>k_0\}.$
For $k< k_0$, we have 
\begin{align*}
	\rss_{\i}(k) &= \z_{\i}^\top (\I-\Y_{\i}[k](\Y_{\i}[k]^\top\Y_{\i}[k])^{-1}\Y_{\i}[k]^\top)\z_{\i}\\
	&= (\Y_{\i}[k_0]\m_{\i} + \bnu_{\i})^\top
	(\I-\Y_{\i}[k](\Y_{\i}[k]^\top\Y_{\i}[k])^{-1}\Y_{\i}[k]^\top)\\
	&\quad \times(\Y_{\i}[k_0]\m_{\i} + \bnu_{\i}).
\end{align*}
Since $k<k_0$, $\Y_{\i}[k]$ is a submatrix of $\Y_{\i}[k_0]$, and we can split
$\Y_{\i}[k_0]\m_{\i}$ as
\[
\Y_{\i}[k_0]\m_{\i} \;=\; \Y_{\i}[k]\b_1 \;+\; \S\b_2,
\]
where $\|\b_2\|_{\ell_2}=\|\M_{\i}(\calJ_{\i}\backslash \calJ_{\i}[k])\|_{\mathrm F}$.
Let $\calI_1 := \calJ_{\i}[k]$ and $\calI_2 := \calJ_{\i}\backslash \calJ_{\i}[k]$.
Define
\begin{align*}
	\hat\G := \frac{1}{T-1}
	\begin{bmatrix}
		\Y_{\i}[k]^\top\Y_{\i}[k] & \Y_{\i}[k]^\top\S\\
		\S^\top\Y_{\i}[k] & \S^\top\S
	\end{bmatrix}
	\;:=\;
	\begin{bmatrix}
		\hat\bGamma_0(\calI_1;\calI_1) & \hat\bGamma_0(\calI_1;\calI_2)\\
		\hat\bGamma_0(\calI_2;\calI_1) & \hat\bGamma_0(\calI_2;\calI_2)
	\end{bmatrix}.
\end{align*}
Then we denote
\begin{align*}
	\G := \EE \hat\G = \mat{\bGamma_0(\calI_1; \calI_1)& \bGamma_0(\calI_1; \calI_2)\\ \bGamma_0(\calI_2; \calI_1)& \bGamma_0(\calI_2; \calI_2)}.
\end{align*}
Since $\big(\I-\Y_{\i}[k](\Y_{\i}[k]^\top\Y_{\i}[k])^{-1}\Y_{\i}[k]^\top\big)\Y_{\i}[k_0] = 0$, we have 
\begin{align*}
	\rss_{\i}(k) = (\S\b_2+ \bnu_{\i})^\top (\I-\Y_{\i}[k](\Y_{\i}[k]^\top\Y_{\i}[k])^{-1}\Y_{\i}[k]^\top)(\S\b_2+ \bnu_{\i}).
\end{align*}
On the other hand, we have 
\begin{align*}
	\rss_{\i}(k_0) = \bnu_{\i}^\top (\I-\Y_{\i}[k_0](\Y_{\i}[k_0]^\top\Y_{\i}[k_0])^{-1}\Y_{\i}[k_0]^\top)\bnu_{\i}.
\end{align*}
Therefore,
\begin{align*}
	\rss_{\i}(k) - \rss_{\i}(k_0)
	&\geq \b_2^\top\S^\top
	(\I-\Y_{\i}[k](\Y_{\i}[k]^\top\Y_{\i}[k])^{-1}\Y_{\i}[k]^\top)\S\b_2 \\
	&\quad + 2\b_2^\top\S^\top
	(\I-\Y_{\i}[k](\Y_{\i}[k]^\top\Y_{\i}[k])^{-1}\Y_{\i}[k]^\top)\bnu_{\i}. 
\end{align*}
We denote $\hat\M = \hat\G^{-1}$, $\M = \G^{-1}$. And $\hat\M$ is partitioned with the same shape as $\hat\G$, $\hat\M = \mat{\hat\M_{1,1}& \hat\M_{1,2}\\ \hat\M_{2,1}&\hat\M_{2,2}}$. 
For the first term, we have
\begin{align*}
	\S^\top(\I-\Y_{ij}[k](\Y_{ij}[k]^\top\Y_{ij}[k])^{-1}
	\Y_{ij}[k]^\top)\S
	= (T-1)(\hat\M_{2,2})^{-1}. 
\end{align*}
And 
\begin{align}\label{hatM-M}
	\op{\hat\M - \M} = \op{\hat\G^{-1} - \G^{-1}} \leq \op{\hat\G^{-1}}\cdot\op{\hat\G - \G}\cdot\op{\G^{-1}}. 
\end{align}
Notice from Lemma \ref{lemma:main}, with probability tending to 1, 
\begin{align*}
	\op{\hat\G-\G}
	&\leq |\calJ_{\i}|\cdot\linf{\hat\G -\G} \\
	&= O_p\big((T^{-1}\log(M\vee N\vee T))^{1/2}
	|\calJ_{\i}|\ks\mu_{2q}^2(1-\delta)^{-4}\big).
\end{align*}
Since $\G$ is a submatrix of $\bGamma_0$, we have $\lambda_{\min}(\G)\geq \lambda_{\min}(\bGamma_0)\geq \kappa_1$. 
Also, as long as $(T^{-1}\log(M\vee N\vee T))^{1/2}\mu_{2q}^2\knb\ks(1-\delta)^{-4}\lesssim \kappa_1$, we have $\lambda_{\min}(\hat\G)\geq \lambda_{\min}(\G) - \frac{\kappa_1}{2}\geq \frac{\kappa_1}{2}$. 
Therefore $\op{\hat\M - \M}\leq c_0\kappa_1^{-1}$ for some $c_0\in(0,1)$. So
\begin{align*}
	&\lambda_{\min}\{\S^\top(\I-\Y_{\i}[k](\Y_{\i}[k]^\top\Y_{\i}[k])^{-1}
	\Y_{\i}[k]^\top)\S\} \\
	&\qquad = T\lambda_{\min}\big((\hat\M_{2,2})^{-1}\big)
	= T\lambda_{\max}^{-1}(\hat\M_{2,2}). 
\end{align*}
Since $\hat\M_{2,2}$ is a submatrix of $\hat\M$, we have $\lambda_{\max}(\hat\M_{2,2})\leq \lambda_{\max}(\hat\M)\leq \lambda_{\max}(\M)+\op{\hat\M - \M}\leq \kappa_1^{-1}+c_0\kappa_1^{-1}$. 
And thus
\begin{align*}
	\lambda_{\min}(\S^\top(\I-\Y_{\i}[k](\Y_{\i}[k]^\top\Y_{\i}[k])^{-1}\Y_{\i}[k]^\top)\S)
	\geq (1+c_0)^{-1}T\kappa_1,
\end{align*} 
which also implies
\begin{align*}
	\b_2^\top\S^\top(\I-\Y_{\i}[k](\Y_{\i}[k]^\top\Y_{\i}[k])^{-1}\Y_{\i}[k]^\top)\S\b_2\geq (1+c_0)^{-1}T\kappa_1\ltwo{\b_2}^2. 
\end{align*}
On the other hand, we have from Cauchy-Schwarz inequality and second part in Lemma \ref{lemma:main}, 
\begin{align*}
	&\quad2\b_2^\top\S^\top(\I-\Y_{ij;k}(\Y_{ij;k}^\top\Y_{ij;k})^{-1}\Y_{ij;k}^\top)\bepsilon_{ij} \\
	&\leq \frac{1}{2}(1+c_0)^{-1}T\kappa_1\ltwo{\b_2}^2+ C_1\kappa_1^{-1}\mu_{2q}^4(1-\delta)^{-4}\ks^2\log(M\vee N\vee T). 
\end{align*}
From Lemma \ref{lemma:rss}, we have $\rss_{\i}(k_0)\lesssim T\mu_{2q}^2$ with probability tending to 1. 
Therefore, for sufficiently large $T$, we have 
\begin{align*}
	\rss_{\i}(k) - \rss_{\i}(k_0)
	&\geq  \frac{1}{2}(1+c_0)^{-1}T\kappa_1\ltwo{\b_2}^2 \\
	&\quad - C_1\kappa_1^{-1}\mu_{2q}^4(1-\delta)^{-4}
	\ks^2\log(M\vee N\vee T)\\
	&\geq c_1\rss_{\i}(k_0)\frac{\kappa_1}{\mu_{2q}^2}\ltwo{\b_2}^2
\end{align*}
for some $c_1\in(0,1)$. And using $\log(1+x)\geq \frac{1}{2}x$ for $x\in(0,2)$, we have
\begin{align*}
	\log \rss_{\i}(k) -\log \rss_{\i}(k_0) \geq \min\{c_2\frac{\kappa_1}{\mu_{2q}^2}\ltwo{\b_2}^2, \log 3\}. 
\end{align*}
Therefore for sufficiently large $T$, we have 
\begin{align*}
	\bic_{\i}(k) - \bic_{\i}(k_0)
	&\geq \min\{c_2\frac{\kappa_1}{\mu_{2q}^2}\ltwo{\b_2}^2, \log 3\} \\
	&\quad - D_0\frac{|\calJ_{\i}|}{T}\log(M\vee N\vee T) > 0,
\end{align*}
under the margin condition that $$\ltwo{\b_2}= \fro{\M_{\i}(\calJ_{\i}\backslash \calJ_{\i}[k])}\geq\fro{\M_{\i}(\calJ_{\i}\backslash \calJ_{\i}[k_0-1])}.$$

Now we show, for all $\i\in\calS$ and $k>k_0$, $\PP(\bic_{\i}(k)<\bic_{\i}(k_0))\rightarrow 0$. We have
\begin{align*}
	\rss_{\i}(k) = \inf_{\bgamma}\ltwo{\z_{\i} - \Y_{\i}[k]\bgamma}^2.
\end{align*}
We can write $\Y_{\i}[k]\bgamma = \Y_{\i}[k_0]\bgamma_1 + \S\bgamma_2$. Denote
$$
\tilde \S = \big(I - \Y_{\i}[k_0]\big(\Y_{\i}[k_0]^\top\Y_{\i}[k_0]\big)^{-1}\Y_{\i}[k_0]^\top\big)\S.
$$
Then
\begin{align*}
	\rss_{\i}(k) &= \inf_{\bgamma}\ltwo{\z_{\i} - \Y_{\i}[k]\bgamma}^2
	=  \inf_{\bgamma_1,\bgamma_2}\ltwo{\z_{\i} - \Y_{\i}[k_0]\bgamma_1- \S\bgamma_2}^2\\
	&=\inf_{\bgamma_1,\bgamma_2}\ltwo{\z_{\i} - \Y_{\i}[k_0]\Big(\bgamma_1 + \big(\Y_{\i}[k_0]^\top\Y_{\i}[k_0]\big)^{-1}\Y_{\i}[k_0]^\top\bgamma_2\Big)- \tilde\S\bgamma_2}^2\\
	&= \inf_{\bgamma_1,\bgamma_2}\ltwo{\z_{\i} - \Y_{\i}[k_0]\bgamma_1- \tilde\S\bgamma_2}^2.
\end{align*}
An explicit minimizer is
$$
\bgamma_1 = \big(\Y_{\i}[k_0]^\top\Y_{\i}[k_0]\big)^{-1}\Y_{\i}[k_0]^\top\z_{\i},
\qquad
\bgamma_2 = \big(\tilde\S^\top\tilde\S\big)^{-1}\tilde\S^\top\z_{\i}.
$$
Therefore,
\begin{align*}
	\rss_{\i}(k) = \rss_{\i}(k_0) - \ltwo{\tilde\S\big(\tilde\S^\top\tilde\S\big)^{-1}\tilde\S^\top\bnu_{\i}}^2.
\end{align*}
Next we bound $\ltwo{\tilde\S(\tilde\S^\top\tilde\S)^{-1}\tilde\S^\top\bnu_{\i}}^2$:
\begin{align*}
	\ltwo{\tilde\S(\tilde\S^\top\tilde\S)^{-1}\tilde\S^\top\bnu_{\i}}^2
	&= \bnu_{\i}^\top\tilde\S(\tilde\S^\top\tilde\S)^{-1}\tilde\S^\top\bnu_{\i}
	\leq \lambda_{\max}\big((\tilde\S^\top\tilde\S)^{-1}\big)\ltwo{\tilde\S^\top\bnu_{\i}}^2.
\end{align*}
Since $\tilde\S^\top\bnu_{\i}\in\RR^{|\calJ_{\i}[k]| - |\calJ_{\i}|}$, we have
$$
\ltwo{\tilde\S^\top\bnu_{\i}}^2\leq \big(|\calJ_{\i}[k]| - |\calJ_{\i}|\big)\,\linf{\tilde\S^\top\bnu_{\i}}^2.
$$
Recall $\tilde\S^\top\bnu_{\i} = \S^\top\!\big(\I - \Y_{\i}[k_0](\Y_{\i}[k_0]^\top\Y_{\i}[k_0])^{-1}\Y_{\i}[k_0]^\top\big)\bnu_{\i}$. Let an arbitrary row of $\S^\top$ be $\x^\top$. Then
\begin{align}\label{x-perp-e}
	&\quad\x^\top\!\big(\I - \Y_{\i}[k_0](\Y_{\i}[k_0]^\top\Y_{\i}[k_0])^{-1}\Y_{\i}[k_0]^\top\big)\bnu_{\i}\notag\\
	&= \x^\top\bnu_{\i} - \x^\top\Y_{\i}[k_0](\Y_{\i}[k_0]^\top\Y_{\i}[k_0])^{-1}\Y_{\i}[k_0]^\top\bnu_{\i}.
\end{align}
From Lemma \ref{lemma:main}, $\x^\top\bnu_{\i} = O_p\big((T\log(M\vee N\vee T))^{1/2}\mu_{2q}^2\ks^{1/2}(1-\delta)^{-3/2}\big)$. Moreover,
\begin{align*}
	&\quad\x^\top\Y_{\i}[k_0](\Y_{\i}[k_0]^\top\Y_{\i}[k_0])^{-1}\Y_{\i}[k_0]^\top\bnu_{\i}\\
	&= \x^\top\Y_{\i}[k_0]\Big((\Y_{\i}[k_0]^\top\Y_{\i}[k_0])^{-1} - (\EE\,\Y_{\i}[k_0]^\top\Y_{\i}[k_0])^{-1} \Big)\Y_{\i}[k_0]^\top\bnu_{\i}\\
	&\qquad + \x^\top\Y_{\i}[k_0](\EE\,\Y_{\i}[k_0]^\top\Y_{\i}[k_0])^{-1}\Y_{\i}[k_0]^\top\bnu_{\i}.
\end{align*}
Since $\tfrac{1}{T}\EE\,\Y_{\i}[k_0]^\top\Y_{\i}[k_0]$ is a submatrix of $\bGamma_0$, we have
\begin{align*}
	\lambda_{\max}\big((\EE\,\Y_{\i}[k_0]^\top\Y_{\i}[k_0])^{-1}\big)
	= \lambda_{\min}^{-1}(\EE\,\Y_{\i}[k_0]^\top\Y_{\i}[k_0])
	\leq T^{-1}\kappa_1^{-1}. 
\end{align*}
Therefore, from Lemma \ref{lemma:main} and $\max_{\i}[\bGamma_0]_{\i;\i}\leq \|[\E_{t}]_{\i}\|_{2q}^2\leq \mu_{2q}^2$, we have 
\begin{align*}
	&\quad|\x^\top\Y_{\i}[k_0](\EE\,\Y_{\i}[k_0]^\top\Y_{\i}[k_0])^{-1} \Y_{\i}[k_0]^\top\bnu_{\i}|\\
	&\lesssim (T\log(M\vee N\vee T))^{1/2}\mu_{2q}^2\,|\calJ_{\i}|\,\ks^{1/2}(1-\delta)^{-2}\,\kappa_1^{-1}\mu_{2q}^2.
\end{align*}
Next we bound $\op{(\Y_{\i}[k_0]^\top\Y_{\i}[k_0])^{-1} - (\EE\,\Y_{\i}[k_0]^\top\Y_{\i}[k_0])^{-1}}$ similarly as in \eqref{hatM-M}:
\begin{align*}
	&\op{(\Y_{\i}[k_0]^\top\Y_{\i}[k_0])^{-1} - (\EE\,\Y_{\i}[k_0]^\top\Y_{\i}[k_0])^{-1}} \\
	&\hspace{3cm}\lesssim T^{-1}(T^{-1}\log(M\vee N\vee T))^{1/2}\mu_{2q}^2\,|\calJ_{\i}|\,\ks(1-\delta)^{-4}\kappa_1^{-2}. 
\end{align*}
Thus
\begin{align*}
	&\quad \x^\top\Y_{\i}[k_0]\Big((\Y_{\i}[k_0]^\top\Y_{\i}[k_0])^{-1} - (\EE\,\Y_{\i}[k_0]^\top\Y_{\i}[k_0])^{-1} \Big) \Y_{\i}[k_0]^\top\bnu_{\i} \\
	&\lesssim T^{-1}(T^{-1}\log(M\vee N\vee T))^{1/2}\mu_{2q}^2\,|\calJ_{\i}|\,\ks(1-\delta)^{-3}\kappa_1^{-2} \cdot |\calJ_{\i}|\cdot T\mu_{2q}^2\\
	&\hspace{5cm}\cdot (T\log(M\vee N\vee T))^{1/2}\mu_{2q}^2\,\ks^{1/2}(1-\delta)^{-2}\\
	&= \ks^{3/2}\,|\calJ_{\i}|^2\,\log(M\vee N\vee T)\,\mu_{2q}^6\,\kappa_1^{-2}(1-\delta)^{-6}. 
\end{align*}
Returning to \eqref{x-perp-e}, we conclude
\begin{align*}
	\linf{\tilde\S^\top\bnu_{\i}}
	= O_p\big((T\log(M\vee N\vee T))^{1/2}\mu_{2q}^4\,
	|\calJ_{\i}|\,\ks^{1/2}(1-\delta)^{-2}\kappa_1^{-1}\big).
\end{align*}
Hence
\begin{align*}
	\rss_{\i}(k_0) - \rss_{\i}(k)
	&\lesssim  \big(|\calJ_{\i}[k]| - |\calJ_{\i}|\big)\,\log(M\vee N\vee T)\,
	\mu_{2q}^8\,|\calJ_{\i}|^2\\
	&\quad \times \ks(1-\delta)^{-4}\,\kappa_1^{-3}. 
\end{align*}
From Lemma \ref{lemma:rss}, $\rss_{\i}(k)\gtrsim T\mu_{2q}^2$. Therefore
\begin{align*}
	\frac{\rss_{\i}(k_0) - \rss_{\i}(k)}{\rss_{\i}(k)}
	&\lesssim\frac{|\calJ_{\i}[k]| - |\calJ_{\i}|}{T}\,\log(M\vee N\vee T)\,
	\mu_{2q}^6\,|\calJ_{\i}|^2\\
	&\quad \times \ks(1-\delta)^{-4}\,\kappa_1^{-3}.
\end{align*}
Using $\log(1+x)\leq x$ for $x>0$, we have 
\begin{align*}
	&\quad\log\rss_{\i}(k_0) - \log\rss_{\i}(k)\\
	&\lesssim\frac{|\calJ_{\i}[k]| - |\calJ_{\i}|}{T}\,\log(M\vee N\vee T)\,
	\mu_{2q}^6\,|\calJ_{\i}|^2\\
	&\quad \times\ks(1-\delta)^{-4}\,\kappa_1^{-3}.
\end{align*}
Therefore 
\begin{align*}
	\bic_{\i}(k) - \bic_{\i}(k_0) &\geq D_0\,\frac{|\calJ_{\i}[k]| - |\calJ_{\i}|}{T}\,\log(M\vee N\vee T) \\
	&\quad+ \log\rss_{\i}(k) - \log\rss_{\i}(k_0)\\
	&> 0,
\end{align*}
as long as we choose $D_0\gtrsim \mu_{2q}^6\,|\calJ_{\i}|^2\,\ks\,(1-\delta)^{-4}\,\kappa_1^{-3}$. 

\end{proof}

\section{Auxiliary Lemmas}\label{sec:aux-lemmas}
The auxiliary results in this section are used to support the proofs in Section~\ref{sec:main-proofs}.
Lemma~\ref{lemma:concentration} provides the functional-dependence concentration inequality that is used
as the main probabilistic tool. Lemma~\ref{lemma:main} applies this concentration result to obtain
entrywise bounds for the sample auto-covariance matrix and the noise-design cross
terms; these bounds are used in the proofs of Theorems~2, 4, and 5. Lemma~\ref{lemma:rss}
controls the residual sum of squares over candidate neighborhoods and is used in
the proof of Theorem~5. Lemma~\ref{lemma:product-sparse-matrix} records a locality-induced sparsity property of
powers of the autoregressive coefficient matrix, which is used in the proof of
Lemma~\ref{lemma:main}. Lemma~\ref{lemma:matrix-perturbation} is the perturbation result for the best rank-\(R\) projection
and is used in the proof of Theorem~3. Finally, Lemma~\ref{lemma:autocov} provides the
fixed-dimensional law of large numbers used under Regime~1 in the proof of
Theorem~2.

We start with some notations. Let $\bGamma\in\RR^{MN\times MN}$. For multi-indices
$\i=(i_1,i_2),\j=(j_1,j_2)\in\calS = [M]\times[N]$, define
\begin{align}\label{def:multi-idx}
    [\bGamma]_{\i,\j} := [\bGamma]_{(i_2-1)M+i_1, (j_2-1)M+j_1}, 
\end{align}
where the last equality uses column-major linear indexing.
The sub-matrix $\bGamma(\calJ_1; \calJ_2)$ is then formed by extracting the elements of $\bGamma$ corresponding to the multi-indices in $\calJ_1$ and $\calJ_2$, arranged using column-major ordering.

Following \cite{wu2005nonlinear}, we introduce the following functional dependent measure.
Let $z_i$ be a stationary process of the form $z_i = g(\calF_i)$, where $g$ is a measurable function and $\calF_i = (\cdots, e_{-1}, e_0,e_1,\cdots)$ with independent and identically distributed random variables $\{e_i\}$. 
The functional dependent measure is defined as 
\begin{align*}
	\theta_{i,q} = \|z_i - z_i^*\|_q =  \|g(\calF_i) - g(\calF_i^*)\|_q,
\end{align*}
where $z_i^* = g(\calF_i^*)$ is the coupled process of $z_i$, $\calF_i^* = (\cdots, e_{-1}, e_0^*,e_1,\cdots)$ with $\{e_0,e_0^*\}$ being independent and identically distributed. Here $\|\cdot\|_q:=(\EE|\cdot|)^{1/q}$ for some $q\geq1$. Then we have the following from Theorem 2 in \cite{liu2013probability} that plays the crucial role in our proof. 

\begin{lemma}\label{lemma:concentration}
	Let $S_n = \sum_{i=1}^n z_n$ and $\Theta_{m,q} = \sum_{i=m}^{\infty}\theta_{i,q}$. Assume for each $m$, $\Theta_{m,q} = O(m^{-\alpha})$ with $\alpha > \frac{1}{2}-\frac{1}{q}$ and $q>2$. Then we have for all $x>0$, 
	\begin{align*}
		\PP(|S_n|\geq x) \leq C_1\frac{\Theta_{0,q}^qn}{x^q} + C_3\exp(-C_2\frac{x^2}{\Theta_{0,q}^2n}),
	\end{align*}
	for $C_1,C_2,C_3>0$ depending only on $q$. 
\end{lemma}

Recall $[\hat\bGamma_0]_{\i,\j} = \frac{1}{T}\sum_{t= 1}^T [\X_t]_{\i}[\X_t]_{\j}$ and thus $\bGamma_0 = \EE\hat\bGamma_0$. 
We denote $$\x_{\i} = ([\X_1]_{\i},\cdots, [\X_{T-1}]_{\i})^\top,$$ and recall $\bnu_{\i} = ([\E_2]_{\i},\cdots, [\E_T]_{\i})^\top$.  
We first establish the following concentration inequality.

\begin{lemma}\label{lemma:main}
	Under Assumption 1, there exists absolute constant $C>0$, such that 
	\begin{enumerate}
		\item For $\i,\j\in\calS$, we have for $x>0$, 
\begin{align*}
		\PP\bigg(\big|[\hat\bGamma_0]_{\i,\j} - [\bGamma_0]_{\i,\j}\big|\geq x\bigg)
		&\lesssim CT\bigg(\frac{\mu_{2q}^{2}\ks(1-\delta)^{-4}}{Tx}\bigg)^q \\
		&\quad + C\exp\bigg(-\frac{Tx^2}{(\mu_{2q}^2\ks(1-\delta)^{-4})^2}\bigg). 
\end{align*}
		And this leads to 
		\begin{align*}
			\max_{\i,\j\in\calS} \big|[\hat\bGamma_0]_{\i,\j} - [\bGamma_0]_{\i,\j}\big| 
			= O_{p}\big((T^{-1}\log(M\vee N\vee T))^{1/2}\mu_{2q}^2\ks(1-\delta)^{-4}\big).
		\end{align*}
		\item For $\i,\j\in\calS$, we have for $x>0$, 
\begin{align*}
		\PP(\big|\bnu_{\i}^\top\x_{\j} \big|\geq x)
		&\leq CT\left(\frac{\mu_{2q}^2\ks^{1/2}(1-\delta)^{-2}}{x}\right)^q \\
		&\quad + C\exp\left(-\frac{x^2}{T\cdot (\mu_{2q}^2\ks^{1/2}(1-\delta)^{-2})^2}\right).
\end{align*}
		And this leads to 
\begin{align*}
		\max_{\i,\j\in\calS} \big|\bnu_{\i}^\top\x_{\j} \big|
		=O_{p}\big((T\log(M\vee N\vee T))^{1/2}
		\mu_{2q}^2\ks^{1/2}(1-\delta)^{-2}\big).
\end{align*}
	\end{enumerate}
\end{lemma}
\begin{proof}
	We shall use Lemma \ref{lemma:concentration} to prove this. Denote $\mu_q = \max_{\i}\|[\E_0]_{\i}\|_q$. 	
	Using innovation representation, we have
	\begin{align*}
		\x_t = \sum_{l=0}^{\infty} \M^l\bepsilon_{t-l}.
	\end{align*}
	As a result, for each $\i = (i_1,i_2)\in\calS$, 
	\begin{align*}
		&\quad (\e_{i_2}\otimes\e_{i_1})^\top\x_t = \sum_{l=0}^{\infty} (\e_{i_2}\otimes\e_{i_1})^\top\M^l\bepsilon_{t-l} \\
        &=  \sum_{l=0}^{\infty} \sum_{\j}(\e_{i_2}\otimes\e_{i_1})^\top\M^l(\e_{j_2}\otimes\e_{j_1})(\e_{j_2}\otimes\e_{j_1})^\top\bepsilon_{t-l}. 
	\end{align*}
%	Recall $\M = \mat{\m_{11}^\top\\ \m_{21}^\top\\ \vdots \\ \m_{MN}^\top}$, where $\m_{ij} = \vec(\M_{ij})$ is supported on the neighborhood of $(i,j)$. 
Following the definition of $\ks$ and Lemma \ref{lemma:product-sparse-matrix}, we have $\lzero{(\e_{i_2}\otimes\e_{i_1})^\top\M^l}\leq \min\{\ks l^2, MN\}$.
And thus 
\begin{align}\label{EX2q}
	\max_{\i}\|(\e_{i_2}\otimes\e_{i_1})^\top\x_t\|_{2q}
	&\leq  \max_{\i}\left\|\sum_{l=0}^{\infty} \sum_{\j}
	(\e_{i_2}\otimes\e_{i_1})^\top\M^l(\e_{j_2}\otimes\e_{j_1})\right.\notag\\
	&\qquad\qquad\left.\times
	(\e_{j_2}\otimes\e_{j_1})^\top\bepsilon_{t-l}\right\|_{2q}\notag\\
	&\leq \max_{\i}\sum_{l=0}^{\infty} \sum_{\j}|[\M^l]_{\i,\j}|\cdot \|[\E_{t-l}]_{\j}\|_{2q}\notag\\
	&\leq \max_{\i}\sum_{l=0}^{\infty} \ks^{1/2} l\ltwo{(\e_{i_2}\otimes\e_{i_1})^\top\M^l}\cdot \mu_{2q}\notag\\
	&\leq \ks^{1/2}\mu_{2q}\sum_{l\geq 0} l\delta^l \leq C\ks^{1/2}\mu_{2q}(1-\delta)^{-2}\notag\\
	&< +\infty,
\end{align}
where in the second-to-last line, we use the fact $\op{\M^l}\leq \delta^l$.

On the other hand, we can bound $\|(\e_{i_2}\otimes\e_{i_1})^\top\x_t - (\e_{i_2}\otimes\e_{i_1})^\top\x_t^*\|_{2q}$ as follows using the innovation representation:
\begin{align}\label{eq:diffX}
	&\|(\e_{i_2}\otimes\e_{i_1})^\top\x_t
	 - (\e_{i_2}\otimes\e_{i_1})^\top\x_t^*\|_{2q} \notag\\
	&\quad= \left\|\sum_{\j}(\e_{i_2}\otimes\e_{i_1})^\top\M^t(\e_{j_2}\otimes\e_{j_1})
	(\e_{j_2}\otimes\e_{j_1})^\top(\bepsilon_{0} -\bepsilon_{0} ^*)\right\|_{2q}\notag\\
	&\leq C\mu_{2q} \ks^{1/2}t\delta^t. 
\end{align}
Now we can bound $\|[\X_t]_{\i}[\X_t]_{\j} - [\X_t]_{\i}^*[\X_t]_{\j}^*\|_{q}$ using \eqref{EX2q} and \eqref{eq:diffX}: 
\begin{align*}
	&\quad\|[\X_t]_{\i}[\X_t]_{\j} - [\X_t]_{\i}^*[\X_t]_{\j}^*\|_{q} \\
	&\leq \|(\e_{i_2}\otimes\e_{i_1})^\top\x_t\cdot
	(\e_{j_2}\otimes\e_{j_1})^\top(\x_t -\x_t^*)\|_{q} \\
	&\quad + \|(\e_{i_2}\otimes\e_{i_1})^\top(\x_t -\x_t^*)\cdot
	(\e_{j_2}\otimes\e_{j_1})^\top\x_t^*\|_{q}\\
	&\leq \|(\e_{i_2}\otimes\e_{i_1})^\top\x_t\|_{2q}
	\cdot\|(\e_{j_2}\otimes\e_{j_1})^\top(\x_t -\x_t^*)\|_{2q} \\
	&\quad + \|(\e_{j_2}\otimes\e_{j_1})^\top\x_t^*\|_{2q}
	\cdot\|(\e_{i_2}\otimes\e_{i_1})^\top(\x_t -\x_t^*)\|_{2q}\\
	&\leq C\mu_{2q}^2\ks(1-\delta)^{-2}t\delta^t. 
\end{align*}
Therefore
\begin{align*}
	\Theta_{m,q}
	&:= \sum_{t =m}^{\infty}\|[\X_t]_{\i}[\X_t]_{\j}
	-[\X_t]_{\i}^*[\X_t]_{\j}^*\|_{q} \\
	&\leq C\mu_{2q}^2\ks(1-\delta)^{-4}\delta^m = O(m^{-\alpha})
\end{align*}
for any $\alpha>1$. 

From Lemma \ref{lemma:concentration}, we conclude
\begin{align*}
	\PP\bigg( \big|[\hat\bGamma_0]_{\i,\j} - [\bGamma_0]_{\i,\j}\big|\geq x\bigg)
	&\lesssim T\bigg(\frac{\mu_{2q}^{2}\ks(1-\delta)^{-4}}{Tx}\bigg)^q  \\
	&\quad +\exp\bigg(-\frac{Tx^2}{(\mu_{2q}^2\ks(1-\delta)^{-4})^2}\bigg). 
\end{align*}
By setting $x = T^{-1/2}\mu_{2q}^2\ks(1-\delta)^{-4}(\log(M\vee N\vee T))^{1/2}$, we conclude
\begin{align*}
	\max_{\i,\j\in\calS} \big|[\hat\bGamma_0]_{\i,\j} - [\bGamma_0]_{\i,\j}\big|
	= O_{p}\big(T^{-1/2}\mu_{2q}^2\ks(1-\delta)^{-4}
	(\log(M\vee N\vee T))^{1/2}\big).
\end{align*}

Similarly, we can show
\begin{align*}
	\PP\big(\big|\bnu_{\i}^\top\x_{\j} \big|\geq x\big)
	&\lesssim T\left(\frac{\mu_{2q}^2\ks^{1/2}(1-\delta)^{-2}}{x}\right)^q \\
	&\quad + \exp\left(-\frac{x^2}{T\cdot (\mu_{2q}^2\ks^{1/2}(1-\delta)^{-2})^2}\right),
\end{align*}
and 
\begin{align*}
	\max_{\i,\j\in\calS} \big|\bnu_{\i}^\top\x_{\j} \big|
	= O_{p}\big((T\log(M\vee N\vee T))^{1/2}
	\mu_{2q}^2\ks^{1/2}(1-\delta)^{-2}\big).
\end{align*}

\end{proof}

\begin{lemma}\label{lemma:rss}
	Under Assumption 1, we have for each $\i$ and $k\geq k_0$, 
	\begin{align*}
		\max_{\i}\bigg|\rss_{\i}(k) - (T-1)\bSigma_{\i;\i}\bigg| = O_p\big((T\log(M\vee N\vee T))^{1/2}\mu_{2q}^2\big)
	\end{align*}
	holds as $T\rightarrow\infty$ and $\frac{\max_{\i}|\calJ_{\i}[K_0]|^2\cdot \ks^2\log(M\vee N\vee T)}{T}\leq \tilde c$ for some $\tilde c>0$ depending only on $\kappa_1,\delta, \mu_{2q}$.
\end{lemma}
\begin{proof}
When $k\geq k_0$, we can decompose $\rss_{\i}(k)$ as 
\begin{align*}
	\rss_{\i}(k) = \bnu_{\i}^\top\bnu_{\i} - \bnu_{\i}^\top\Y_{\i}[k] (\Y_{\i}[k]^\top\Y_{\i}[k])^{-1}\Y_{\i}[k]^\top\bnu_{\i}.
\end{align*}
We consider the event 
\begin{align*}
	\calE_1 = \bigg\{&\max_{\i,\j\in\calS} \big|[\hat\bGamma_0]_{\i,\j} - [\bGamma_0]_{\i,\j}\big| 
	\lesssim (T^{-1}\log(M\vee N\vee T))^{1/2}\mu_{2q}^2\ks(1-\delta)^{-4},\\
	& \max_{\i,\j\in\calS} \big|\bnu_{\i}^\top\x_{\j} \big|
	\lesssim \big(T\log(M\vee N\vee T)\big)^{1/2}
	\mu_{2q}^2\ks^{1/2}(1-\delta)^{-2}\bigg\}.
\end{align*}
From Lemma \ref{lemma:main}, $\PP(\calE_1)\rightarrow 0$. We shall continue our proof under $\calE_1$. 
From Lemma \ref{lemma:concentration}, we obtain 
\begin{align}\label{eij2}
	\max_{\i} \big|\bnu_{\i}^\top \bnu_{\i} - (T-1)\sigma_{\i}^2\big| = O_{p}\big((T\log(M\vee N\vee T))^{1/2}\mu_{2q}^2\big). 
\end{align}
	Recall $\calJ_{\i}[k]$ is the index set corresponding with the columns of $\Y_{ij;k}$, then for all $i,j$, under $\calE_1$, 
\begin{align*}
		&\quad (T-1)^{-1}\lambda_{\min}(\Y_{\i}[k]^\top\Y_{\i}[k]) \\
		&\geq \lambda_{\min}\big(\bGamma_0(\calJ_{\i}[k];\calJ_{\i}[k])\big) \\
		&\quad - \op{(T-1)^{-1}\Y_{\i}[k]^\top\Y_{\i}[k]
		 - \bGamma_0(\calJ_{\i}[k];\calJ_{\i}[k])}\\
		& \geq \kappa_1 - |\calJ_{\i}[k]|(T^{-1}\log(M\vee N\vee T))^{1/2}
		\mu_{2q}^2\ks(1-\delta)^{-4}\\
		&\geq \kappa_1(1-c_1),
\end{align*}
	for some $c_1\in(0,1)$ as long as
	\[
	|\calJ_{i}[k]|(T^{-1}\log(M\vee N\vee T))^{1/2}\mu_{2q}^2\ks(1-\delta)^{-4} \leq c_1\kappa_1.
	\]
	And under $\calE_1$,
	\[
	\linf{\Y_{\i}[k_0]^\top\bnu_{\i}}
	\lesssim (T\log(M\vee N\vee T))^{1/2}\mu_{2q}^2\ks^{1/2}(1-\delta)^{-2}.
	\]
	Therefore, with
	\[
	\Pi_{\i,k_0}=\Y_{\i}[k_0](\Y_{\i}[k_0]^\top\Y_{\i}[k_0])^{-1}\Y_{\i}[k_0]^{\top},
	\]
\begin{align*}
		|\bnu_{\i}^\top\Pi_{\i,k_0}\bnu_{\i}|
		&\lesssim |\calJ_{\i}|\ks(1-\delta)^{-4}\mu_{2q}^4
		\kappa_1^{-1}\log(M\vee N\vee T).
\end{align*}
	Together with \eqref{eij2}, and
	\[
	|\calJ_{\i}|\ks(1-\delta)^{-4}\mu_{2q}^4\kappa_1^{-1}\log(M\vee N\vee T)
	\lesssim (T\log(M\vee N\vee T))^{1/2}\mu_{2q}^2,
	\]
	we finish the proof. 
\end{proof}

\begin{lemma}\label{lemma:product-sparse-matrix}
%	Let $\M\in\RR^{MN\times MN}$ be such that 
%	\begin{align*}
%		\M = \sum_{\i\in\calS}(\e_{i_2}\otimes \e_{i_1})\vec(\M_{\i})^\top.
%	\end{align*}
%Let $\calJ_{\i} = \supp(\M_{\i})$ and $J = \max_{\i}|\calJ_{\i}|$. Then f
For any $p>0$, we have 
\begin{align*}
	\lzero{(\e_{i_2}\otimes \e_{i_1})^\top\M^p}\leq \min\{\ks p^2,MN\}. 
\end{align*}
\end{lemma}
\begin{proof}
We prove this for $p=2$ and it can be similarly extended to larger $p$. 
Recall all the support $\calJ_{\i}$ are included in the same rectangle 
$\{i_1-k_1,\cdots, i_1+k_1\}\times \{i_2-k_2,\cdots,i_2+k_2\}$
of size $\ks = (2k_1+1)(2k_2+1)$. We denote $\M$ by its columns: $\M = \mat{\tilde\m_{11}&\cdots&\tilde \m_{MN}}$. Then 
\begin{align*}
	(\e_{i_2}\otimes \e_{i_1})^\top\M^2(\e_{j_2}\otimes \e_{j_1}) = \m_{\i}^\top\m_{\j} = \sum_{\k}[\M]_{\i,\k}[\M]_{\k,\j}.
\end{align*}
The summation is non-vanishing only when $\calJ_{\i}\cap\calJ_{\j}\neq \emptyset$. When $|i_1-j_1| > 2k_1$ or $|i_2-j_2|>2k_2$, the intersection is empty. And thus there are at most $(4k_1+1)(4k_2+2)\leq 4\ks$ $\j$s making the term non-vanishing.
\end{proof}

%\begin{lemma}\label{lemma:product-sparse-matrix}
%	Let $\M_1,\M_2\in\RR^{MN\times MN}$ be such that 
%	\begin{align*}
%		\M_l = \mat{\m_{l,11}^\top\\ \vdots \\ \m_{l,MN}^\top}
%	\end{align*}
%	for $l = 1,2$. If the support of $\textsf{mat}(\m_{l,ij})$ is included in a rectangular:
%	\begin{align*}
%		\calJ_{l,ij}:=\supp(\textsf{mat}(\m_{l,ij}))\subset[1\vee(i-k_{l,1}),M\wedge(i+k_{l,1})]\times [1\vee(j-k_{l,2}),N\wedge(j+k_{l,2})].
%	\end{align*}
%	Then 
%	\begin{align*}
%		\lzero{(\e_j\otimes \e_i)^\top\M_1\M_2}\leq (2k_{1,1}+2k_{2,1}+1)(2k_{2,1}+2k_{2,2}+1). 
%	\end{align*}
%\end{lemma}
%\begin{proof}
%	Let 
%	\begin{align*}
%		\M_{l} = \mat{\tilde\m_{l,11}&\cdots &\tilde\m_{l,MN}}. 
%	\end{align*}
%	Then 
%	\begin{align*}
%		(\e_j\otimes \e_i)^\top\M_1\M_2(\e_t\otimes\e_s) &= \m_{1,ij}^\top\tilde\m_{2,st}\\
%		&=\sum_{i',j'}[\M_{1,ij}]_{i'j'}[\M_{2,i'j'}]_{st}\\
%		&=\sum_{(i',j')\in\calJ_{1,ij}}[\M_{1,ij}]_{i'j'}[\M_{2,i'j'}]_{st}. 
%	\end{align*}
%	If $|s-i|>k_{1,1}+k_{2,1}$ or $|t-j|>k_{1,2}+k_{2,2}$, then the above term vanishes since $(s,t)\notin\calJ_{2,i'j'}$. 
%\end{proof}

\begin{lemma}\label{lemma:matrix-perturbation}
	Let $\M\in\RR^{P_1\times P_2}$ be a rank $R$ matrix. Let $\M = \U\bSigma\V^\top$ be its compact singular value decomposition and denote $\lambda_{\min}>0$ the smallest non-zero singular value of $\M$. Then for any matrix $\hat\M$ such that $\fro{\hat\M - \M}\leq\frac{\lambda_{\min}}{8}$, we have that the best rank $R$ approximation of $\hat\M$ under Frobenius norm (denoted by $\svd_R(\hat\M)$) satisfies
	\begin{align*}
		\svd_R(\hat\M) - \M = \Z - \U_{\perp}\U_{\perp}^\top\Z\V_{\perp}\V_{\perp}^\top + \R, 
	\end{align*}
	where $\Z = \hat\M - \M$ and $\U_{\perp}$ $(\V_{\perp} \text{resp.})$ is the orthogonal complement of $\U$ $(\V \text{resp.})$, and the remainder term $\R$ satisfies 
	\begin{align*}
		\fro{\R} \leq 40\frac{\fro{\hat\M - \M}^2}{\lambda_{\min}}. 
	\end{align*}
\end{lemma}
\begin{proof}
	See the proof of Lemma 18 in \cite{shen2025computationally}. 
\end{proof}

\begin{lemma}\label{lemma:autocov}
	Under \textsf{Regime 1}, suppose $\rho(\M)<1$. And let $\E_t$ be i.i.d. with mean zero and $\EE|[\E_t]_{\i}|^{2+\tilde q}<+\infty$ for each $\i\in\calS$ and some $\tilde q>0$.
	Then we have 
	\begin{align*}
		\frac{1}{T-1}\sum_{t=1}^{T-1}\vec(\X_t)\vec(\X_t)^\top \overset{p}{\rightarrow} \bGamma_0
	\end{align*} 
	as $T\rightarrow\infty$. 
\end{lemma}
\begin{proof}
	See Theorem 8.2.3 in \cite{fuller2009introduction}. 
\end{proof}

\section{Additional Methodological and Numerical Details}\label{sec:supp-additional-details}
This section collects additional methodological details and numerical results.

\subsection{Non-identifiability under non-nested neighborhoods}\label{sec:non-nested}
In this section, we give the formal statement of the ambiguity for neighborhood selection within non-nested candidates.
% The following proposition underscores the challenge of distinguishing between $\calJ$ and $\tilde\calJ$.
\begin{proposition}\label{thm:non-nested}
	Let $K_{1}> 0, K_{2}\geq 0$. Then there exists a distribution $\calP$, and $\X_t\overset{\text{i.i.d.}}{\sim} \calP$, such that there exists $k_1<K_{1}, k_2>K_{2}$, and $\hat\M\in\RR^{\calS}$, the following holds:
	\begin{enumerate}[label*=$\bullet$]
		\item $(2k_1+1)(2k_2+1)<(2K_{1}+1)(2K_{2}+1)$,
		\item $\supp(\hat\M) = \tilde\calJ$,
		\item $y_t = \inp{\X_t}{\hat\M}$. 
	\end{enumerate}
\end{proposition}
\begin{proof}
	We set $k_1 = 0, k_2 = K_{2}+1$ (see Figure \ref{fig:non-nested}). 
	Let $\X$ be supported on only two entries $\{(i_0-K_{1},j_0), (i_0, j_0-K_{2}-1)\}$ and that $\X_t(i_0-K_{1},j_0) = \X_t(i_0, j_0-K_{2}-1)\sim\calP_0$, where $\calP_0$ can be an arbitrary distribution.  
	And let $\hat\M= \M(i_0-K_{1}, j_0)\e_{i_0}\e_{j_0-K_{2}-1}^\top$. Then we have $y_t = \inp{\X_t}{\M} = \X_t(i_0-K_{1}, j_0)\cdot\M(i_0-K_{1}, j_0) = \inp{\X_t}{\hat\M}$. 
\end{proof}
\begin{figure}[h] % h stands for here; t for top; b for bottom
	\centering
	\includegraphics[width=0.7\textwidth]{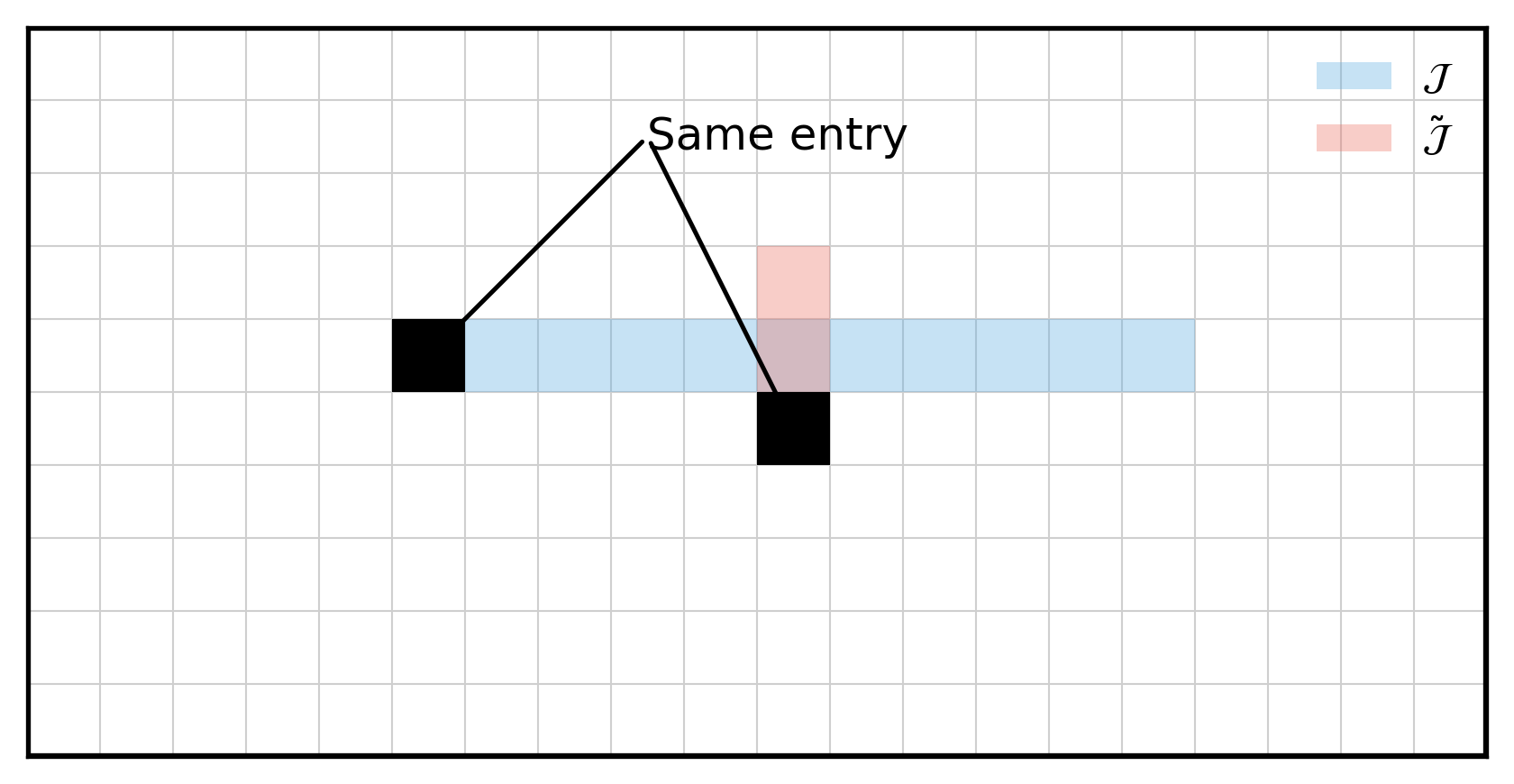} % Adjust width as needed
	\caption{An example with the ground-truth bandwidths $K_{1} = 5, K_{2} = 0$ (in black), and $k_{1} = 0, k_{2} = 1$ (in red). }
	\label{fig:non-nested}
\end{figure}

%\section{Some Linear Algebras}
%For $\A\in\RR^{M\times M}$ and bandwidths $K_{1,1},\cdots,K_{1,M}$, define $\bar\a_i^\top = [\A_{i,u}]_{u:|u-i|\leq K_{1,i}}$ be the entries of row $i$ whose column indices lie within $K_{1,i}$ of $i$. 
%And we define $\calG_1(\A) = \mat{\bar\a_1^\top&\cdots&\bar\a_M^\top}^\top$ be the vector extracting the non-zeros in $\A$. 
%Similarly we define $\calG_2(\B)$ for matrix $\B\in\RR^{N\times N}$ with bandwidths $K_{2,1},\cdots,K_{2,N}$.

\subsection{Extension to the general lag-P model}\label{sec:lagP}
This subsection collects the extension of \model{} to the general lag-\(P\) setting.

\subsubsection{Estimation for the general lag-P model}
We introduce the estimation for the general lag $P$ model. We work on the following model:
\begin{align}\label{supp:LIAR:lagP} 
	\X_t
	=
	\sum_{\i\in\calS}\sum_{p\in[P]}
	\inp{\X_{t-p}}{\M_{p,\i}}\ \e_{i_1}\e_{i_2}^\top
	+\E_t,
\end{align}
For each $\i\in\calS$, we denote $\x_t^{(\i)} := \X_t(\calJ_{\i}), \m_p^{(\i)} = \M_{p,\i}(\calJ_{\i})\in\RR^{|\calJ_{\i}|}$.
Define the stacked coefficient and per-time design vectors
\begin{align*}
	\m_{\i} := \mat{\m_1^{(\i)}\\\vdots\\ \m_P^{(\i)}}, \quad \r_t^{(\i)} := \mat{\x_{t-1}^{(\i)}\\ \vdots\\ \x_{t-P}^{(\i)}}\in\RR^{P|\calJ_{\i}|}
\end{align*}
Then the scalar regression for entry $\i$ reads, for $t=P+1,\dots,T$,
	$[\X_t]_{\i} = \inp{\r_t^{(\i)}}{\m_{\i}} + [\E_t]_{\i}. $
    Stacking over $t$ gives the standard linear model
$
\z_{\i} = \Y_{\i}\m_{\i}+\bnu_{\i},
$
where 
\begin{align}\label{supp:defz}
	\z_{\i} &= \mat{[\X_{P+1}]_{\i}\\ \vdots \\ [\X_T]_{\i}},
	&\bnu_{\i}&= \mat{[\E_{P+1}]_{\i}\\ \vdots \\ [\E_T]_{\i}}\in\RR^{T-P},\notag\\
	\Y_{\i} &= \mat{\r_{P+1}^{(\i)\top}\\ \vdots\\ \r_{T}^{(\i)\top}}
	\in\RR^{(T-P) \times P|\calJ_{\i}|}.
\end{align}
The least square estimator is
\[
\hat\m_{\i}= \arg\min_{\m} \frac{1}{2}\ltwo{\z_{\i} - \Y_{\i}\m}^2
= (\Y_{\i}^\top\Y_{\i})^{-1}\Y_{\i}^\top\z_{\i}.
\]
Notice this procedure can be done in parallel for different $\i\in\calS$.
We summarize the procedure in Algorithm \ref{alg:pls-P}. 
\begin{algorithm}[H]
	\caption{Parallel Least Square for General Lag $P$}
	\begin{algorithmic}\label{alg:pls-P}
		\STATE{\textbf{Input: }Data $\{\X_t\}_{t=1}^T$, local neighborhood $\{\calJ_{\i}\}_{\i\in\calS}$, lag parameter $P$}
		\STATE{\textbf{ParFor} $\i\in\calS$}
		\STATE{\quad Collect the covariate $\Y_{\i} $ and response $\z_{\i}$ defined in \eqref{supp:defz}}
		\STATE{\quad Solve least square $\hat\m_{\i}= (\Y_{\i}^\top\Y_{\i})^{-1}\Y_{\i}^\top\z_{\i}$}	
		\STATE{\textbf{End ParFor}}
		\STATE{\textbf{Output: }$\{\hat\m_{\i}\}_{\i\in\calS}$}
	\end{algorithmic}
\end{algorithm}

\paragraph*{Estimation of Coefficients for \modelone{}. }
For \modelone{}, $\calJ_{\i}$ are rectangles by construction, and we denote $\M_{p}^{[\i]}:=\M_{p,\i}[\calJ_{\i}]$ as a sub-matrix of $\M_{p,\i}$. 
Arrange these matrices into the block matrix
\begin{align*}
	\M_{p}^{\blk} = \mat{\M_{p}^{[1,1]}&\cdots&\M_{p}^{[1,N]}\\ \vdots&&\vdots \\ \M_{p}^{[M,1]}&\cdots &\M_{p}^{[M,N]}}. 
	%	= \underbrace{\mat{\A_{p,1}\\\vdots \\\A_{p,M}}}_{=:\A_p} \underbrace{\mat{\B_{p,1}\\\vdots \\\B_{p,N}}^\top}_{=:\B_p^\top}.
\end{align*}
Then under \modelone{}, $\M_{p}^{\blk}$ has rank at most $R$. 
This motivates us to consider the best rank $R$ approximation of the following estimator of $\M_{p}^{\blk}$, where $\hat\m_{\i} = \mat{\hat\m_{1}^{(\i)\top}& \cdots& \hat\m_{P}^{(\i)\top}}^\top$:
\begin{align*}
	\hat\M_{p}^{\blk} = \mat{\textsf{mat}(\hat\m_{p}^{(1,1)})&\cdots&\textsf{mat}(\hat\m_{p}^{(1,N)})\\ \vdots&&\vdots \\ \textsf{mat}(\hat\m_{p}^{(M,1)})&\cdots &\textsf{mat}(\hat\m_{p}^{(M,N)})},
\end{align*}
where $\textsf{mat}$ is the inverse operation of $\vec$. 
Our algorithm is summarized in Algorithm \ref{supp:alg:spliar}.
\begin{algorithm}[H]
	\caption{Parallel Least Square for \modelone}
	\begin{algorithmic}\label{supp:alg:spliar}
		\STATE{\textbf{Input: }Local neighborhood $\calJ_{\i}$, data $\{\X_t\}_{t=1}^T$, lag parameter $P$, rank $R$}
		\STATE{Run Algorithm \ref{alg:pls-P} for $\{\hat\m_{\i}\}_{\i\in\calS}$}
		\STATE{Perform rank $R$ SVD approximation on $\hat\M_{p}^{\blk}$ and get:
			$\hat\M_{p,R}^{\blk}=\svd_R(\hat\M_{p}^{\blk})$}
		%			 $\hat\M_{p,R}^{\blk}=\sum_{r=1}^R\hat\u_{p,r}\hat\v_{p,r}^\top$}
	\STATE{\textbf{Output: }$\{\hat\M_{p,R}^{(\i)}\}_{p\in[P], \i\in\calS}$, where $\hat\M_{p,R}^{(\i)}$ be the $(i_1,i_2)$-th block of $\hat\M_{p,R}^{\blk}$}
\end{algorithmic}
\end{algorithm}

% \paragraph*{Identifiability Issue. }
% In the case where $R=1$ \citep{chen2021autoregressive}, it is assumed that one of the coefficient matrices has a unit Frobenius norm. They also mentioned the identifiability issue becomes subtler when $R>1$. In our formulation of \modelone{}, this concern disappears as 
% the primary quantity of interest is the equivalence class rather than two individual components.

\subsubsection{Bandwidth selection for the general lag-P model}
We also restrict the search to a nested sequence indexed by $k$ for each location $\i = (i_1,i_2)$:
$$\calJ_{\i}[1]\subsetneq\cdots\subsetneq\calJ_{\i}[k_0]\subsetneq\cdots\subsetneq\calJ_{\i}[K_0]\subset\calS,$$
where $\calJ_{\i}[k_0] = \calJ_{\i}$, and $K_0$ is a prescribed cap. 
Write $s_k = |\calJ_{\i}[k]|$ (depends on $k$, and implicitly on $\i$).
For any level $k$, let $\x_t[k]:= \X_t\big(\calJ_{\i}[k]\big)\in\RR^{s_k}$ and stack $P$ lags into one vector $\r_t[k] = \mat{\x_t[k]^\top&\cdots&\x_{t-P+1}[k]^\top}^\top\in\RR^{Ps_k}$. Then we form the design matrix $\Y[k]\in\RR^{(T-P)\times Ps_k}$:
\begin{align*}
\Y[k]= \mat{\r_{P+1}[k] &\cdots &\r_T[k]}^\top.
\end{align*}
Also recall $\z = \mat{[\X_{P+1}]_{\i}& \cdots & [\X_T]_{\i}}^\top$. 
The residual sum of squares is
\begin{align*}
\rss_{\i}(k) := \ltwo{\z - \Y[k](\Y[k]^\top\Y[k])^{-1}\Y[k]^\top\z}^2.
\end{align*}
And we use the Bayesian information criterion for the bandwidth selection. 
\begin{align}\label{supp:BIC-entrywise}
\bic_{\i}(k)
=
\log \rss_{\i}(k)
+
D_0\,\frac{|\calJ_{\i}[k]|\cdot P}{T}\,\log(M\vee N\vee T).
\end{align}
For each $\i\in\calS$, we select $\hat k_{\i}$ by
$\hat k_{\i}=\arg\min_{k\le K_0}\ \bic_{\i}(k).$

\subsection{Additional Simulation Results}

\subsubsection{Visualization of the CLT}
To illustrate the asymptotic normality result in Theorem~1 of the main paper, we conduct an oracle diagnostic under the same data-generating mechanism as the estimation-accuracy experiment in Section~5. Specifically, we consider the lag-one \model{} model with \(P=1\), \((M,N)=(10,10)\), and a square neighborhood with radius \(K=3\). The local coefficients are generated as in Section~5 and the resulting transition matrix is rescaled to have spectral radius smaller than one. The innovations are Gaussian with the same noise level as in Section~5, and the first \(500\) observations are discarded as burn-in.

We focus on the interior location \(\i_0=(5,5)\). Its true neighborhood \(\calJ_{\i_0}\) is not truncated, so \(|\calJ_{\i_0}|=(2K+1)^2=49\). Let \(j_0\) denote the coordinate of \(\m_{\i_0}\) corresponding to the self-lag entry at \(\i_0\). Under the column-major ordering used for \(\calJ_{\i_0}\), this coordinate is \(j_0=25\) in one-based indexing.

For each sample size \(T\in\{500,1000,2000,5000\}\), we generate \(B=500\) independent replications. In each replication, we estimate \(\m_{\i_0}\) by least squares using the oracle true neighborhood \(\calJ_{\i_0}\), rather than a BIC-selected neighborhood. This construction isolates the least-squares CLT from the additional effect of neighborhood selection. We then compute the studentized statistic
\[
Z_{\i_0,j_0}
=
\frac{\hat\m_{\i_0,j_0}-\m_{\i_0,j_0}}
{\hat\sigma_{\i_0}\left\{\left[(\Y_{\i_0}^\top\Y_{\i_0})^{-1}\right]_{j_0,j_0}\right\}^{1/2}},
\]
where
\[
\hat\sigma_{\i_0}^2
=
\frac{\ltwo{\z_{\i_0}-\Y_{\i_0}\hat\m_{\i_0}}^2}
{T-1-|\calJ_{\i_0}|}.
\]
Figure~\ref{fig:clt-visualization} compares the empirical distributions of \(Z_{\i_0,j_0}\) with the standard normal benchmark across sample sizes. The distributions are centered around zero and become increasingly close to the standard normal density as \(T\) increases.

\begin{figure}[htbp]
\centering
\includegraphics[width=0.95\textwidth]{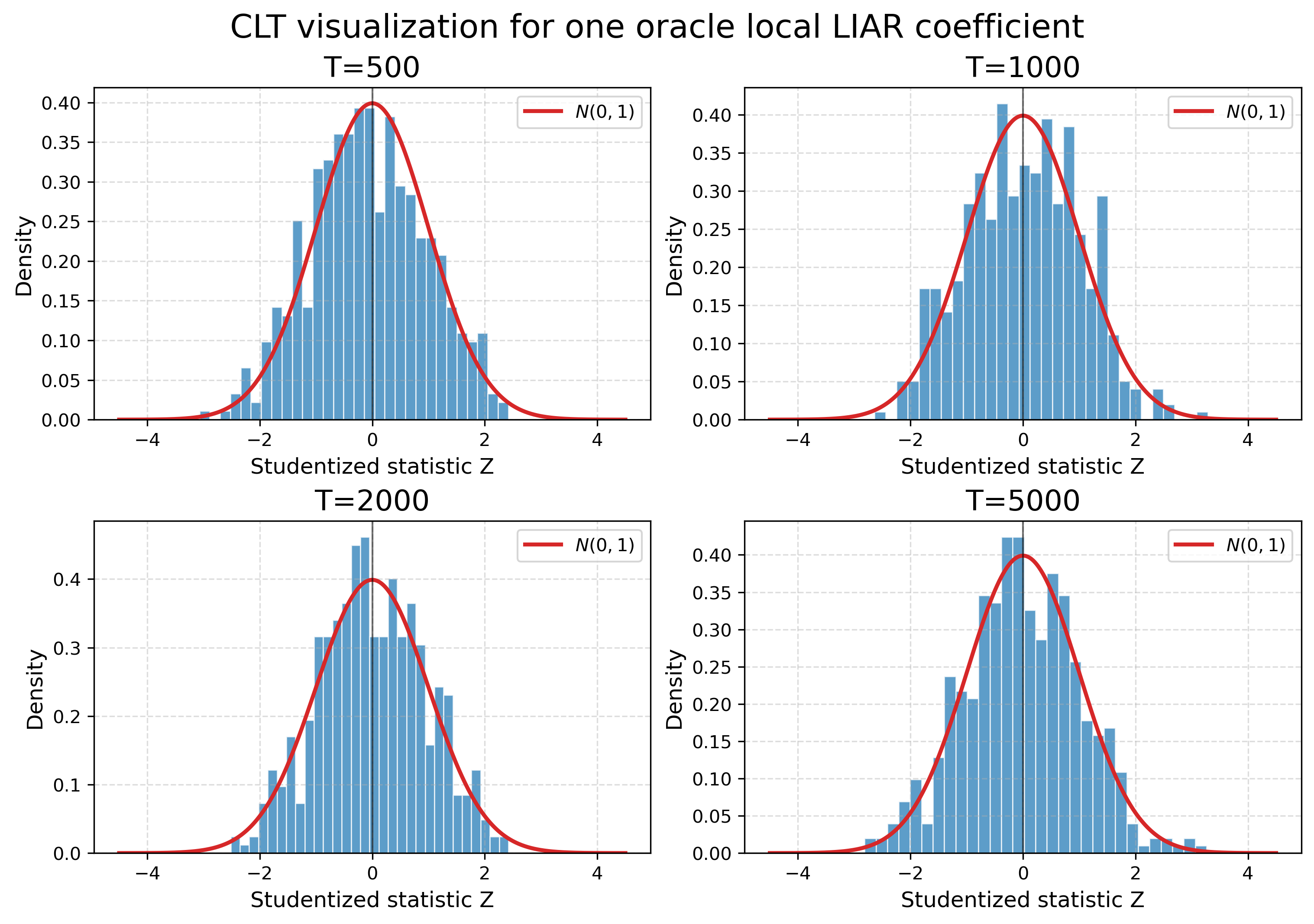}
\caption{Visualization of the CLT for the representative coefficient \(\m_{\i_0,j_0}\). The empirical distributions of \(Z_{\i_0,j_0}\) are compared with the standard normal benchmark across sample sizes.}
\label{fig:clt-visualization}
\end{figure}

\subsubsection{Additional numerical results for neighborhood selection}
After the oracle CLT diagnostic above, we next consider a larger-sample version of the BIC selection experiment in Section~5. The data-generating mechanism is unchanged: the true neighborhood radius is \(K=3\), and BIC selects over \(k\in\{0,1,2,3,4,5\}\). Only the sample size is increased to \(T\in\{4000,5000,6000\}\).

Figure~\ref{fig:bic-large-sample} reports the per-pixel success rates for \(M=N\in\{10,15,20\}\). Compared with the moderate-sample results in the main paper, the success rates further improve as \(T\) increases. These results are consistent with the neighborhood-selection consistency result in Theorem~5.

\begin{figure}[htbp]
    \centering
    \includegraphics[width=0.95\textwidth]{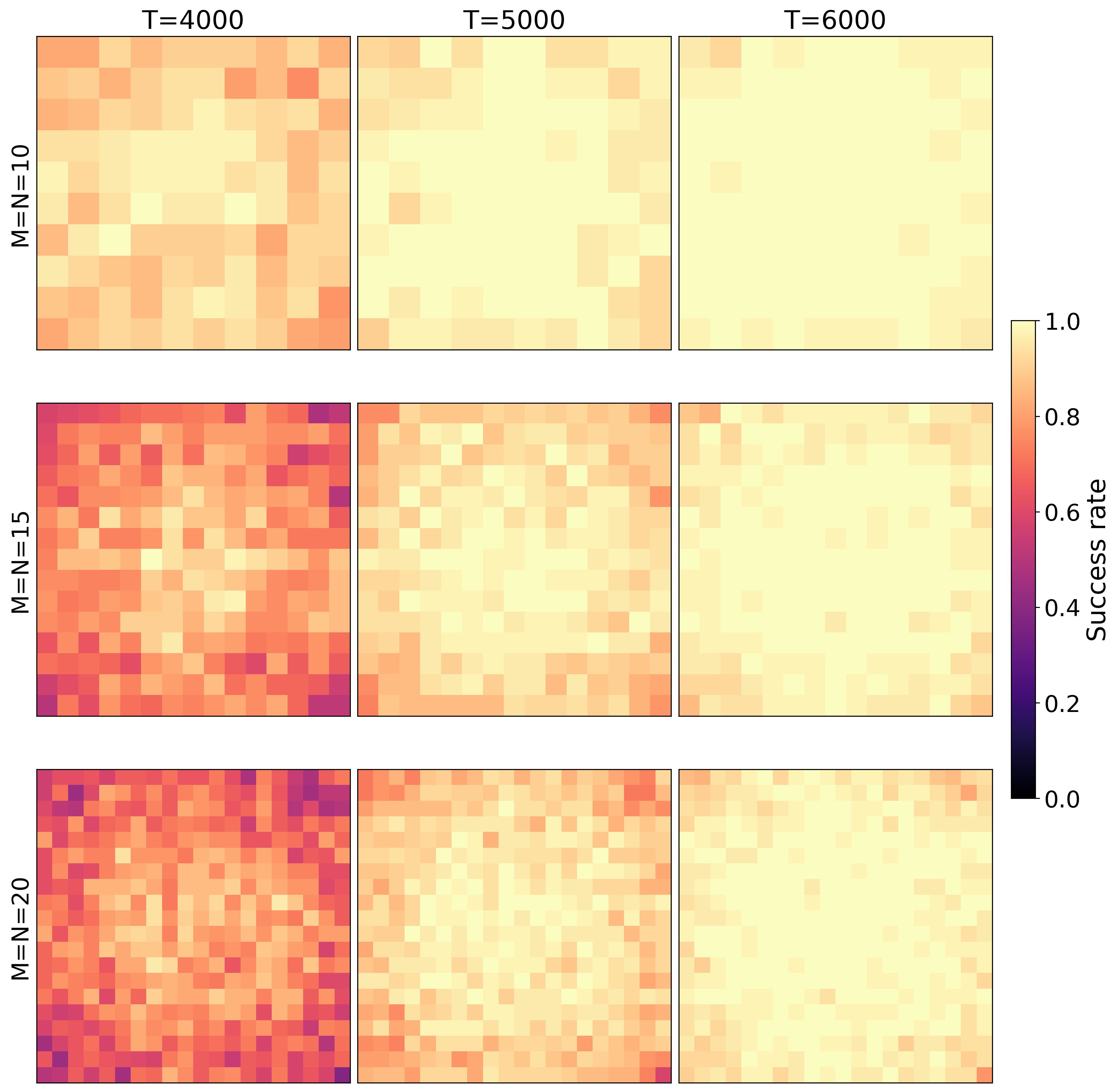}
    \caption{Additional neighborhood-selection results for
    \(M=N\in\{10,15,20\}\) and \(T\in\{4000,5000,6000\}\).}
    \label{fig:bic-large-sample}
\end{figure}

\section{Tensor Local Interaction Autoregression and TEC Application}\label{sec:supp-tensor}
This section collects the tensor extension of \model{} and the tensor TEC application. We first formulate the tensor local interaction autoregression and its entrywise least-squares estimator, and then apply the same construction to TEC data organized by two spatial modes and one within-hour temporal mode.

\subsection{Tensor \model{} formulation and estimation}
We consider a local interaction autoregressive model for tensor-valued time series. The modeling idea mirrors the matrix case: each tensor entry evolves from past values in a local neighborhood, while the neighborhood itself may vary across entries and across tensor modes.

We use calligraphic-font bold-face letters, such as \(\bcalM\) and \(\bcalX\), to denote tensors. An order-\(d\) tensor is a \(d\)-way array; for example, \(\bcalX\in\RR^{N_1\times \cdots \times N_d}\) has size \(N_j\) along its \(j\)th mode. Let \(\calS_d=[N_1]\times\cdots\times[N_d]\) be the tensor index set. For a tensor time series \(\{\bcalX_t\}_{t\geq0}\subset\RR^{\calS_d}\), the tensor \model{} takes the form
\begin{align*}
\bcalX_t
=
\sum_{\i\in\calS_d}
\inp{\bcalX_{t-1}}{\bcalM_{\i}}\cdot
\e_{i_1}\circ\cdots\circ\e_{i_d}
+\bcalE_t,
\end{align*}
where \(\i=(i_1,\ldots,i_d)\in\calS_d\), and \(\bcalM_{\i}\) is supported on a local neighborhood \(\calJ_{\i}\) of \(\i\). For example, a rectangular neighborhood with radius vector \((k_{1,\i},\ldots,k_{d,\i})\) contains indices whose \(\ell\)th coordinate lies within \(k_{\ell,\i}\) of \(i_\ell\), truncated at the tensor boundary.

Vectorizing both sides gives
\begin{align*}
\x_t
=
\underbrace{
\sum_{\i\in\calS_d}
(\e_{i_d}\otimes\cdots\otimes\e_{i_1})\vec(\bcalM_{\i})^\top
}_{=: \M}
\x_{t-1}
+\bepsilon_t,
\end{align*}
where \(\x_t=\vec(\bcalX_t)\) and \(\bepsilon_t=\vec(\bcalE_t)\). Denote the local covariate vector and coefficient vector by
\[
\x_t^{(\i)}=\bcalX_t(\calJ_{\i})\in\RR^{|\calJ_{\i}|},
\qquad
\m_{\i}=\bcalM_{\i}(\calJ_{\i})\in\RR^{|\calJ_{\i}|}.
\]
Then, for each entry \(\i\) and \(t=2,\ldots,T\), the tensor model reduces to the scalar regression
\[
[\bcalX_t]_{\i}
=
\inp{\x_{t-1}^{(\i)}}{\m_{\i}}
+[\bcalE_t]_{\i}.
\]
Stacking over time gives the standard linear model
\[
\z_{\i}=\Y_{\i}\m_{\i}+\bnu_{\i},
\]
where
\begin{equation*}
\begin{split}
&\z_{\i}=\mat{[\bcalX_2]_{\i}&\cdots&[\bcalX_T]_{\i}}^\top,
\qquad
\bnu_{\i}=\mat{[\bcalE_2]_{\i}&\cdots&[\bcalE_T]_{\i}}^\top\in\RR^{T-1},\\
&\Y_{\i}=\mat{\x_1^{(\i)}&\cdots&\x_{T-1}^{(\i)}}^\top
\in\RR^{(T-1)\times |\calJ_{\i}|}.
\end{split}
\end{equation*}
The least-squares estimator is
\[
\hat\m_{\i}
=
\arg\min_{\m}\frac{1}{2}\ltwo{\z_{\i}-\Y_{\i}\m}^2
=
(\Y_{\i}^\top\Y_{\i})^{-1}\Y_{\i}^\top\z_{\i}.
\]
Algorithm~\ref{alg:tensor} summarizes the parallel implementation.
\begin{algorithm}[H]
\caption{Parallel least squares for tensor \model{}}
\label{alg:tensor}
\begin{algorithmic}
\STATE{\textbf{Input: } Local neighborhoods \(\{\calJ_{\i}\}_{\i\in\calS_d}\), data \(\{\bcalX_t\}_{t=1}^T\)}
\STATE{\textbf{ParFor} \(\i\in\calS_d\)}
\STATE{\quad Collect the covariate \(\Y_{\i}\) and response \(\z_{\i}\)}
\STATE{\quad Solve least squares \(\hat\m_{\i}=(\Y_{\i}^\top\Y_{\i})^{-1}\Y_{\i}^\top\z_{\i}\)}
\STATE{\textbf{End ParFor}}
\STATE{\textbf{Output: } \(\{\hat\m_{\i}\}_{\i\in\calS_d}\)}
\end{algorithmic}
\end{algorithm}
Theoretical results for coefficient estimation, auto-covariance estimation, and neighborhood selection extend to this tensor setting by the same proof strategy as in the matrix case. The framework also allows tensor autoregressions with no locality along some modes, by choosing \(\calJ_{\i}\) to span the full range along those dimensions while remaining local along others.

\subsection{Tensor TEC application}
The TEC application uses the tensor formulation above to separate coarse spatial variation from within-hour temporal variation. Starting from the same GNSS-derived TEC product \citep{sun2023complete}, we first aggregate the spatial grid to \(4^\circ\times4^\circ\), yielding a \(46\times91\) spatial grid. We then stack, for each clock hour, the twelve 5-minute frames into a third mode, which represents the within-hour slot. This produces hourly tensors of size \(46\times91\times12\). Indexing these tensors by hour gives a tensor time series of shape \(46\times91\times12\times720\) for the 30-day period in June 2019.

We use the first \(80\%\) of hours for training and the remaining \(20\%\) for testing. Test RMSE is computed from one-step-ahead full-tensor predictions over the test period, averaging squared prediction errors over all test hours and tensor entries. Neighborhood sizes are selected entrywise by BIC over the candidate set
\[
\{(0,0,0),(0,0,1),(0,1,1),(1,1,1),(1,1,2),(1,2,2),(2,2,2)\}.
\]
Figure~\ref{fig:tensor-k-freq} reports the selection frequencies across all locations and within-hour slots. The selections concentrate on small neighborhoods, with \((0,1,1)\) most common and \((0,0,1)\) the primary alternative. This concentration suggests that the hourly TEC tensor dynamics are driven mainly by short-range interactions along the spatial and within-hour dimensions, rather than by dense full-tensor dependence.

To examine when and where these selected neighborhoods arise, Figure~\ref{fig:tensor-slot-prop} shows the proportion of each selected neighborhood versus within-hour slot, from 5 to 60 minutes. The slot-wise proportions are fairly stable across the hour, with only mild changes near hour boundaries. This stability suggests that the learned local dependence pattern is persistent within an hour, with only modest changes near transitions between adjacent hourly blocks. Figure~\ref{fig:tensor-mode-map} displays the spatial mode map, defined as the most frequently selected neighborhood per pixel over the 12 within-hour slots.

\begin{figure}[H]
\centering
\begin{subfigure}[t]{0.48\textwidth}
\centering
\includegraphics[width=0.9\linewidth]{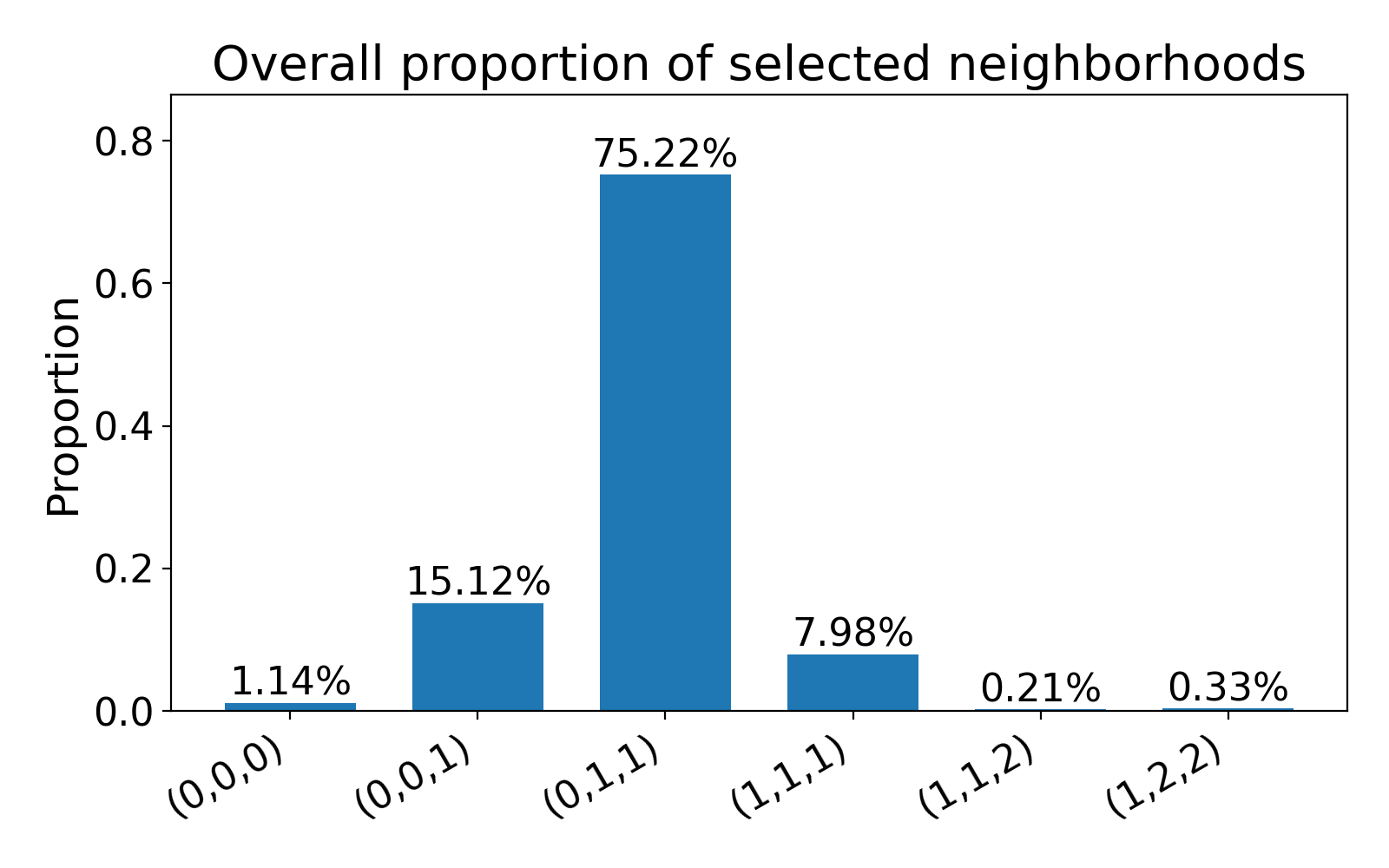}
\caption{Selection frequencies of neighborhoods \((k_1,k_2,k_3)\) over all locations and slots.}
\label{fig:tensor-k-freq}
\end{subfigure}\hfill
\begin{subfigure}[t]{0.48\textwidth}
\centering
\includegraphics[width=0.9\linewidth]{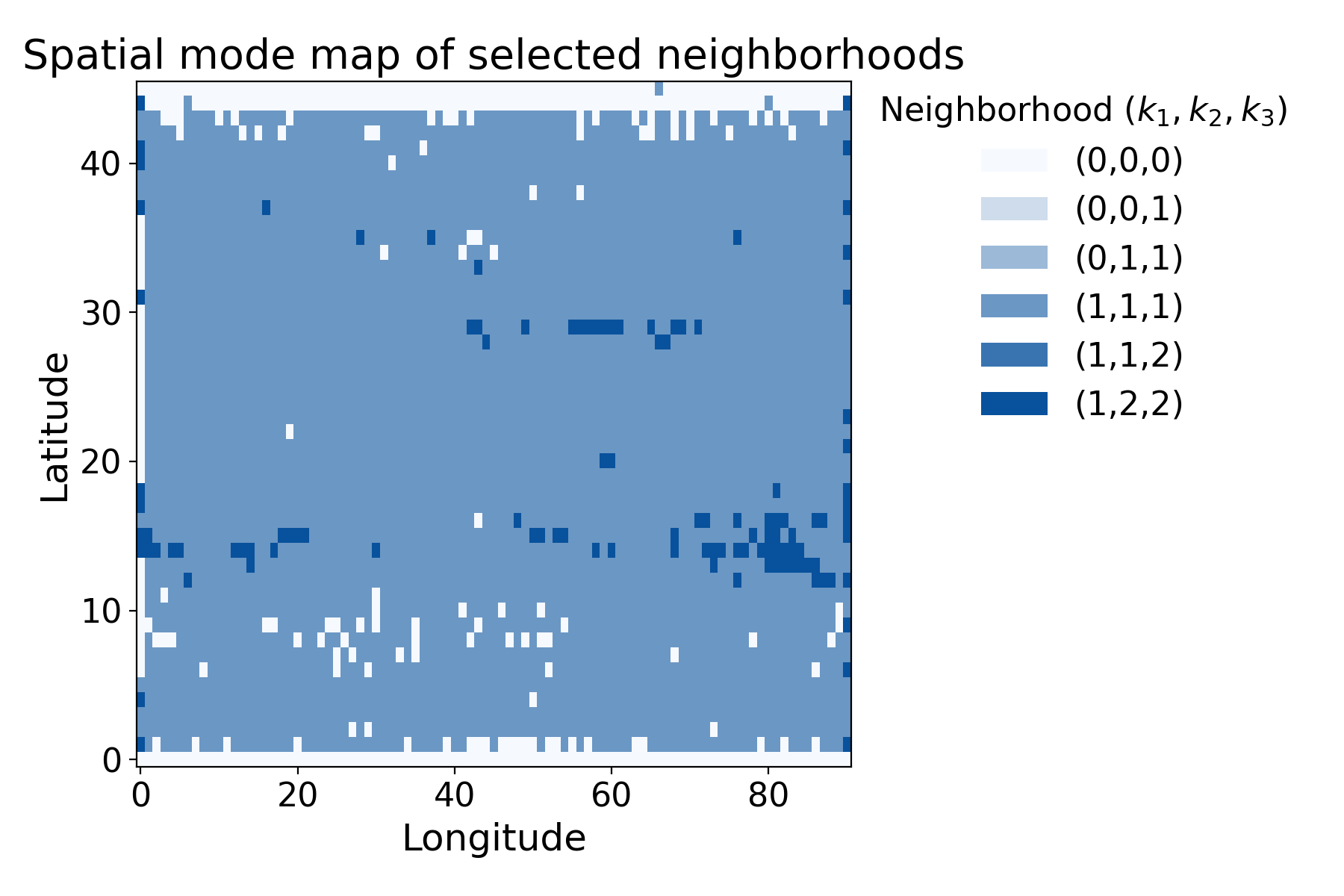}
\caption{Spatial mode map: per pixel, the most frequently selected neighborhood over 12 slots.}
\label{fig:tensor-mode-map}
\end{subfigure}

\vspace{0.8em}

\begin{subfigure}[t]{0.98\textwidth}
\centering
\includegraphics[width=0.82\linewidth]{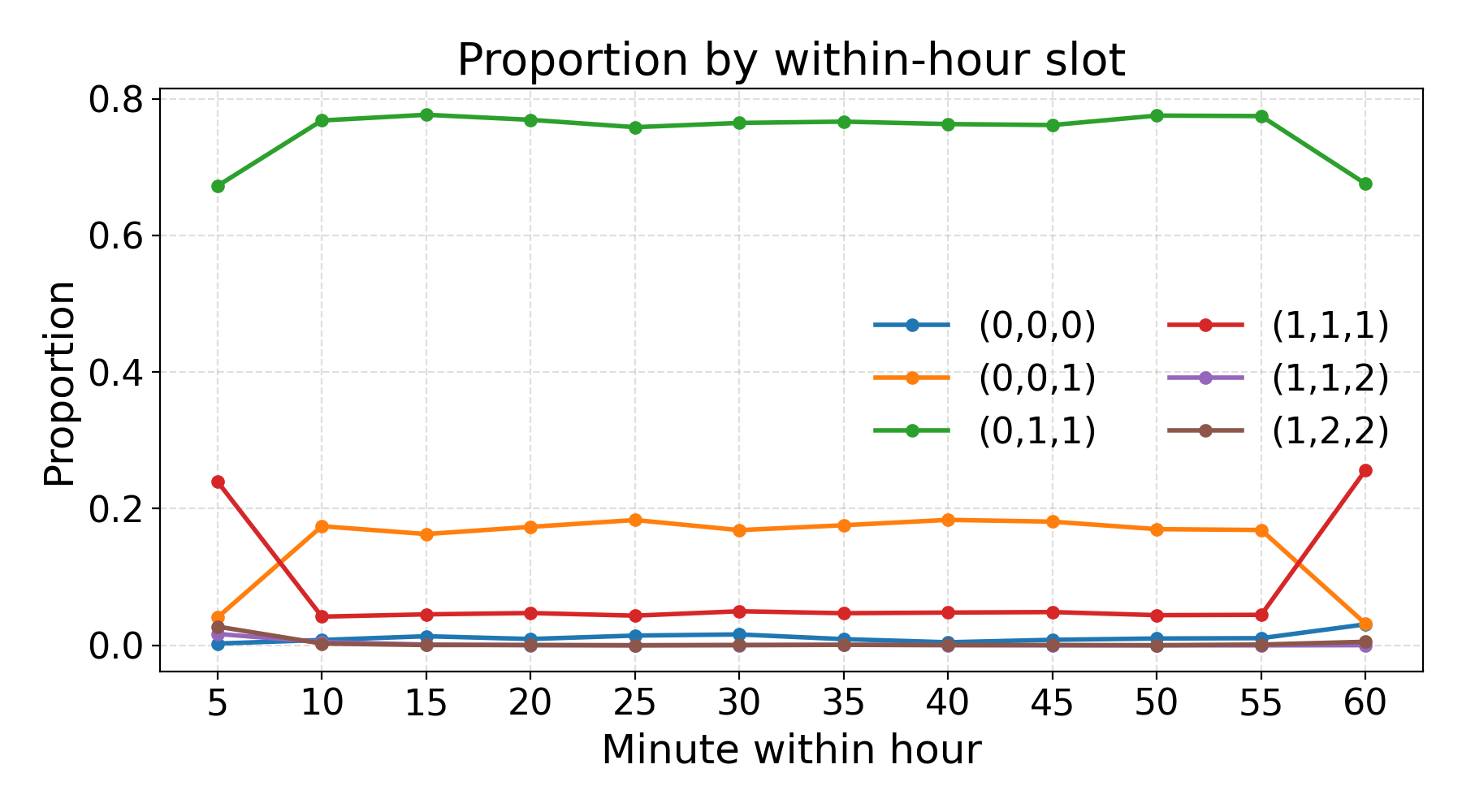}
\caption{Proportion of each neighborhood versus within-hour slot, from 5 to 60 minutes.}
\label{fig:tensor-slot-prop}
\end{subfigure}

\caption{Neighborhood selection in the tensor TEC series using entrywise BIC. Top: (a) overall selection frequencies and (b) spatial mode map. Bottom: (c) proportions by within-hour slot.}
\label{fig:tensor-neighborhoods}
\end{figure}

As a benchmark, we fit the pixel-wise autoregression \textsf{LIAR-P}. Its test \(\mathrm{RMSE}\) is \(1.29\), whereas \model{} with BIC-selected neighborhoods attains a lower test \(\mathrm{RMSE}\) of \(1.12\). Thus, the tensor neighborhoods selected by BIC capture spatio-temporal interactions that are useful for one-step-ahead full-tensor prediction.

\end{document}